\documentclass[dvipsnames,11pt]{article}
\usepackage{amsmath,amssymb,amsthm,color,bm,mathrsfs,extarrows,tikz,graphicx,mathtools,enumitem,fancybox,makecell}
\usepackage[square,sort,comma,numbers]{natbib}
\usepackage{subcaption}
\usepackage[ruled]{algorithm2e}
\usepackage[margin=1in]{geometry}
\usepackage{xcolor,bbm}
\usepackage[pagebackref]{hyperref}
\usepackage{comment}
\usepackage{soul}
\usepackage{thm-restate}
\usepackage{ifthen}
\usepackage{mathdots}
\usepackage{enumitem}

\allowdisplaybreaks

\setlist[itemize]{itemsep=0pt}
\setlist[enumerate]{itemsep=0pt}
\hypersetup{colorlinks=true,urlcolor=blue,linkcolor=blue,citecolor=[rgb]{.42,.56,.14},}
\usetikzlibrary{decorations.pathreplacing,decorations.markings,decorations.pathmorphing,decorations.shapes,arrows.meta,positioning,patterns}

\usepackage[capitalise,nameinlink]{cleveref}
\Crefname{lemma}{Lemma}{Lemmas}
\Crefname{fact}{Fact}{Facts}
\Crefname{theorem}{Theorem}{Theorems}
\Crefname{corollary}{Corollary}{Corollaries}
\Crefname{claim}{Claim}{Claims}
\Crefname{example}{Example}{Examples}
\Crefname{problem}{Problem}{Problems}
\Crefname{definition}{Definition}{Definitions}
\Crefname{notation}{Notation}{Notations}
\Crefname{assumption}{Assumption}{Assumptions}
\Crefname{subsection}{Subsection}{Subsections}
\Crefname{section}{Section}{Sections}
\Crefformat{equation}{(#2#1#3)}

\newtheorem{theorem}{Theorem}[section]
\newtheorem*{theorem*}{Theorem}

\newtheorem{proposition}[theorem]{Proposition}
\newtheorem*{proposition*}{Proposition}

\newtheorem*{property*}{Property}
\newtheorem{lemma}[theorem]{Lemma}
\newtheorem*{lemma*}{Lemma}

\newtheorem*{corollary*}{Corollary}
\newtheorem*{conjecture*}{Conjecture}
\newtheorem{fact}[theorem]{Fact}
\newtheorem*{fact*}{Fact}

\newtheorem*{exercise*}{Exercise}

\newtheorem*{hypothesis*}{Hypothesis}
\newtheorem{conjecture}[theorem]{Conjecture}

\theoremstyle{definition}
\newtheorem{definition}[theorem]{Definition}

\newtheorem{example}[theorem]{Example}

\newtheorem{exercise-easy}[theorem]{Exercise}
\newtheorem{exercise-med}[theorem]{Exercise}
\newtheorem{exercise-hard}[theorem]{Exercise$^\star$}
\newtheorem{claim}[theorem]{Claim}
\newtheorem*{claim*}{Claim}

\newtheorem{remark}[theorem]{Remark}
\newtheorem*{remark*}{Remark}

\newtheorem*{observation*}{Observation}

\DeclareSymbolFont{extraup}{U}{zavm}{m}{n}
\DeclareMathSymbol{\varheart}{\mathalpha}{extraup}{86}
\DeclareMathSymbol{\vardiamond}{\mathalpha}{extraup}{87}

\DeclareMathOperator*{\E}{\mathbb E}
\DeclareMathOperator*{\Var}{\mathrm{Var}}
\DeclareMathOperator*{\Cov}{\mathrm{Cov}}

\renewcommand{\Pr}{\operatorname*{\mathbf{Pr}}}

\newcommand{\Mod}[1]{\ (\mathrm{mod}\ #1)}

\newcommand{\eps}{\varepsilon}
\newcommand{\abs}[1]{\left| #1 \right|}
\newcommand{\vabs}[1]{\left\| #1 \right\|}

\newcommand{\pbra}[1]{\left( #1 \right)}
\newcommand{\sbra}[1]{\left[ #1 \right]}
\newcommand{\cbra}[1]{\left\{ #1 \right\}}
\newcommand{\floorbra}[1]{\left\lfloor #1 \right\rfloor}

\renewcommand{\mid}{\,\middle\vert\,}

\newcommand{\Bin}{\mathsf{Bin}}
\newcommand{\bin}{\{0,1\}}

\newcommand{\poly}{\mathsf{poly}}
\newcommand{\polylog}{\mathsf{polylog}}

\newcommand{\err}{\mathsf{err}}
\newcommand{\supp}[1]{\mathsf{supp}\pbra{#1}}

\newcommand{\ac}{\mathsf{AC^0}}
\newcommand{\nc}{\mathsf{NC^0}}
\newcommand{\sym}{\mathrm{sym}}

\newcommand{\Deven}{\mathtt{evens}}
\newcommand{\Dodd}{\mathtt{odds}}
\newcommand{\Dall}{\mathtt{all}}

\newcommand{\Fbb}{\mathbb{F}}
\newcommand{\Nbb}{\mathbb{N}}
\newcommand{\Rbb}{\mathbb{R}}
\newcommand{\Zbb}{\mathbb{Z}}

\newcommand{\Dcal}{\mathcal{D}}
\newcommand{\Ecal}{\mathcal{E}}
\newcommand{\Fcal}{\mathcal{F}}
\newcommand{\Gcal}{\mathcal{G}}
\newcommand{\Hcal}{\mathcal{H}}
\newcommand{\Ical}{\mathcal{I}}

\newcommand{\Mcal}{\mathcal{M}}

\newcommand{\Pcal}{\mathcal{P}}
\newcommand{\Qcal}{\mathcal{Q}}

\newcommand{\Ucal}{\mathcal{U}}

\newcommand{\Wcal}{\mathcal{W}}

\newcommand{\tvdist}[1]{\vabs{#1}_\mathsf{TV}}

\renewcommand{\bar}{\overline}

\title{Symmetric Distributions from Shallow Circuits}
\author{
Daniel M. Kane\thanks{University of California, San Diego. Email: \texttt{dakane@ucsd.edu}. Supported by NSF Medium Award CCF-2107547 and NSF CAREER Award CCF-1553288.}
\and
Anthony Ostuni\thanks{University of California, San Diego. Email: \texttt{aostuni@ucsd.edu}.}
\and
Kewen Wu\thanks{Caltech. Email: \texttt{shlw\_kevin@hotmail.com}.}
}
\date{}

\begin{document}

\maketitle

\begin{abstract}
    We characterize the symmetric distributions that can be (approximately) generated by shallow Boolean circuits.
    More precisely, let $f\colon \{0,1\}^m \to \{0,1\}^n$ be a Boolean function where each output bit depends on at most $d$ input bits.
    Suppose the output distribution of $f$ evaluated on uniformly random input bits is close in total variation distance to a symmetric distribution $\mathcal{D}$ over $\{0,1\}^n$.
    Then $\mathcal{D}$ must be close to a mixture of the uniform distribution over $n$-bit strings of even Hamming weight, the uniform distribution over $n$-bit strings of odd Hamming weight, and $\gamma$-biased product distributions for $\gamma$ an integer multiple of $2^{-d}$.
    Moreover, the mixing weights are determined by low-degree, sparse $\mathbb{F}_2$-polynomials.
    This extends the previous classification for generating symmetric distributions that are also uniform over their support.
\end{abstract}

\newpage
\setcounter{tocdepth}{2}
\tableofcontents
\newpage

\section{Introduction}\label{sec:intro}

One of the most celebrated results in complexity theory is that $\ac$ circuits\footnote{Recall that these are (families of) Boolean circuits of constant depth and unbounded fan-in gates. They may be contrasted with $\nc$ circuits, which are also of constant depth, but have bounded fan-in gates.} require an exponential number of gates to compute the parity function \cite{FSS84, Ajt83, yao1985separating, hastad1986almost, haastad1986computational}.
Surprisingly, the weaker circuit class $\nc$ suffices to perform a very similar task.
In particular, mapping uniformly random bits $(x_1, x_2, \ldots, x_n)$ to $(x_1 \oplus x_2, x_2 \oplus x_3, \ldots, x_{n-1} \oplus x_n, x_n \oplus x_1)$ produces the uniform distribution over $n$-bit strings of even parity, or equivalently, over input-output pairs $(x, \textrm{PARITY}(x))$ \cite{babai1987random, boppana1987one}.
This observation begs the question: what computational resources are required to (approximately) \emph{generate} specific distributions, as opposed to the traditional task of \emph{computing} specific functions?

A fundamental question in its own right, the complexity of sampling from distributions also has numerous applications.
For example, results on the hardness of sampling can be translated to data structure lower bounds \cite{viola2012complexity, lovett2011bounded, beck2012large, viola2020sampling, chattopadhyay2022space, viola2023new, yu2024sampling, kane2024locality, alekseev2025sampling}, provide input-independent quantum-classical separations \cite{watts2023unconditional, viola2023new, kane2024locality, grier2025quantum}, and are key components to the construction of explicit codes \cite{shaltiel2024explicit}.
Moreover, techniques and intuition developed in this setting have successfully been applied to pseudorandom generators \cite{viola2012complexity, lovett2011bounded, beck2012large} and extractors \cite{viola2012extractors, de2012extractors, viola2014extractors, chattopadhyay2016explicit, cohen2016extractors}.

While the general problem was considered in early work (see, e.g., \cite{jerrum1986random}), a focus on generating distributions via shallow circuits was advocated for more recently by Viola \cite{viola2012complexity}.
Since then, the field has seen a number of exciting developments (see, e.g., the recent works \cite{filmus2023sampling, viola2023new, kane2024locality, shaltiel2024explicit, kane2024locally2, alekseev2025sampling} and references therein). 
One notable takeaway from prior works is that $\nc$ circuits can sample very few uniform symmetric distributions (i.e., uniform distributions over a symmetric support), even allowing for small errors.
In particular, the line of work \cite{viola2012complexity, filmus2023sampling, viola2023new, kane2024locality, kane2024locally2} recently culminated in the following classification result, which confirmed a conjecture of Filmus, Leigh, Riazanov, and Sokolov \cite{filmus2023sampling}.
For a function $f\colon\bin^m\to\bin^n$, let $f(\Ucal^m)$ be the distribution resulting from applying $f$ to $x \sim \Ucal^m$, the uniform distribution over $\bin^m$.

\begin{theorem}[{\cite{kane2024locally2}}]\label{thm:nc0_classification}
Let $\eps\in[0,1]$ be arbitrary.
Assume $f\colon\bin^m\to\bin^n$ is computable by an $\nc$ circuit of constant depth and $f(\Ucal^m)$ is $\eps$-close in total variation distance to a uniform symmetric distribution where $n$ is sufficiently large.
Then $f(\Ucal^m)$ is $O(\eps)$-close to one of the following six special uniform symmetric distributions:
    \begin{itemize}
        \item Point distribution on $0^n$,
        \item Point distribution on $1^n$,
        \item Uniform distribution over $\cbra{0^n,1^n}$,
        \item Uniform distribution over $n$-bit strings with even Hamming weights,
        \item Uniform distribution over $n$-bit strings with odd Hamming weights,
        \item Uniform distribution over all $n$-bit strings.
    \end{itemize}
\end{theorem}

We emphasize that \Cref{thm:nc0_classification} works with \emph{uniform} symmetric distributions, which does not capture, for example, the $(1/4)$-biased product distribution that is symmetric and easily sampleable by $\nc$ circuits.

\subsection{Our Result}

In this paper, we take another step toward understanding the sampleability of distributions by extending \Cref{thm:nc0_classification} to handle \emph{arbitrary} symmetric distributions.
Let the \emph{locality} of a Boolean function be the largest number of input bits that any output bit depends on.
Note that $\nc$ is precisely the class of sequences of functions with bounded locality.

\begin{theorem}[Informal version of \Cref{thm:main}]\label{thm:special_main}
    Let $d\ge0$ be an integer.
    For any $\eps\in(0,1]$, there exists some $\delta$ such that $\delta\to0$ as $\eps\to0$ and the following holds.

    Suppose $f\colon\bin^m\to\bin^n$ is a $d$-local function and $f(\Ucal^m)$ is $\eps$-close in total variation distance to a symmetric distribution where $n$ is sufficiently large in terms of $d$ and $\eps$.
    Then $f(\Ucal^m)$ is $\delta$-close to some mixture of
    \begin{enumerate}
        \item The uniform distribution over $n$-bit strings of even Hamming weight,

        \item The uniform distribution over $n$-bit strings of odd Hamming weight, and

        \item $\gamma$-biased product distributions on $n$ bits for $\gamma$ an integer multiple of $2^{-d}$.
    \end{enumerate}
    Moreover, the mixing weights are determined by degree-$O_d(1)$ $\Fbb_2$-polynomials with $O_d(n)$ monomials.
\end{theorem}
The ``moreover'' conclusion implies this mixture can be exactly produced by $O_d(1)$-local functions (see \Cref{rmk:NC0_can_sample}).

\paragraph{Learning Structured Distributions.}
As mentioned above, the study of sampling is rife with applications to other areas of theoretical computer science.
Continuing this trend, we highlight one consequence of \Cref{thm:special_main} to learning theory.

The reconstruction of an unknown probability density function based on observed data is a fundamental problem in both statistics and computer science. 
The typical setting is the PAC-learning model \cite{valiant1984theory,blumer1989learnability,kearns1994learnability}: given access to independent samples of an unknown distribution $\Dcal$, the goal is to output a hypothesis distribution $\Dcal'$ close to $\Dcal$.
Much research has been carried out for Gaussian mixtures \cite{daskalakis2014faster}, log-concave distributions \cite{dumbgen2009maximum}, monotone distributions \cite{birge1987estimating}, sums of independent integer random variables \cite{daskalakis2013learning}, junta distributions \cite{aliakbarpour2016learning}, mixtures of structured distributions \cite{lindsay1995mixture}, and more.
The interested reader may wish to consult the survey \cite{diakonikolas2016learning} by Diakonikolas for additional background and references.

A black-box use of \Cref{thm:special_main} in this setting is to learn symmetric distributions that are locally sampleable.
Let $\Dcal$ be a symmetric distribution over $\bin^n$. 
Assume $\Dcal$ is produced by a $d$-local function, i.e., $\Dcal=f(\Ucal^m)$ for some $d$-local function $f\colon\bin^m\to\bin^n$.
Then for any $\eps > 0$ with $n$ sufficiently large in terms of $d$ and $\eps$, \Cref{thm:special_main} and the standard cover method \cite{yatracos1985rates} (see also \cite[Theorem 1.5.1]{diakonikolas2016learning}) imply that we can efficiently learn $\Dcal$ up to $\eps$-error in total variation distance with $O_d(1/\eps^2)$ samples with high probability.
We believe this should hold for all $\eps>0$. Indeed, in \Cref{sec:open_prob} we propose \Cref{conj:exact_classification} for an exact classification of locally sampleable symmetric distributions, which, if true, would imply the desired learning theoretic result.

\subsection{Open Problems}\label{sec:open_prob}

Recall that \Cref{thm:special_main} does not work for the case of $\eps=0$, i.e., \emph{exact} sampling.
Indeed, we cannot conclude mixtures of the form given by \Cref{thm:special_main} are the only distributions that can be exactly sampled by $\nc$ circuits. 
We identify the following example, which illustrates that this is not simply a weakness in our analysis, but rather there are other symmetric distributions that can be sampled exactly.

\begin{example}\label{ex:conj}
    Let $\Ucal_{1/4}^n$ be the $(1/4)$-biased product distribution over $\bin^n$. Additionally, let $\Deven$ and $\Dodd$ denote the uniform distribution over $n$-bit strings of even Hamming weight and odd Hamming weight, respectively.
    The distribution\footnote{More formally, $\Pcal$ is the distribution that assigns $x$ probability $\Ucal_{1/4}^n(x) + 2^{-n-1}\Deven(x) - 2^{-n-1}\Dodd(x)$.} 
    $$
    \Pcal = \Ucal_{1/4}^n + 2^{-n-1}\Deven - 2^{-n-1}\Dodd
    $$
    is not of the form given by \Cref{thm:special_main}, yet it can be sampled exactly by a bitwise AND of $\Deven$ and the uniform distribution over $\bin^n$, which is $3$-local.
    The full details can be found in \Cref{app:open_problems}.
\end{example}

Additionally, we suspect that, similarly to the uniform symmetric case \cite{kane2024locally2}, one can take the upper bound in \Cref{thm:special_main} to be linear in $\tvdist{f(\Ucal^m) - \Dcal}$ without any dependency on $d$.
Combining with the previous discussion, we conjecture the following strengthening of \Cref{thm:special_main} holds. 
Observe that \Cref{conj:exact_classification} captures \Cref{ex:conj}.

\begin{conjecture}\label{conj:exact_classification}
    For every $d \in \Nbb$, $\eps \in (0,1)$, and $n$ large enough in terms of $d$, if $f\colon\bin^m\to\bin^n$ is a $d$-local function and $f(\Ucal^m)$ is $\eps$-close in total variation distance to a symmetric distribution, then $f(\Ucal^m)$ is $O(\eps)$-close to some mixture
    \[
        \Mcal = \sum_i \alpha_i \cdot g_i(\Dcal^{(i)}_1, \dots, \Dcal^{(i)}_d),
    \]
    where each $g_i$ is a bitwise function and each $\Dcal^{(i)}_j$ is either the uniform distribution over $n$-bit strings of even Hamming weight or over $n$-bit strings of odd Hamming weight.
    Moreover, the mixing weights are determined by degree-$O_d(1)$ $\Fbb_2$-polynomials with $O_d(n)$ monomials, as in \Cref{thm:main}.
\end{conjecture}

Note the bitwise condition on the above $g_i$'s guarantees the output distribution is symmetric.

\paragraph{Paper Organization.}
We provide an overview of the proof of \Cref{thm:special_main} in \Cref{sec:overview}, as well as a brief comparison of our techniques to those of prior literature.
Background material and useful results are collected in \Cref{sec:prelim}.
The bulk of our work is in \Cref{sec:characterize}, where we state and prove \Cref{thm:main}, the full version of \Cref{thm:special_main}.
The appendices contain a number of deferred proofs.
\section{Proof Overview}\label{sec:overview}

In this section, we sketch the proof of \Cref{thm:special_main} before discussing how the details compare to prior work.

\subsection{Proof Overview of Theorem \ref{thm:special_main}}

We begin with a useful observation from \cite{kane2024locally2}: the total variation distance between the distribution $f(\Ucal^m)$ and any symmetric distribution $\Pcal$ over $\bin^n$ is, up to constant factors, equal to the distance between the corresponding Hamming weight distributions of $f(\Ucal^m)$ and $\Pcal$ plus the distance between $f(\Ucal^m)$ and its symmetrization (i.e., the distribution resulting from randomly permuting the coordinates of a string $x \sim f(\Ucal^m)$).
Expressed symbolically, we have
\begin{equation}\label{overview:eq:TVD_weight_plus_sym}
    \tvdist{f(\Ucal^m) - \Pcal} = \Theta(\tvdist{|f(\Ucal^m)| - |\Pcal|} + \tvdist{f(\Ucal^m) - f(\Ucal^m)_\sym}).
\end{equation}
The proof is straightforward, and it essentially follows from several applications of the triangle inequality which show that the two ways $f(\Ucal^m)$ can be far from $\Pcal$ are a weight mismatch or $f(\Ucal^m)$ being far from symmetric itself (see \Cref{lem:distance_to_sym}).

By assumption, we know $f(\Ucal^m)$ is $\eps$-close to a symmetric distribution $\Dcal$.
Thus, \Cref{overview:eq:TVD_weight_plus_sym} implies $\tvdist{f(\Ucal^m) - f(\Ucal^m)_\sym} = O(\eps)$.
Hence, it suffices to prove a weaker version of \Cref{thm:special_main} that only compares the weight distributions (\Cref{lem:combining_weights}).
More precisely, we want to show that for some $\delta$ tending to 0 with $\eps$, the Hamming weight distribution of $f(\Ucal^m)$, denoted $|f(\Ucal^m)|$, is $\delta$-close to a mixture $\Mcal$ of 
\begin{enumerate}
    \item The binomial distribution $\Bin(n,\gamma)$ with $n$ trials and success probability $\gamma$ for $\gamma$ an integer multiple of $2^{-d}$,

    \item The binomial distribution $\Bin(n,1/2)$ conditioned on the outcome being even, and

    \item The binomial distribution $\Bin(n,1/2)$ conditioned on the outcome being odd.
\end{enumerate}
Moreover, we want the mixing weights to be determined by low-degree $\Fbb_2$-polynomials with few monomials.
(It will later become clear what ``determined'' means in this context.)
We first prove a weaker version of this result that does not include mixing weight information (\Cref{lem:combining_weights_primitive}).
Afterwards, we will discuss how to obtain the desired control over the weights.
Before proceeding to the details, we first give a brief, high-level overview of the four main steps of the proof and how they fit together.

In the first step, we show that any fixing of the values of input bits which affect many output bits results in the bias of the output weight distribution of $f(\Ucal^m)$ concentrating around a fixed dyadic rational\footnote{Recall a \emph{dyadic rational} is a number that can be expressed as a fraction whose denominator is a power of two.} $\gamma$ multiple of $n$.
This allows us to represent the weight distribution $|f(\Ucal^m)|$ as a mixture of distributions produced by the restricted functions.
We then argue in step two that after grouping the parts of the mixture which concentrate around the same $\gamma$, each group assigns roughly the same amount of mass to any contiguous interval as the binomial distribution with success probability $\gamma$.
The third step is to prove a continuity result for each grouped part of the mixture; namely, that the mass assigned to some weight $w$ and to $w+\Delta$ for a small integer $\Delta$ are roughly equal.
(There is a slight subtlety here in the case of $\gamma=1/2$, but we will defer its discussion to later in the proof overview.)
This allows us to conclude in the fourth and final step that most output weights are assigned comparable mass by the part of the mixture of $|f(\Ucal^m)|$ concentrated around $\gamma$ and the corresponding binomial distribution.
In particular, our weight distribution is close in total variation distance to a mixture of binomial distributions, which is what we wanted to show.

\paragraph{Step 1: Removing Large Influences.}\label{par:remove_large_inf}

In an extremely ideal setting, we might wish that all of $f$'s output bits are roughly independent with bias around $\gamma = a/2^d$ for some integer $0 \le a \le 2^d$; in this case the output weight would resemble the binomial distribution $\Bin(n,\gamma)$.
Of course, this is far too much to assume.
In actuality, it may be that one specific input bit affects the value of \emph{every} output bit, so none of them are independent.

To progress toward this dream scenario, we follow in the footsteps of many prior works (e.g., \cite{viola2012complexity, lovett2011bounded, beck2012large, viola2020sampling, viola2023new, filmus2023sampling, kane2024locality, kane2024locally2, grier2025quantum}) by strategically conditioning on certain input bits to express the distribution as a mixture of more structured sub-distributions.
For example, suppose some input bit $b$ affects many output bits.
Since $b$ takes value 0 and 1 with equal probability, we may express the distribution $f(\Ucal^m)$ as the mixture
\[
    f(\Ucal^m) = \frac{1}{2}\cdot f\pbra{\Ucal^m \mid b = 0} + \frac{1}{2} \cdot f\pbra{\Ucal^m \mid b = 1}.
\]
Observe that the large influence of $b$ is no longer a problem in either of the restricted distributions, as its value is fixed.

In particular, we will condition on all ``high degree'' input bits that affect more than $n/A$ output bits for some $A$ to be chosen later.
The goal is now to argue that the function resulting from each conditioning produces a distribution which is roughly of the form we originally sought.
Note that this will not depend on the actual values that the input bits are set to in the conditioning.

We start with a result from \cite{kane2024locality}: let $\Dcal_q$ be the uniform distribution over $n$-bit strings of Hamming weight $q$.
If $q/n$ is far from every integer multiple of $2^{-d}$, then any $d$-local function must produce a distribution far (in total variation distance) from $\Dcal_q$.
This essentially follows from the fact that the input bits determining any fixed output bit can only be set in $2^d$ equally likely ways.
Observe that symmetric distributions are simply mixtures of $\Dcal_q$ for different $q$'s.
Thus we can show (\Cref{lem:dyadic_weight_eps}) that if $f(\Ucal^m)$ is close to a symmetric distribution $\Dcal$, then for a typical $x \sim \Ucal^m$, the normalized output weight $|f(x)|/n$ has distance at most $n^{-1/(800d)}$ from an integer multiple of $2^{-d}$.
Note that the closest multiple may be different for different inputs; however, after each conditioning $\rho \in \bin^S$ on the set of high degree input bits $S \subseteq [m] \coloneqq \cbra{1,2,\dots, m}$, we can strengthen this result (\Cref{lem:dyadic_weight_after_cond}) to say that for a typical $x \sim \Ucal^{[m]\setminus S}$, the normalized output weight $|f(x, \rho)|/n$ is close to the \emph{same} integer multiple of $2^{-d}$.
Here, we will only need ``high degree'' to correspond to affecting more than $n/O_{d}(1)$ many output bits (i.e., can take $A = O_d(1)$), although later we will obtain other constraints on how we must set $A$.
Henceforth, we will shorthand $f(x, \rho)$ by $f_\rho(x)$ for clarity.

At a high level, the proof uses the second-moment method.
By conditioning on the high degree input bits, we ensure that the variance of the resulting distribution is small, and thus we have concentration around a particular weight.
Moreover, this weight must be close to an integer multiple of $n/2^{d}$, or else our previous insights imply $f(\Ucal^m)$ cannot be close to $\Dcal$.
To be slightly more precise, a direct second-moment argument by itself is insufficient to rule out the case that some non-negligible fraction of the time the output weight is close to some other multiple of $n/2^{d}$.
However, one can show (\Cref{clm:lem:dyadic_weight_after_cond_5}) that if this occurs, $f_\rho(\Ucal^{[m]\setminus S})$ would also have to assign decent probability mass to weights between these integer multiples, which would again imply $f(\Ucal^m)$ cannot be close to $\Dcal$.

Finally, we combine these deductions to prove that for each conditioning $\rho \in \bin^S$ of the bits in $S$, there exists a set $T \coloneqq T_\rho \subseteq [n]$ of size $|T| \le O_{d,k}(1)$ such that every $k$-tuple of output bits in $[n] \setminus T$ has marginal distribution $\Ucal_{\gamma_\rho}^k$, the $\gamma_\rho$-biased product distribution over $\bin^k$, for $\gamma_\rho$ an integer multiple of $2^{-d}$ (see \Cref{prop:independence_after_cond}).
Here, $k$ is some parameter at most $O_d(\log(1/\eps))$.
The proof operates by constructing a degree-$k$ multilinear polynomial $P_i \colon \bin^n \to \bin$ for each $k$-tuple with the ``wrong'' distribution, where the expectation of $P_i$ over $f_\rho(\Ucal^{[m]\setminus S})$ is larger than over the $\gamma$-biased product distribution $\Ucal_\gamma^n$ by at least an additive $2^{-kd}$ term.
Note this follows from using our locality bound to view $P_i(f_\rho(\Ucal^{[m]\setminus S}))$ as a polynomial of degree $kd$ in the input bits.
Summing the polynomials together into $P = \sum_i P_i$ magnifies the difference in expectations, and if the number of terms (i.e., number of bad $k$-tuples) is too large, we end up contradicting the assumption that $f(\Ucal^m)$ is close to a symmetric distribution.

\paragraph{Step 2: Kolmogorov Distance.}\label{par:KD}

We now group the restricted functions according to their biases, defining $F_\gamma(\Ucal^{[m]\setminus S}) = \E_{\rho : \gamma_\rho = \gamma} [f_\rho(\Ucal^{[m]\setminus S})]$ for each $\gamma$.
In this second step, we aim to show that every contiguous interval of output weights is assigned roughly the same amount of mass by $|F_\gamma(\Ucal^{[m]\setminus S})|$ and $\Bin(n,\gamma)$.
Since this difference in probability mass is convex over mixtures, it suffices to consider the individual distributions $|f_\rho(\Ucal^{[m]\setminus S})|$.
Moreover, it is enough to provide a bound on the \emph{Kolmogorov distance} 
\[
    \max_t \bigg|\Pr\sbra{|f_\rho(\Ucal^{[m]\setminus S})| \ge t} - \Pr\sbra{\Bin(n, \gamma) \ge t}\bigg|.
\]

In a sentence, the bound follows from combining the $k$-wise independence of most output coordinates with the fact that $T$ is too small to have much of an effect.
In the case of $\gamma = 1/2$ and $|T| = 0$, Diakonikolas, Gopalan, Jaiswal, Servedio, and Viola \cite{diakonikolas2010bounded} gave an upper bound of roughly $1/\sqrt{k}$ using techniques from approximation theory. 
(In fact, their result holds for arbitrary threshold functions.)
This was generalized by Gopalan, O'Donnell, Wu, and Zuckerman \cite{gopalan2010fooling} to include the case of arbitrary $\gamma \in (0,1)$.
In our case, $|T|$ will likely not be 0, but it is small enough that a similar result (\Cref{prop:kol_dist}) still holds.
In the edge cases of $\gamma \in \bin$, we can no longer use \cite{diakonikolas2010bounded, gopalan2010fooling}, nor would the size of $T$ be negligible even if we could, but these special cases can be addressed later via simple arguments (see, e.g., the proof of \Cref{lem:close_to_bin_weights_mixture}).

For reasons that will become apparent in the subsequent step, we will also need a comparable bound in the case of $\gamma = 1/2$, even accounting for the parity of the output weight.
That is, we wish to show
\[
    \Pr\sbra{|f_\rho(\Ucal^{[m]\setminus S})| \ge t \text{ and } |f_\rho(\Ucal^{[m]\setminus S})| \text{ is even}} \approx \Pr\sbra{\Bin(n, 1/2) \ge t}\cdot \Pr\sbra{|f_\rho(\Ucal^{[m]\setminus S})| \text{ is even}}.
\]
Our analysis here is similar to the previous case, but we now crucially rely on \cite[Theorem 3.1]{chattopadhyay2020xor}.
In our context, it implies that there is a small set of input bits $R \subseteq [m]\setminus S$ such that re-randomizing over $R$ typically re-randomizes the parity of $f_\rho$'s output weight.
Moreover, since $R$ is small and we have already conditioned on input bits of large degree, the vast majority of output bits are unaffected by $R$.
Thus, we can compare $\Pr\sbra{|f_\rho(\Ucal^{[m]\setminus S})| \ge t}$ to $\Pr\sbra{\Bin(n, 1/2) \ge t}$ using these unaffected output bits as a proxy for the entirety of the output bits.
We briefly note that for the error bounds on the ``parity Kolmogorov distance'' to be meaningful given our choice of $k$, we need to take $A$ (defined in our earlier high degree threshold) to be $O_{d,\eps}(1)$ -- larger than is necessary for \hyperref[par:remove_large_inf]{Step 1}.

\paragraph{Step 3: Approximate Continuity.}\label{par:approx_cont}

Our third step is to argue that for each bias $\gamma$, the distribution $F_\gamma(\Ucal^{[m]\setminus S})$ can be expressed as a mixture $\lambda\cdot E_\gamma + (1-\lambda)\cdot W_\gamma$, where $\lambda$ is small and $W_\gamma$ assigns similar probability mass to similar weights.
Once again, it suffices to analyze the distributions produced by the individual restricted functions $f_\rho(\Ucal^{[m]\setminus S})$.
Here, we show that $f_\rho(\Ucal^{[m]\setminus S})$ can be written as a mixture of distributions $E$ and $W$, where $E$ is extremely far from every symmetric distribution supported on strings of Hamming weight $\gamma n \pm n^{2/3}$, and the weight distribution of $W$ satisfies a certain continuity property (see \Cref{prop:continuity}).
More specifically, for positive integers $w$ and $\Delta$, the probability that $W$'s output weight is $w$ differs by no more than $O_d\pbra{\Delta/n}$ from the probability that $W$'s output weight is $w + \Delta$.
Since $f(\Ucal^m)$ is close to $\Dcal$, one can show that $F_\gamma(\Ucal^{[m]\setminus S})$ is close to $\Dcal$ conditioned on the output weight being $\gamma n \pm n^{2/3}$ (see \Cref{clm:gamma_restricted_TVD}).
Thus, proving the above result for $f_\rho(\Ucal^{[m]\setminus S})$ implies minimal mass will typically be assigned to the $E$ part of the mixtures, and so $\lambda$ (in the mixture defining $F_\gamma(\Ucal^{[m]\setminus S})$) will be small, as desired.

In our analysis, we will require a structural result about hypergraphs from \cite{kane2024locality}.
Observe that we can associate any function $g\colon\bin^m \to \bin^n$ to a hypergraph on the vertex set $[n]$ with an edge for each input bit $b$ containing all of the output bits that depend on $b$.
In the case that $g$ is $d$-local, this hypergraph has maximum degree at most $d$.
We follow standard hypergraph terminology in defining the \emph{neighborhood} of a vertex $v$ to be the set of vertices sharing an edge with $v$.
We additionally call two neighborhoods $N_1, N_2$ \emph{connected} if there exist two adjacent (i.e., contained in the same edge) vertices $v_1 \in N_1, v_2 \in N_2$.

By \cite[Corollary 4.11]{kane2024locality}, we can find a collection of $r = \Omega_d(n)$ non-connected neighborhoods in $\Hcal$ of size $O_d(1)$ by only removing $O_d(n)$ (with a small implicit constant) edges.
Translating the result to $f_\rho$, we find that there exists a not-too-large set $B \subseteq [m]\setminus S$ of input bits such that any conditioning $\sigma \in \bin^B$ yields a sub-function $f_{\rho, \sigma} \colon \bin^{[m]\setminus (S\cup B)} \to \bin^n$ with many small, pairwise independent collections of output bits.
We will define $E$ and $W$ according to the behavior of these neighborhoods, similarly to the arguments in \cite{kane2024locality, kane2024locally2, grier2025quantum}.

Suppose for at least $r/2$ of the $r$ non-connected neighborhoods $N_1, \dots, N_r$, we have that $f_{\rho,\sigma}(\Ucal^{[m]\setminus (S\cup B)})$ restricted to the neighborhood $N_i$ has a marginal distribution which differs from the $\gamma$-biased product distribution over $N_i$, denoted $\Ucal_\gamma^{N_i}$.
Since $\Ucal_\gamma^{N_i}$ is pointwise close to the uniform distribution over strings of Hamming weight around $\gamma n$, the marginal distributions of these neighborhoods (in $f_{\rho,\sigma}$) must also differ from the marginal distributions of the symmetric distribution $\Dcal$.
We can then accumulate these errors via concentration bounds to show that $f_{\rho,\sigma}(\Ucal^{[m]\setminus (S\cup B)})$ is far from $\Dcal$ (see \Cref{lem:tvdist_after_product}).

We now set $E$ to be the mixture over all conditions of the bits in $B$ where most resulting neighborhoods differ from the corresponding $\gamma$-biased product distribution, and set $W$ to be the mixture over the remaining conditionings.
Since each sub-distribution in the mixture $E$ is far from $\Dcal$ by the previous paragraph, a union bound argument implies $E$, itself, must also be far from $\Dcal$ (see \Cref{lem:tvdist_after_conditioning}).
It remains to show the continuity property for $W$.

Suppose the $r$ non-connected neighborhoods are generated by $v_1, v_2, \dots, v_r \in [n]$ in the sense that all bits in the $i$-th neighborhood $N_i$ are either $v_i$ or in an edge that also contains $v_i$.
We further condition on all input bits that do not affect any of $v_1, \dots, v_r$, so that the value of every output bit outside of $N_1 \cup \cdots \cup N_r$ is fixed.
In this way, the output weight of $f_{\rho,\sigma}(\Ucal^{[m]\setminus (S\cup B)})$ becomes a fixed integer (corresponding to the fixed bits outside of the neighborhoods) plus the sum of the neighborhoods' output weights.
Since we are in the case where most neighborhoods are extremely structured, we can show that for each $2 \le \ell \le t$, the output weight distribution modulo $\ell$ of $N_i$ is not constant for a constant fraction of the $N_i$'s with high probability (see \Cref{clm:X_i_weight_fixed}).
Finally, we can obtain our desired continuity result through known density comparison theorems for sums of independent, non-constant integer random variables, such as \cite[Theorem A.1]{kane2024locally2}, which relies on classical anticoncentration tools.

There is one important exception to the above analysis: the case of $\ell = 2$ and $\gamma = 1/2$.
Here, we are not guaranteed to typically get many non-connected neighborhoods with variable output weight modulo 2 (even if most of them have the ``correct'' marginal distribution).
To better understand where our reasoning breaks down, consider one of the most surprising examples in this area of work.
Define $h\colon \bin^n\to\bin^n$ to be
\[
    h(x_1, \dots, x_n) = (x_1 \oplus x_2, x_2 \oplus x_3, \dots, x_{n-1}\oplus x_n, x_n \oplus x_1),
\]
so that $h(\Ucal^n)$ is the uniform distribution over $n$-bit strings of even Hamming weight.
Observe that $h$ is extremely simple, only requiring two bits of locality.
Additionally, the marginal distribution onto any $k \le n-1$ coordinates is exactly the product distribution $\Ucal_{1/2}^k$.

Let us focus on a particular output bit $y_i = x_i \oplus x_{i+1}$. 
(For simplicity, assume $i$ is not too close to 1 or $n$.)
Its neighborhood is $y_{i-1} = x_{i-1} \oplus x_{i}$, $y_i$, and $y_{i+1} = x_{i+1} \oplus x_{i+2}$, so its Hamming weight modulo 2 is
\[
    (x_{i-1} \oplus x_{i}) \oplus (x_i \oplus x_{i+1}) \oplus (x_{i+1} \oplus x_{i+2}) = x_{i-1} \oplus x_{i+2}.
\]
If we follow our earlier analysis and fix the value of the input bits that do not affect $y_i$ (i.e., $x_{i-1}$ and $x_{i+2}$), the neighborhood's Hamming weight modulo two is fixed, regardless of the values of $x_i$ and $x_{i+1}$.

At a high level, the reason for this exceptional case is that it is the only setting of $\ell$ and $\gamma$ for which the binomial distribution $\Bin(t, \gamma) \bmod \ell$ can equal $\Bin(t-1, \gamma) \bmod \ell$.
In other words, the weight distribution modulo $\ell$ of $\Ucal_\gamma^k$ conditioned on the first bit being 0 is different than that distribution conditioned on the first bit being 1, unless $\ell = 2$ and $\gamma = 1/2$.
This difference means that in the case of $\ell = 2$ and $\gamma = 1/2$, we can only derive a continuity result for weights that are an even distance apart.
Hence, we must be cognizant of the output weight's parity throughout much of our analysis, which explains the required parity version of our Kolmogorov distance result mentioned earlier in the proof overview.

\paragraph{Step 4: Putting It Together.}

At this point, all that remains is to combine the pieces.
Recall we have already conditioned on high degree (i.e., larger than $n/O_{d,\eps}(1)$) input bits to find sub-functions $\cbra{f_\rho}_\rho$, each producing a distribution with some approximate bias $\gamma_\rho$.
Additionally, the mixtures $F_\gamma(\Ucal^{[m]\setminus S})$ obtained by grouping the sub-functions around their biases each satisfy a Kolmogorov distance bound and approximate continuity.

We partition $\cbra{0,1,\dots, n}$ into consecutive intervals of length $c\sqrt{n}$ for some small $c > 0$, and restrict our attention to the $O(\log(1/\alpha))$ of them that contain all but $O(\alpha)$ of the mass.
By our Kolmogorov distance result (\Cref{lem:biased_kol_dist}), we have that $|F_\gamma(\Ucal^{[m]\setminus S})|$ and the binomial distribution $\Bin(n,\gamma)$ assign similar mass to each of these intervals.
Moreover, our continuity result (\Cref{prop:continuity}) implies the mass in a fixed interval is almost uniformly distributed.
Thus for most weights $w$, $|F_\gamma(\Ucal^{[m]\setminus S})|$ assigns roughly the same mass to $w$ as $\Bin(n,\gamma)$ does.
Summing over all weights provides an upper bound on the total variation distance between the two weight distributions.
We remark that there is a slight complication in the case of $\gamma = 1/2$, as there \Cref{prop:continuity} only lets us compare weights that are an even distance apart. 
However, by further subdividing each interval we consider into its even and odd parts, we are able to carry out a similar argument as before.

Finally, we recall that $f(\Ucal^m)$ is a mixture of the distributions $F_\gamma(\Ucal^{[m]\setminus S})$.
Since the weight distribution of each $F_\gamma(\Ucal^{[m]\setminus S})$ is close to the weight distribution corresponding to one from \Cref{thm:special_main}, their mixture $|f(\Ucal^m)|$ naturally is close to a mixture of those weight distributions (see \Cref{lem:combining_weights_primitive}).
We conclude by recalling that \Cref{overview:eq:TVD_weight_plus_sym} implies it was sufficient to classify the output weight distribution.

\paragraph{Mixing Weights.}

The above argument shows that $f(\Ucal^m)$ is close to a mixture $\Mcal$ of the form specified in \Cref{thm:special_main}, but without the additional control on the mixing weights.
The reason for this shortcoming is that the threshold $A$ at which we consider an input bit ``high degree'' depends not just on $d$, but also on $\eps$, and so the mixing weights in the obtained mixture also depend on $\eps$.
If we instead chose a threshold that only depended on $d$, we would not have been able to obtain an effective error bound in the Kolmogorov distance step (i.e., \Cref{prop:kol_dist}) when $\gamma = 1/2$.

The key observation is that the arguments of \hyperref[par:remove_large_inf]{Step 1} only require conditioning on input bits of degree larger than $n/O_d(1)$.
In other words, if we perform some different conditioning $\kappa$ on the set of input bits $S' \subseteq [m]$ above the weaker threshold $n/O_d(1)$, the distributions generated by the restricted functions will still have output weights which strongly concentrate around some integer multiple of $n/2^{d}$.
Moreover, the mixing weights in this setting \emph{are} integer multiples of $2^{-O_d(1)}$.

It remains to argue that the mixing weights on $\Mcal$, the mixture we proved $f(\Ucal^m)$ is close to by the stronger conditioning $\rho$, correspond to the mixing weights derived from the weaker conditioning $\kappa$ (see \Cref{lem:combining_weights}).
Since both conditionings produce mixtures of the same distribution, we have
\[
     \sum_{\gamma} \Pr_{\rho}\sbra{\gamma_\rho = \gamma} \cdot \E_{\rho : \gamma_\rho = \gamma} f_\rho(\Ucal^{[m]\setminus S}) = f(\Ucal^m) = \sum_{\gamma} \Pr_{\kappa}\sbra{\gamma_\kappa = \gamma} \cdot \E_{\kappa : \gamma_\kappa = \gamma} f_\kappa(\Ucal^{[m]\setminus S'}).
\]

By standard concentration bounds, the probability that $x \sim f(\Ucal^m)$ has Hamming weight close to $\gamma n$ is almost entirely determined by the mass assigned to the distributions whose output weight concentrates around $\gamma$.
In other words, we know $\Pr_{\rho}\sbra{\gamma_\rho = \gamma} \approx \Pr_{\kappa}\sbra{\gamma_\kappa = \gamma}$ for each $\gamma$.
By our earlier analysis and the assumption that $f(\Ucal^m)$ is close to $\Dcal$, it must be that the mixing weight on $\Bin(n, \gamma)$ in $\Mcal$ is roughly $\Pr_{\rho}\sbra{\gamma_\rho = \gamma}$, so it must also be close to $\Pr_{\kappa}\sbra{\gamma_\kappa = \gamma}$, which has the form we sought.

The analysis in the case of $\gamma = 1/2$ is similar, but as always, we need to keep track of the output weight's parity.
Here, the analogous equivalence is
\begin{equation}\label{eq:pf_overview_even}
    \Pr_{\rho}\sbra{\gamma_\rho = 1/2 \land |f(\Ucal^m)| \text{ is even}} \approx \frac{1}{2^{|S'|}}\sum_{\kappa : \gamma_\kappa = 1/2} \Pr_{x\sim \Ucal^{[m]\setminus S'}}\sbra{|f_\kappa(x)| \text{ is even}}.
\end{equation}
Since $f$ is $d$-local with $n$ output bits, the parity of $|f_\kappa(x)|$ can be represented as a degree-$d$ $\Fbb_2$-polynomial with $O_d(n)$ monomials.
Thus, the mixing weight for the uniform distribution over $n$-bit strings of even Hamming weight in $\Mcal$ must be close to the right-hand side of \Cref{eq:pf_overview_even}, which is the precise sense in which the mixing weights are ``determined by low-degree $\Fbb_2$-polynomials''.
The odd parity case is essentially identical.
This concludes the proof overview of \Cref{thm:special_main}.

\subsection{Technical Comparison to Prior Works}

We now briefly survey the approaches of related prior works.
To avoid an overly verbose digression, we attempt to focus only on the most relevant papers and do not hope to be exhaustive.
Similarly, we restrict our attention within the mentioned works to the results and techniques most pertinent to our own; most, if not all, of the papers discussed contain a number of other interesting results.

The study of the complexity of sampling goes back to at least the 1980s in the influential work of Jerrum, Valiant, and Vazirani \cite{jerrum1986random}.
In the context of shallow circuits, several clever sampling constructions (including the one in the introduction) were devised in \cite{babai1987random, boppana1987one, impagliazzo1996efficient}, while the first serious treatment of the complexity of sampling (with shallow circuits) appeared in \cite{viola2012complexity}.
There, Viola proved an assortment of sampling-related results, but we will confine our attention to two on the hardness of approximately generating $\Dcal_{n/2}$, the uniform distribution over $n$-bit strings of Hamming weight $n/2$.
(This specific choice is purely for simplicity, and both results also apply to the setting of $\Dcal_{\alpha n}$ for $\alpha \in (0,1)$.)

The first result we will mention \cite[Theorem 1.6]{viola2012complexity} is an unconditional lower bound of $2^{-O(d)} - O(1/n)$ on the distance between $f(\Ucal^m)$, where $f\colon \bin^m \to \bin^n$ is a $d$-local\footnote{In fact, the result is proven for an adaptive version of locality.} function, and $\Dcal_{n/2}$.
The proof, much like our own (see \hyperref[par:KD]{Step 2}), relies on a Kolmogorov distance bound obtained from a $k$-wise independence assumption \cite{diakonikolas2010bounded, gopalan2010fooling}.
If the output distribution $f(\Ucal^m)$ is $k$-wise independent for some large integer $k$, then \cite[Theorem I.5]{gopalan2010fooling} implies it has the wrong Hamming weight with constant probability, and we are done.
Otherwise, there exists a $k$-tuple $T \subseteq [n]$ of output bits that are not uniformly distributed over $\bin^k$.
In this case, one observes that the probability assigned to any element of $\bin^k$ is an integer multiple of $2^{-kd}$, so the marginal distribution over $T$ must be at least $2^{-kd}$-far from uniform.
The proof concludes by noting that the marginal distribution of $\Dcal_{n/2}$ onto $T$ is very close to uniform.
Note that we use similar granularity ideas in the proof of \Cref{prop:independence_after_cond} (in \hyperref[par:remove_large_inf]{Step 1}).

The second result \cite[Theorem 1.3]{viola2012complexity} provides a much stronger lower bound of $1 - 1/\poly(n)$ for the distance between $f(\Ucal^m)$ and $\Dcal_{n/2}$ for any $O(\log n)$-local function $f$, but is \emph{conditional} on the input size $m$ not substantially exceeding the information-theoretic minimum required to generate the target distribution.
The proof begins by partitioning the input bits $u$ of $f$ as $u = (x,y)$ and expressing
\[
    f(u) = f(x,y) = h(y) \circ g_1(x_1,y) \circ \dots \circ g_s(x_s, y),
\]
where each $g_i$ may depend arbitrarily on $y$ but only on the single bit $x_i$ of $x$.
Additionally, each $g_i(x_i, y) \in \bin^{O(d)}$.
By a greedy approach, this can be accomplished with $s \ge \Omega(n/d^2)= n/\polylog(n)$.
The argument proceeds by conditioning on certain input bits (in this case $y$) and splitting into two scenarios depending on the amount of concentration; this strategy has been employed by several later works, including this one.

In the ``concentrated'' case where at least $\sqrt{n}$ many $g_i$'s become fixed, the output distribution of $f(x,y)$ has small support. 
Even taking the union over all choices of such $y$, one finds that many elements in the support of $\Dcal_{n/2}$ cannot be obtained.
Note that this is where the input length restriction originates.
Alternatively, at least $s - \sqrt{n} = n/\polylog(n)$ of the $g_i$'s take multiple values.
One would like to apply standard anticoncentration results to argue the output weight is too spread out, but as noted earlier in the proof overview (see \hyperref[par:approx_cont]{Step 3}), it may be that the Hamming weight of a $g_i$ is fixed, even if the output is not.
Viola circumvents this issue by adding an additional ``test'' to determine whether at most $2\sqrt{n}$ many $g_i$'s can output the all zeros string; in this case, taking two or more values implies the weight cannot be fixed, and thus anticoncentration inequalities can be applied.

We now turn to the recent work of Filmus, Leigh, Riazanov, and Sokolov \cite{filmus2023sampling}, who addressed the case of $\Dcal_{k}$ for $k = o(n)$.
They proved \cite[Theorem 1.2]{filmus2023sampling} that $\widetilde{\Omega}(\log(n/k))$ bits of locality (even adaptively chosen) are required to generate $\Dcal_{k}$ to constant error.
Moreover, the same result holds for the uniform distribution over strings whose Hamming weight lies in a set $S$, where $\max_{s\in S}s = k$.
Their proof first reduces to the case of considering $\Dcal_1$, and then proceeds via a hitting set vs independent set dichotomy.
This can again be viewed as fitting into the concentration vs anticoncentration paradigm.
If there are $O(d2^d)$ input bits that together affect every output bit, then one can afford to condition on them and induct on the remaining $(d-1)$-local sub-distributions.
Bounds from these distributions can eventually be recombined to deduce bounds on the full distribution using some version of a union bound, such as \Cref{lem:tvdist_after_conditioning}.
Otherwise, there are $\Omega(2^d)$ independent output bits, and one can argue it is very likely that two of these bits evaluate to 1, since each (nonzero) output bit is 1 with probability at least $2^{-d}$.
Interestingly, the authors are able to use the \emph{robust sunflower lemma} \cite{rossman2014monotone, alweiss2021improved} from extremal combinatorics to improve their quantitative bounds; we refer the reader to their original paper for details.

Also recently, strong unconditional locality lower bounds for sampling $\Dcal_{\alpha n}$ where $\alpha$ is a non-dyadic rational were proven independently in \cite{viola2023new, kane2024locality}.
At a very high level, both proofs follow from the observation (recorded in \cite{viola2012bit}) that the marginal distribution on any output bit of a $d$-local function $f$ must have probabilities that are integer multiples of $2^{-d}$, which cannot accurately approximate (say) $1/3$, corresponding to the marginal distribution for $\Dcal_{n/3}$.
We begin by describing the approach of Viola \cite[Theorem 25]{viola2023new}, continuing to focus on the case of $\alpha = 1/3$ for simplicity.
First assume by contradiction $\tvdist{f(\Ucal^m) - \Dcal_{n/3}} \le 1-\eps$ for some $\eps$ to be determined.
The proof proceeds by arguing that there must exist a large subset $R \subseteq \bin^m$ of the inputs, such that $f(\Ucal(R))$, $f$ evaluated on inputs drawn uniformly at random from $R$, lands entirely in the support of $\Dcal_{n/3}$ and has almost full min-entropy.\footnote{Recall the \emph{min-entropy} of a random variable $X$ is $\log(1/\max_{x \in \supp{X}} \Pr[X = x])$.}
The so-called \emph{fixed-set lemma} of Grinberg, Shaltiel, and Viola \cite{grinberg2018indistinguishability} then implies that $R$ can be further restricted to a large subset $R' \subseteq R$ whose uniform distribution is nearly indistinguishable by $d$-local functions from a product distribution where each bit is either fixed or uniform.
Then the large min-entropy of $f(\Ucal(R))$ (and thus $f(\Ucal(R'))$) combined with the fact that the distribution is contained in the support of $\Dcal_{n/3}$ implies there must be at least one output bit whose marginal distribution is extremely close to $1/3$.
This, however, contradicts the fact that $1/3$ cannot accurately be approximated by an integer multiple of $2^{-d}$.

The proof in \cite{kane2024locality} instead operates very similarly to the details of \hyperref[par:approx_cont]{Step 3}.
Through a graph-theoretic lemma \cite[Corollary 4.8]{kane2024locality}, they show that there must be a set of at most $r/2^d$ many input bits upon whose conditioning results in $r$ many independent output bits for $r \approx n / 2^{d^2}$.
As previously noted, each of these output bits inevitably incurs some error from mismatched marginal distributions, and these errors can be aggregated via standard concentration inequalities (as in \Cref{lem:tvdist_after_product}).
Bounds from the individual sub-distributions can once again be recombined via a union bound result like \Cref{lem:tvdist_after_conditioning}.

The work \cite{kane2024locality} also contains locality bounds for $\Dcal_{\alpha n}$ when $\alpha$ is a dyadic rational.
The proof has strong similarities to the non-dyadic case, except that now the marginal distributions on output bits of $f$ do not necessarily disagree with those of the target distribution.
The authors overcome this issue by turning to the familiar concentration vs.~anticoncentration paradigm.
The details have almost entirely been spelled out already in \hyperref[par:approx_cont]{Step 3}, as our analysis is nearly identical at that part.
We briefly note that the quantitative dependencies here are much worse than in the non-dyadic case, essentially for the reason that one needs to obtain a much richer structure than simply independent output bits, and so the analogous graph-theoretic lemma in this case has poor bounds.\footnote{Moreover, the bounds in this graph-theoretic lemma are essentially best possible; see \cite[Appendix A]{kane2024locality}.}

Finally, let us mention the predecessor of this work \cite{kane2024locally2}, which classifies what uniform symmetric distributions can be sampled by functions of bounded locality, confirming a conjecture of \cite{filmus2023sampling}.
In broad strokes, it follows similarly to the proof overview, so we will be brief.
Through a strategic conditioning of input bits, one may reduce to the case where no input bits have large degree and almost all the output bits are $k$-wise independent.
Then one can show the resulting Hamming weight distribution is close in Kolmogorov distance to the binomial distribution, and it satisfies a continuity-type property.
Combining these steps together, one classifies the Hamming weight distribution, which is enough to classify the actual distribution, since we assume $f(\Ucal^m)$ is close to a uniform symmetric distribution.
Much of the present work is an extension of the ideas in \cite{kane2024locally2} to the general symmetric case, although there are several parts (e.g., much of the analysis in \hyperref[par:remove_large_inf]{Step 1} and the finer control on the mixing weights) that require substantially new ideas.

We conclude by noting that several works \cite{lovett2011bounded, beck2012large, viola2014extractors, viola2020sampling, chattopadhyay2022space} prove hardness results against the more powerful classes of $\ac$ circuits or read-once branching programs.
There is a reasonable overlap in the techniques used with those discussed above, but the proofs of all these results rely on the special properties of certain pseudorandom objects like good codes \cite{lovett2011bounded, beck2012large, chattopadhyay2022space} or extractors \cite{viola2012extractors, viola2020sampling}.
\section{Preliminaries}\label{sec:prelim}

For a positive integer $n$, we use $[n]$ to denote the set $\cbra{1,2,\ldots,n}$.
We use $\Rbb$ to denote the set of real numbers, use $\Nbb=\cbra{0,1,2,\ldots}$ to denote the set of natural numbers, and use $\Zbb$ to denote the set of integers.
For a binary string $x$, we use $|x|$ to denote its Hamming weight.
We use $\log(x)$ and $\ln(x)$ to denote the logarithm with base $2$ and $e$ respectively.
For $a,b \in \Rbb^{\ge 0}$, we use $a \pm b$ to shorthand a number in the interval $[a-b, a+b]$.

\paragraph{Asymptotics.}
We use the standard $O(\cdot), \Omega(\cdot), \Theta(\cdot)$ notation, and emphasize that in this paper they only hide universal positive constants that do not depend on any parameter. 
Occasionally we will use subscripts to suppress a dependence on particular variable (e.g.,\ $O_d(1)$).
The notation $\poly(\cdot)$ is also sometimes used to denote a quantity that is polynomial with an unspecified exponent.
That is, $\poly(n) = \Theta(n^c)$ for some $c > 0$.

\paragraph{Locality and Hypergraphs.}
Let $f\colon\bin^m\to\bin^n$. We say $f$ is a $d$-local function if each output bit $i\in [n]$ depends on at most $d$ input bits. Unless otherwise stated, $n,m,d$ are positive integers.
 
We sometimes take an alternative view, using hypergraphs to model the dependency relations in $f$.
Let $G=(V,E)$ be an (undirected) hypergraph.
For each $i\in V$, we use $I_G(i)\subseteq E$ to denote the set of edges that are incident to $i$.
We say $G$ has maximum degree $d$ if $|I_G(i)|\le d$ holds for all $i\in V$.
Define $N_G(i)=\cbra{i'\in V\colon I_G(i)\cap I_G(i')\neq\emptyset}$ to be the neighborhood of $i$. We visualize the input-output dependencies of a function $f\colon\bin^m\to\bin^n$ as a hypergraph on the output bits $[n]$ with an edge for each input bit containing all of the output
bits that depend on it.
Note that a $d$-local function corresponds to a hypergraph with maximum degree $d$.

\paragraph{Probability.}
Let $\Pcal$ be a (discrete) distribution. We use $x\sim\Pcal$ to denote a random sample $x$ drawn from the distribution $\Pcal$.
If $\Pcal$ is a distribution over a product space, then we say $\Pcal$ is a product distribution if its coordinates are independent.
In addition, let $S$ be a non-empty set. If $S$ indexes $\Pcal$, we use $\Pcal[S]$ to denote the marginal distribution of $\Pcal$ on coordinates in $S$. 
We reserve $\Ucal$ to denote the uniform distribution over $\bin$.

For a deterministic function $f$, we use $f(\Pcal)$ to denote the output distribution of $f(x)$ given a random $x\sim\Pcal$. 
For every event $\Ecal$, we define $\Pcal(\Ecal)$ to denote the probability that $\Ecal$ occurs under distribution $\Pcal$.
In addition, we use $\Pcal(x)$ to denote the probability mass of $x$ under $\Pcal$, and use $\supp{\Pcal}=\cbra{x:\Pcal(x)>0}$ to denote the support of $\Pcal$.
Of particular interest to this work is the following special class of distributions.
\begin{definition}[Symmetric Distribution]
    A distribution $\Pcal$ over $\bin^n$ is \emph{symmetric} if $\Pcal(x) = \Pcal(y)$ for any $x,y \in\bin^n$ such that $|x| = |y|$.
\end{definition}

Let $\Qcal$ be a distribution. We use $\tvdist{\Pcal-\Qcal}=\frac12\sum_x\abs{\Pcal(x)-\Qcal(x)}$ to denote their total variation distance.\footnote{To evaluate total variation distance, we need two distributions to have the same sample space. This will be clear throughout the paper and thus we omit it for simplicity.}
We say $\Pcal$ is $\eps$-close to $\Qcal$ if $\tvdist{\Pcal-\Qcal}\le\eps$, and $\eps$-far otherwise.

\begin{fact}\label{fct:tvdist}
Total variation distance has the following equivalent characterizations:
$$
\tvdist{\Pcal-\Qcal}=\max_{\text{event }\Ecal}\Pcal(\Ecal)-\Qcal(\Ecal)=\min_{\substack{\text{random variable }(X,Y)\\\text{$X$ has marginal $\Pcal$ and $Y$ has marginal $\Qcal$}}}\Pr\sbra{X\neq Y}.
$$
\end{fact}

Let $\Pcal_1,\ldots,\Pcal_t$ be distributions.
Then $\Pcal_1\times\cdots\times\Pcal_t$ is a distribution denoting the product of $\Pcal_1,\ldots,\Pcal_t$.
We also use $\Pcal^t$ to denote $\Pcal_1\times\cdots\times\Pcal_t$ if each $\Pcal_i$ is the same as $\Pcal$.
For a finite set $S \subseteq [t]$, we use $\Pcal^S$ to denote the distribution $\Pcal^t$ restricted to the coordinates of $S$.
We say distribution $\Pcal$ is a mixture (or convex combination) of $\Pcal_1,\ldots,\Pcal_t$ if there exist $\alpha_1,\ldots,\alpha_t\in[0,1]$ such that $\sum_{i\in[t]}\alpha_i=1$ and $\Pcal(x)=\sum_{i\in[t]}\alpha_i\cdot\Pcal_i(x)$ for all $x$ in the sample space. When it is clear from context, we will occasionally write mixtures more simply as $\Pcal=\sum_{i\in[t]}\alpha_i\cdot\Pcal_i$.
In the case where the $\Pcal_i$ are all of the form $f_i(\Ucal^m)$ for deterministic functions $f_1, \dots, f_t\colon \bin^m \to \bin^n$, we will occasionally abuse notation by writing $F(\Ucal^m)$ for the mixture $F = \sum_{i\in[t]}\alpha_i\cdot f_i(\Ucal^m)$.

We collect two useful probabilistic results from prior work.
The first says that two distributions must be far apart if many of their marginals do not match.

\begin{lemma}[{\cite[Lemma 4.2]{kane2024locality}}]\label{lem:tvdist_after_product}
Let $\Pcal$, $\Qcal$, and $\Wcal$ be distributions over an $n$-dimensional product space, and let $S\subseteq[n]$ be a non-empty set of size $s$.
Assume
\begin{itemize}
\item $\Pcal[S]$ and $\Wcal[S]$ are two product distributions,
\item $\tvdist{\Pcal[\cbra{i}]-\Wcal[\cbra{i}]}\ge\eps$ holds for all $i\in S$, and
\item $\Wcal(x)\ge\nu\cdot\Qcal(x)$ holds for some $\nu>0$ and all $x$.
\end{itemize}
Then
$$
\tvdist{\Pcal-\Qcal}\ge1-2\cdot e^{-\eps^2s/2}/\nu.
$$
\end{lemma}

The second allows us to reason about the distance between a fixed distribution and a mixture by reasoning about the individual distributions composing the mixture.

\begin{lemma}[{\cite[Corollary 4.2]{viola2020sampling}}]\label{lem:tvdist_after_conditioning}
Let $\Pcal_1,\ldots,\Pcal_t$ and $\Qcal$ be distributions.
Assume there exists a value $\eps$ such that $\tvdist{\Pcal_i-\Qcal}\ge1-\eps$ for all $i\in [t]$.
Then for the balanced mixture $\Pcal = \sum_i \frac{1}{t}\cdot \Pcal_i$, we have
$$
\tvdist{\Pcal-\Qcal}\ge1-t\cdot\eps.
$$
\end{lemma}

\paragraph{Weight Distributions and Symmetrization.}

If $\Pcal$ is a distribution over $\bin^n$, we use $|\Pcal|$ to denote the distribution over weights. That is, $|\Pcal|(w) = \sum_{x : |x|=w}\Pcal(x)$.
We additionally define the symmetrized distribution $\Pcal_\sym$ to be the distribution resulting from randomly permuting the coordinates of a string $x\sim \Pcal$.

We will require the following lemma, which lets us control the distance between two distributions via the distance between their weight distributions, assuming one distribution is symmetric and the other is close to being symmetric.
\begin{lemma}[{\cite[Lemma 4.8]{kane2024locally2}}]\label{lem:distance_to_sym}
Let $\Pcal$ and $\Qcal$ be two distributions on $\{0,1\}^n$ with $\Qcal$ symmetric. Then
$$
\tvdist{\Pcal-\Qcal} = \Theta(\tvdist{|\Pcal|-|\Qcal|} + \tvdist{\Pcal-\Pcal_\sym}).
$$
\end{lemma}

\paragraph{Binomials and Entropy.}
Let $\Hcal(x)=x\cdot\log\pbra{\frac1x}+(1-x)\cdot\log\pbra{\frac1{1-x}}$ be the binary entropy function.
We will use the following estimate regarding binomial coefficients and the entropy function.

\begin{fact}[{See e.g., \cite[Lemma 17.5.1]{cover2006elements}}]\label{fct:individual_binom}
For $1\le k\le n-1$, we have
$$
    \binom nk \ge \frac{2^{n\cdot\Hcal(k/n)}}{\sqrt{8k(1-k/n)}}.
$$
\end{fact}

For positive integer $n$ and parameter $\gamma\in[0,1]$, define $\Bin(n,\gamma)$ to be the binomial distribution of $n$ bits and bias $\gamma$, i.e., $x\sim\Bin(n,\gamma)$ is a random integer between $0$ and $n$ with probability density function $\binom nx\gamma^{x}(1-\gamma)^{n-x}$.
We need the following standard estimates, the proofs of which can be found in \Cref{sec:app_prelim}.

\begin{fact}\label{fct:bin_interval}
    Let $\gamma \in (0,1)$, $a < b \in \Rbb$, and $n \in \Nbb^{\ge 1}$.
    Then the binomial distribution $\Bin(n, \gamma)$ satisfies
    \[
        \Pr\sbra{\Bin(n, \gamma) \in [a,b]} \le O\pbra{\frac{\lceil b-a \rceil}{\sqrt{\gamma (1-\gamma)n}}}.
    \]
\end{fact}

\begin{fact}\label{clm:bin_difference}
    Let $\gamma \in (0,1)$, $a,b \in \Nbb$, and $n \in \Nbb^{\ge 1}$.
    Then the binomial distribution $\Bin(n, \gamma)$ satisfies
    \[
        |\Pr\sbra{\Bin(n,\gamma) = a} - \Pr\sbra{\Bin(n,\gamma) = b}| \le O\pbra{\frac{|b-a|}{\gamma(1-\gamma)n}}.
    \]
\end{fact}

\paragraph{Concentration and Anti-concentration.}
We need the following standard (anti-)concentration bounds.

\begin{fact}[Hoeffding's Inequality]\label{fct:hoeffding}
Assume $X_1,\ldots,X_n$ are independent random variables such that $a\le X_i\le b$ holds for all $i\in[n]$.
Then for all $\delta\ge0$, we have
$$
\max\cbra{\Pr\sbra{\frac1n\sum_{i\in[n]}\pbra{X_i-\E[X_i]}\ge\delta},\Pr\sbra{\frac1n\sum_{i\in[n]}\pbra{X_i-\E[X_i]}\le-\delta}}
\le\exp\cbra{-\frac{2n\delta^2}{(b-a)^2}}.
$$
\end{fact}

\begin{fact}[Chernoff's Inequality]\label{fct:chernoff}
Assume $X_1,\ldots,X_n$ are independent random variables such that $X_i\in[0,1]$ holds for all $i\in[n]$.
Let $\mu=\sum_{i\in[n]}\E[X_i]$.
Then for all $\delta\in[0,1]$, we have
$$
\Pr\sbra{\sum_{i\in[n]}X_i\le(1-\delta)\mu}
\le\exp\cbra{-\frac{\delta^2\mu}2}.
$$
\end{fact}

\begin{fact}[See e.g., {\cite[Lemma 2.9]{dinur2006fourier}}]\label{fct:low-deg_beyond_exp}
Assume $f\colon\bin^n\to\mathbb R$ is a degree $k$ polynomial. Let $\mu=\E_{x\sim\bin^n}[f(x)]$.
Then
$$
\Pr_{x\sim\bin^n}\sbra{f(x)\ge\mu}\ge2^{-O(k)}.
$$
\end{fact}

Recall that random variables $X_1, \dots, X_n$ over some domain $D$ are called \emph{$k$-wise independent} if for any values $d_1,\dots, d_k \in D$ and distinct indices $i_1, \dots, i_k \in [n]$, we have
\[
    \Pr\sbra{X_{i_1} = d_1, \dots, X_{i_k} = d_k} = \prod_{j=1}^k \Pr\sbra{X_{i_j} = d_j}.
\]

\begin{fact}[See e.g., {\cite[Lemma 2.2]{bellare1994randomness}}]\label{fct:k-moments}
    Let $k \ge 4$ be an even integer.
    Suppose $X_1, \dots, X_n$ are $k$-wise independent random variables taking values in $[0,1]$.
    Let $X = X_1 + \cdots + X_n$ and $t > 0$.
    Then,
    \[
        \Pr\sbra{|X - \E[X]| \ge t} \le C_k \cdot \pbra{\frac{nk}{t^2}}^{k/2},
    \]
    where $C_k = 2\sqrt{\pi k}\cdot e^{1/(6k)} \cdot e^{-k/2} \le 1.0004$.
\end{fact}

\section{The Characterization}\label{sec:characterize}

In this section, we prove our main result.
Recall $\Deven$ and $\Dodd$ denote the uniform distribution over $n$-bit strings of even Hamming weight and odd Hamming weight, respectively.

\begin{restatable}{theorem}{thmclassificationfull}\label{thm:main}
    Let $f\colon\bin^m\to\bin^n$ be a $d$-local function, and let $\eps \in (0,1)$.
    Assume $f(\Ucal^m)$ is $\eps$-close to a symmetric distribution $\Dcal$ over $\bin^n$.
    Then if $n$ is sufficiently large in terms of $d$ and $\eps$, $f(\Ucal^m)$ is $O_d\pbra{\frac{1}{\log(1/\eps)}}^{1/5}$-close to a distribution of the form
    \[
        \sum_{\substack{a\in [0,2^d] \cap \Zbb \\ 
        a \ne 2^{d-1}}} c_a\cdot\Ucal_{a/2^d}^n + c_e \cdot \Deven + c_o\cdot \Dodd,
    \]
    where each $c_a = c_a' / 2^C$ for some integer $0 \le c_a' \le 2^C$ and a fixed integer $C = O_d(1)$.
    Moreover, there exist at most $2^C$ many degree-$d$ $\Fbb_2$-polynomials $\{p_i \colon \Fbb_2^m \to \Fbb_2\}$, each with $O_d(n)$ monomials, such that 
    \[
        c_e = \frac{1}{2^C}\cdot \sum_i \Pr_{x\sim \Ucal^m}\sbra{p_i(x) = 0} \quad\text{and}\quad c_o = \frac{1}{2^C}\cdot \sum_i \Pr_{x\sim \Ucal^m}\sbra{p_i(x) = 1}.
    \]
\end{restatable}

Before proceeding to the proof, we make several remarks.

\begin{remark}\label{rmk:NC0_can_sample}
    Any distribution of this form can be exactly produced by an $\nc$ function (with additional input bits and locality):
    \begin{itemize}
        \item Use $C$ bits of locality to select either a product distribution $\Ucal_{a/2^d}^n$ or one of the $\Fbb_2$-polynomials $p_i \colon \Fbb_2^m \to \Fbb_2$.

        \item If some $\Ucal_{a/2^d}^n$ is selected, we can sample from it with an additional $d$ bits of locality.

        \item Otherwise an $\Fbb_2$-polynomial $p_i \colon \Fbb_2^m \to \Fbb_2$ is selected.
        In this case, we wish to produce the distributions $\Deven$ and $\Dodd$ with probability $\Pr_{x\sim \Ucal^m}\sbra{p_i(x) = 0}$ and $\Pr_{x\sim \Ucal^m}\sbra{p_i(x) = 1}$, respectively.
        Arbitrarily partition the $O_d(n)$ monomials of $p_i$ into $n$ bins of size $O_d(1)$, and set $y \in \bin^n$ to have the $j$-th coordinate equal to the sum of the monomials in the $j$-th bin.
        Since $p_i$ has degree $d$, this can be done with $O_d(1)$ bits of locality.
        Note that $y$ has the right weight distribution, but it may not be symmetric.
        To remedy this, we sample $z \sim \Deven$ (with fresh randomness) and output $y \oplus z$.
    \end{itemize}
\end{remark}

\begin{remark}
    An alternative formulation of the last conclusion of \Cref{thm:main} is that there exists a degree-$O_d(1)$ $\Fbb_2$-polynomial $P(x, y) = \sum_{i} \mathbbm{1}(x = i)\cdot p_i(y)$\footnote{In an abuse of notation, we identify an integer $i$ with its binary representation.} with $O_d(n)$ monomials such that $c_e = 2^{-C}\cdot \Pr_z\sbra{P(z) = 0}$ and $c_o = 2^{-C}\cdot \Pr_z\sbra{P(z) = 1}$.
    In this formulation, we can still exactly produce distributions of this form via a similar algorithm to the one in \Cref{rmk:NC0_can_sample}, only now each output bit requires larger locality to compensate for $P$'s larger degree.
\end{remark}

\begin{remark}
    We have chosen to focus on the most commonly studied setting where the random bits fed to $f$ are unbiased.
    For readers interested in more general input biases, we note that a similar result to \Cref{thm:main} (with the biases of the product distributions and the mixing weights appropriately adjusted) should be provable using the techniques presented here.
    It is important, however, that the input bits are identically distributed, as the first of our four main steps (see \Cref{subsec:matching_moments}) requires the possible output biases to lie in a discrete set. 
\end{remark}

Our proof will proceed via the four steps outlined in \Cref{sec:overview}, each corresponding to its own subsection.

\subsection{Removing Large Influences}\label{subsec:matching_moments}

Our first step is to argue that after conditioning on the high degree input bits of $f$, almost any small collection of output bits looks identical to those of a $\gamma$-biased distribution $\Ucal_\gamma^n$, where $\gamma$ is an integer multiple of $2^{-d}$.

\begin{restatable}{proposition}{propindependenceaftercond}\label{prop:independence_after_cond}
Let $f\colon\bin^m\to\bin^n$ be a $d$-local function, and let $A \ge 2^{100d}$ be a parameter.
Assume $f(\Ucal^m)$ is $\eps$-close to a symmetric distribution over $\bin^n$ for some $\eps < 2^{-cdA}$, where $c > 0$ is a sufficiently large absolute constant.
Further assume that $n$ is sufficiently large in terms of $d,k,A,\eps$.
Define $S\subseteq[m]$ to be the set of input bits with degree at least $n/A$.

Let $k \le \log(1/\eps)/C_d$ be an arbitrary integer, where $C_d\ge1$ is a sufficiently large constant depending only on $d$.
Then for each conditioning $\rho\in\bin^S$ on the bits in $S$, there exists a subset $T_\rho\subseteq[n]$ of size $|T_\rho|\le O_{d,k,A}(1)$ such that every $k$-tuple of output bits in $[n]\setminus T_\rho$ has distribution $\Ucal_{\gamma_\rho}^k$, where $\gamma_\rho=a_\rho/2^d$ and $0 \le a_\rho \le 2^d$ is an integer.
\end{restatable}

In our analysis, we will often need to consider the distance between some bias and its closest integer multiple of $2^{-d}$, so we introduce the following notation.

\begin{definition}[Binary Representation Error]\label{def:abs_binary_rep_err}
For each $d\in\Nbb$, we use $\err(\gamma,d)$ to denote the minimum distance of $\gamma$ to an integer multiple of $2^{-d}$.
In particular, given a binary representation of $\gamma$ as $\gamma=\sum_{i\in\Zbb}a_i\cdot 2^i$ where each $a_i\in\bin$, we have
$$
\err(\gamma,d)=\min\cbra{\sum_{i<-d}a_i\cdot2^i,\sum_{i<-d}(1-a_i)\cdot2^i}.
$$
\end{definition}

The proof of \Cref{prop:independence_after_cond} involves a number of similar looking estimates, so we provide a brief overview of the remainder of the section before continuing.
It is known from previous work \cite{kane2024locality} that local functions cannot accurately sample $\Dcal_k$, the uniform distribution over $n$-bit binary strings of Hamming weight $k$, so long as $k/n$ has large binary representation error.

\begin{lemma}[{\cite[Theorem 5.7]{kane2024locality}}]\label{lem:locality_single_non-dyadic}
Let $f\colon\bin^m\to\bin^n$ be a $d$-local function, and let $1\le k\le n-1$ be an integer.
If $\err(k/n,d)\ge\delta$ for some $\delta>0$, then
$$
\tvdist{f(\Ucal^m)-\Dcal_k}\ge1-4\sqrt{2n}\cdot\exp\cbra{-n\cdot\delta^{40d}}.
$$
\end{lemma}

This implies that with high probability, the output weight $|f(x)|$ is close to \emph{some} dyadic rational multiple of $n$, at least to the degree that $f$ is symmetric (\Cref{lem:dyadic_weight_eps}).
In order to ensure $|f(x)|$ is concentrated around a \emph{fixed} dyadic rational, we condition on input bits of degree at least $n/A$.
This bounds the variance of the weight of $f$ to provide good concentration (\Cref{lem:dyadic_weight_after_cond}).
However, there is still a chance that some non-negligible fraction of the time, the output weight is close to a different dyadic rational multiple of $n$.
In this case, we can show (\Cref{clm:lem:dyadic_weight_after_cond_5}) that the weight distribution must also assign decent probability mass to the weights between these dyadic multiples, which by \Cref{lem:locality_single_non-dyadic} would contradict our original assumption on the distance between $f(\Ucal^m)$ and $\Dcal$.

Now we proceed toward proving \Cref{prop:independence_after_cond}. 
Recall that any symmetric distribution $\Dcal$ is simply a mixture of $\Dcal_k$ for different values of $k$.
Thus, if $f(\Ucal^m)$ is close to a symmetric distribution, \Cref{lem:locality_single_non-dyadic} implies most of the output weight must have bias close to some multiple of $2^{-d}$.

\begin{lemma}\label{lem:dyadic_weight_eps}
Let $f\colon\bin^m\to\bin^n$ be a $d$-local function.
Assume $f(\Ucal^m)$ is $\eps$-close to a symmetric distribution $\Dcal$ over $\bin^n$.
Then
\begin{align}\label{eq:lem:dyadic_weight_eps_1}
\Pr_{x\sim\Ucal^m}\sbra{\err\pbra{\frac{\abs{f(x)}}n,d}\ge\frac1{n^{1/(800d)}}}\le O\pbra{\eps+e^{-n^{0.9}}},
\end{align}
where we recall that $\err(\gamma,d)$ is the minimum distance between $\gamma$ and integer multiples of $2^{-d}$.
\end{lemma}
\begin{proof}
Without loss of generality, assume $n$ is sufficiently large.
If
\begin{align}\label{eq:lem:dyadic_weight_eps_2}
\Pr_{z\sim\Dcal}\sbra{\err\pbra{\frac {|z|}n,d}\ge\frac1{n^{1/(800d)}}}\le 100\cdot\pbra{\eps+e^{-n^{0.9}}},
\end{align}
then \Cref{eq:lem:dyadic_weight_eps_1} holds due to the assumption on $f(\Ucal^m)$.
Now we assume \Cref{eq:lem:dyadic_weight_eps_2} does not hold.

For each $0\le k\le n$, recall that $\Dcal_k$ is the uniform distribution over Hamming weight $k$ strings.
Then $\Dcal$ is a mixture of $\cbra{\Dcal_k}$, i.e., $\Dcal=\sum_k\alpha_k\cdot\Dcal_k$.
We say $k$ is \emph{bad} if $\err(k/n,d)\ge n^{-1/(800d)}$.
Then the violation of \Cref{eq:lem:dyadic_weight_eps_2} is equivalent to
\begin{align}\label{eq:lem:dyadic_weight_eps_3}
\sum_{\text{bad }k}\alpha_k>100\cdot\pbra{\eps+e^{-n^{0.9}}}.
\end{align}
By \Cref{lem:locality_single_non-dyadic}, for each bad $k$, there is an event $\Ecal_k$ such that
\begin{itemize}
\item under $f(\Ucal^m)$, it happens with probability at most $4\sqrt{2n}\cdot e^{-n^{0.95}}$;
\item under $\Dcal_k$, it happens with probability at least $1-4\sqrt{2n}\cdot e^{-n^{0.95}}\ge1/2$.
\end{itemize}
Hence considering $\Ecal=\bigvee_{\text{bad }k}\Ecal_k$, we have
\begin{itemize}
\item under $f(\Ucal^m)$, it happens with probability at most $4n\sqrt{2n}\cdot e^{-n^{0.95}}$ which is at most $e^{-n^{0.9}}$, since we assumed $n$ is sufficiently large;
\item under $\Dcal$, it happens with probability at least 
$$
\sum_{\text{bad }k}\alpha_k\cdot\pbra{1-4\sqrt{2n}\cdot e^{-n^{0.95}}}\ge50\cdot\pbra{\eps+e^{-n^{0.9}}}.
$$
\end{itemize}
This means $f(\Ucal^m)$ is not $\eps$-close to $\Dcal$, a contradiction.
\end{proof}

By conditioning on the high degree input bits, we can reduce the variance of the output weight distribution to obtain a version of \Cref{lem:dyadic_weight_eps} where the output weight is concentrated around a \emph{fixed} dyadic rational multiple of $n$.

\begin{lemma}\label{lem:dyadic_weight_after_cond}
Let $f\colon\bin^m\to\bin^n$ be a $d$-local function with $d \ge 1$, and let $A \ge 2^{100d}$ be a parameter.
Assume $f(\Ucal^m)$ is $\eps$-close to a symmetric distribution over $\bin^n$ for some $\eps < 2^{-cdA}$, where $c > 0$ is a sufficiently large absolute constant.
Further assume that $n$ is sufficiently large in terms of $d, A, \eps$.
Define $S\subseteq[m]$ to be the set of input bits with degree at least $n/A$.
Then for each conditioning $\rho\in\bin^S$ on bits in $S$, there exists an integer $0\le a_\rho\le2^d$ such that
$$
\Pr_{x\sim\Ucal^{[m]\setminus S}}\sbra{\abs{\frac{\abs{f(x,\rho)}}n-\frac{a_\rho}{2^d}}\ge\frac1{n^{1/(800d)}}}\le \poly(\eps).
$$
\end{lemma}

For clarity, we prove \Cref{lem:dyadic_weight_after_cond} through a series of claims, the most routine of which have their proofs deferred to \Cref{app:missing_sec:characterize}.
The high-level idea is to use the second moment method to show that for any restriction $\rho$ on the high degree input bits, the Hamming weight of $f(x,\rho)$ is typically close to its expectation, which by the previous lemma must be close to a multiple of $2^{-d}$.
We then turn to a more involved argument to boost the quantitative behavior of these bounds.

\begin{proof} 
First note that $|S|\le d A$.
Fix an arbitrary $\rho\in\bin^S$ and define $g\colon\bin^{m-|S|}\to\cbra{0,1,\dots,n}$ by $g(x)=\abs{f(x,\rho)}$.
Then $g(x)=\sum_{i=1}^ng_i(x)$, where each $g_i\colon\bin^{m-|S|}\to\bin$ is a $d$-junta (i.e., depends on at most $d$ of its input bits) and every input bit appears in at most $n/A$ different $g_i$'s.
Therefore
\begin{align}
\Var[g]
&=\sum_{i,j\in[n]}\Cov\pbra{g_i,g_j}
\notag\\
&\le\sum_{i\in[n]}\#\cbra{j\in[n]\mid g_j\text{ correlates with }g_i}
\tag{since each $g_i$ is Boolean}\\
&\le n\cdot dn/A
\le dn^2/2^{100d}. \tag{since $A \ge 2^{100d}$}
\end{align}
Define $p=\E[g]/n$.
Then by Chebyshev's inequality, we have
\begin{align}\label{eq:lem:dyadic_weight_after_cond_2}
\Pr_{x\sim\bin^{[m]\setminus S}}\sbra{\abs{\frac{|f(x,\rho)|}n-p}>2^{-30d}}=\Pr\sbra{\abs{\frac{g(x)}n-p}>2^{-30d}}\le\frac{d}{2^{40d}}\le\frac14.
\end{align}

Combining the above fact that $|f(x,\rho)|/n$ is typically close to $p$ with the fact that it must also typically be close to an integer multiple of $2^{-d}$ (by \Cref{lem:dyadic_weight_eps}), we obtain the following claim.

\begin{claim}[Proved in \Cref{app:missing_sec:characterize}]\label{clm:lem:dyadic_weight_after_cond_1}
$\err(p,d)\le 2\cdot 2^{-30d}$.
\end{claim}

By \Cref{clm:lem:dyadic_weight_after_cond_1}, there exists an integer $0\le a\le2^d$ such that $|p-a/2^d|\le2\cdot 2^{-30d}$.
Now it suffices to show
\begin{align}\label{eq:lem:dyadic_weight_after_cond_3}
\Pr_{x\sim\bin^{[m]\setminus S}}\sbra{\abs{\frac{|f(x,\rho)|}n-\frac a{2^d}}>\frac1{n^{1/(800d)}}}\le \poly(\eps).
\end{align}
We start with two primitive bounds.
The first follows from combining \Cref{lem:dyadic_weight_eps} with the observation that any event is assigned at most $2^{|S|}$ times more mass by $f(\Ucal^{[m]\setminus S}, \rho)$ than by $f(\Ucal^m)$.

\begin{claim}[Proved in \Cref{app:missing_sec:characterize}]\label{clm:lem:dyadic_weight_after_cond_2}
$$
\Pr_{x\sim\bin^{[m]\setminus S}}\sbra{\err\pbra{\frac{|f(x,\rho)|}n,d}>\frac1{n^{1/(800d)}}}\le \poly(\eps).
$$
\end{claim}

\begin{claim}\label{clm:lem:dyadic_weight_after_cond_3}
$$
\Pr_{x\sim\bin^{[m]\setminus S}}\sbra{\abs{\frac{|f(x,\rho)|}n-\frac a{2^d}}>\frac{3}{2^{30d}}}\le\frac14.
$$
\end{claim}
\begin{proof}[Proof of \Cref{clm:lem:dyadic_weight_after_cond_3}]
This follows directly from \Cref{eq:lem:dyadic_weight_after_cond_2} and \Cref{clm:lem:dyadic_weight_after_cond_1}.
\end{proof}

Define
$$
\delta\coloneqq \Pr_{x\sim\bin^{[m]\setminus S}}\sbra{\abs{\frac{|f(x,\rho)|}n-\frac a{2^d}}>4^{-d}}.
$$
Then by \Cref{clm:lem:dyadic_weight_after_cond_2}, we can relate the LHS of \Cref{eq:lem:dyadic_weight_after_cond_3} to $\delta$, because it is very unlikely that $|f(x,\rho)|/n$ is between $n^{-1/(800d)}$- and $4^{-d}$-close to $a/2^d$.

\begin{claim}[Proved in \Cref{app:missing_sec:characterize}]\label{clm:lem:dyadic_weight_after_cond_4}
$$
\Pr_{x\sim\bin^{[m]\setminus S}}\sbra{\abs{\frac{|f(x,\rho)|}n-\frac a{2^d}}>\frac1{n^{1/(800d)}}}\le \delta+\poly(\eps).
$$
\end{claim}

\Cref{clm:lem:dyadic_weight_after_cond_3} implies a constant $1/4$ upper bound on $\delta$.
To further improve it, we prove the following claim.
\begin{claim}\label{clm:lem:dyadic_weight_after_cond_5}
$$
\Pr_{x\sim\bin^{[m]\setminus S}}\sbra{8^{-d}\le\abs{\frac{|f(x,\rho)|}n-\frac a{2^d}}\le4^{-d}}\ge\delta/2^{20d}.
$$
\end{claim}

Such a claim is true, because if $|f(x,\rho)|/n$ is typically close to $a/2^d$ but has a $\delta$ probability of being far from it, then there must be a decent probability that $|f(x,\rho)|/n$ is close, but not too close, to $a/2^d$, since $|f(x,\rho)|/n$ is roughly continuous in $x$.
Formally, the proof of \Cref{clm:lem:dyadic_weight_after_cond_5} will operate via a coupling argument.
We consider two independent inputs $x$ and $z$ where $|f(z,\rho)|/n$ is close to $a/2^d$ but $|f(x,\rho)|/n$ is not.
By slowly changing $x$ into $z$ by flipping bits of $x$ on which they disagree in a random order, we can find inputs $y$ where $|f(y,\rho)|/n$ is in the range in question.

\begin{proof}[Proof of \Cref{clm:lem:dyadic_weight_after_cond_5}]
Let $\pi$ be a uniformly random permutation on $[m]\setminus S$.
Pick $t$ uniformly at random among $0,1,\ldots,m-|S|$ and sample $r\sim\mathrm{Bin}(m-|S|,1/2)$ (i.e., $\Pr[r=s]=2^{|S|-m}\binom{m-|S|}s$ for all $s=0,1,\ldots,m$).
Define $y\in\bin^{[m]\setminus S}$ as $x$ with the first $t$ bits in $\pi$ flipped; and define $z\in\bin^{[m]\setminus S}$ as $x$ with the first $r$ bits in $\pi$ flipped.
Since $y$ has the same distribution as $x$, it suffices to show
\begin{align}\label{eq:clm:lem:dyadic_weight_after_cond_5_1}
\Pr_{x,z,\pi,t}\sbra{8^{-d}\le\abs{\frac{|f(y,\rho)|}n-\frac a{2^d}}\le4^{-d}}\ge\delta/2^{20d}.
\end{align}

Observe that $z$ is uniform over $\bin^{[m]\setminus S}$ and is independent of $x$.
Define $\Ecal_x$ to be the event that $\abs{\frac{|f(x,\rho)|}n-\frac a{2^d}}>4^{-d}$ and $\Ecal_z$ to be the event that $\abs{\frac{|f(z,\rho)|}n-\frac a{2^d}}\le3/2^{30d}$.
Then by independence and \Cref{clm:lem:dyadic_weight_after_cond_3}, we have
\begin{align}\label{eq:clm:lem:dyadic_weight_after_cond_5_2}
\Pr_{x,z}\sbra{\Ecal_z\mid\Ecal_x}\ge\frac34.
\end{align}

Recall we wish to show that if $\Ecal_x$ and $\Ecal_z$ both hold, there is a good probability (over $t$) that $|f(y,\rho)|/n$ is between $8^{-d}$- and $4^{-d}$-close to $a/2^d$.
Note that changing $t$ by 1 only changes $|f(y,\rho)|$ by at most the degree of the relevant input.
Since $|f(y,\rho)|$ must pass through the ``bad'' region, \Cref{eq:clm:lem:dyadic_weight_after_cond_5_1} holds as long as it does not pass through too quickly.
This is captured by the following event $\Ecal_\pi$.

Denote $m'=m-|S|$ and for each $j\in[m']$ we use $\deg_\pi(j)$ to denote the degree of the $j$-th input bit under $\pi$.
Define $L=\floorbra{m'/2^{10d}}$.
Let $\Ecal_\pi$ be the event that no $L$ consecutive (under $\pi$) input bits starting at a multiple of $L$ have degree sum larger than $n/2^{5d}$; or more formally that $\sum_{j=1}^L\deg_\pi(L\cdot i+j)\le n/2^{5d}$ holds for each $i=0,1,\ldots,\floorbra{m'/L}$.
Whenever $\Ecal_\pi$ holds, we know that any $L$ consecutive (under $\pi$) bit flips of $x$ will only change the output weight of $f$ by at most $2\cdot n/2^{5d}$, since any length $L$ interval can only intersect two length $L$ intervals that start at a multiple of $L$ (as in the definition of $\Ecal_\pi$).
For $i=0$, we have 
$$
\E_\pi\sbra{\sum_{j=1}^L\deg_\pi(j)}=L\cdot\E_\pi\sbra{\deg_\pi(1)}\le \frac{dnL}{m'}\le\frac{dn}{2^{10d}}
$$ 
as the total degree is at most $dn$.
Define the indicator variable $I_i$ to denote that the $i$-th input bit (in the standard order) is in $\pi(1),\ldots,\pi(L)$. Let $\deg(i)$ be the degree of the $i$-th input bit (in the standard order).
Then we also have
\begin{align*}
\Var_\pi\sbra{\sum_{j=1}^L\deg_\pi(j)}
&=\Var_\pi\sbra{\sum_{i=1}^{m'}\deg(i)\cdot I_i}
=\sum_i\deg(i)^2\Var_\pi\sbra{I_i}+\sum_{i\neq i'}\deg(i)\deg(i')\Cov\sbra{I_i,I_{i'}}\\
&\le\sum_i\deg(i)^2\cdot\Var_\pi\sbra{I_i}
\tag{since $\Cov\sbra{I_i,I_{i'}}=\frac{L(L-1)}{m'(m'-1)}-\pbra{\frac L{m'}}^2\le0$}\\
&\le \frac L{m'}\sum_i\deg(i)^2 \tag{since $\Var_\pi\sbra{I_i} \le \Pr_\pi\sbra{I_i = 1}$} \\
&\le \frac L{m'}\cdot \max_i \deg(i) \cdot \sum_i\deg(i) \le\frac L{m'}\cdot\frac n{2^{100d}}\cdot dn,
\end{align*}
since $\deg(i)\le n/A \le n/2^{100d}$ and $\sum_i\deg(i) \le dn$.
Therefore
\begin{align*}
\Pr_\pi\sbra{\sum_{j=1}^L\deg_\pi(j)\ge n/2^{5d}}
&\le\Pr_\pi\sbra{\sum_{j=1}^L\deg_\pi(j)-\E\sbra{\sum_{j=1}^L\deg_\pi(j)}\ge n/2^{5d+1}}
\tag{since $2^{-5d}-d\cdot2^{-10d}\ge2^{-5d}/2$ for $d\ge1$}\\
&\le\frac{\Var_\pi\sbra{\sum_{j=1}^L\deg_\pi(j)}}{4n^2\cdot2^{-10d}}
\tag{by Chebyshev inequality}\\
&\le\frac{d\cdot L}{4\cdot 2^{90d} \cdot m'}.
\end{align*}
Note that the same estimate works for all $i=1,\ldots,\floorbra{m'/L}$.
Hence by independence and a union bound
\begin{align}
\Pr_{x,\pi}\sbra{\Ecal_\pi\mid\Ecal_x}
&=\Pr_\pi\sbra{\Ecal_\pi}
\ge1-\pbra{1+\floorbra{\frac{m'}L}}\cdot\frac{d L}{4\cdot 2^{90d} \cdot m'}
\ge 1 - \frac{d}{2^{20d}}
\ge\frac78.
\label{eq:clm:lem:dyadic_weight_after_cond_5_3}
\end{align}

Combining \Cref{eq:clm:lem:dyadic_weight_after_cond_5_2} and \Cref{eq:clm:lem:dyadic_weight_after_cond_5_3}, we have $\Pr_{x,z,\pi}\sbra{\Ecal_\pi\land\Ecal_z\mid\Ecal_x}\ge5/8$ by a union bound.
Since $\Pr_x\sbra{\Ecal_x}=\delta$, we have
\begin{align}\label{eq:clm:lem:dyadic_weight_after_cond_5_4}
\Pr_{x,z,\pi}\sbra{\Ecal_x\land\Ecal_z\land\Ecal_\pi}\ge\frac{5\delta}8.
\end{align}

Recall that $z$ is $x$ with the first $r$ bits (in the order of $\pi$) flipped.
Assuming $\Ecal_x$ and $\Ecal_z$, we know that $\abs{\frac{|f(x,\rho)|}n-\frac a{2^d}}>4^{-d}$ and $\abs{\frac{|f(z,\rho)|}n-\frac a{2^d}}\le3/2^{30d}$.
For each $i=0,1,\ldots,m'$, we use $z^{(i)}$ to denote the string $x$ with the first $i$ bits (in the order of $\pi$) flipped.
Then $z^{(0)}=x$ and $z^{(r)}=z$.
Note that the Hamming difference between $f(z^{(i)},\rho)$ and $f(z^{(i-1)},\rho)$ is upper bounded by the degree of the $i$-th input bit (under $\pi$), which is further upper bounded by $n/2^{100d}$ by our construction.
Since $4^{-d}>6^{-d}>3/2^{30d}$, there exists some $i^*\in\cbra{0,1,\ldots,r}$ such that
$$
6^{-d}\le\abs{\frac{\abs{f(z^{(i^*)},\rho)}}n-\frac a{2^d}}\le6^{-d}+2^{-100d}.
$$
Assuming $\Ecal_\pi$, for each $j=i^*,i^*-1,\ldots,i^*-L$ the Hamming difference between $f(z^{(j)},\rho)$ and $f(z^{(i^*)},\rho)$ is upper bounded by the total degree of input bits from $i^*$ to $i^*-L$, which is at most $2\cdot n/2^{5d}$.
Hence for each such $j$, we have
$$
8^{-d}\le6^{-d}-2^{-5d+1}\le\abs{\frac{\abs{f(z^{(j)},\rho)}}n-\frac a{2^d}}\le6^{-d}+2^{-100d}+2^{-5d+1}\le4^{-d}.
$$
In particular, this implies $i^*-L>0$ as $\abs{\frac{|f(z^{(0)},\rho)|}n-\frac a{2^d}}=\abs{\frac{|f(x,\rho)|}n-\frac a{2^d}}>4^{-d}$.
At this point, recall that $y=z^{(t)}$ where $t\sim\cbra{0,1,\ldots,m'}$.
Hence
\begin{align*}
&\phantom{\ge}\Pr_{x,z,\pi,t}\sbra{8^{-d}\le\abs{\frac{\abs{f(y,\rho)}}n-\frac a{2^d}}\le4^{-d}\mid\Ecal_x\land\Ecal_z\land\Ecal_\pi}\\
&\ge\Pr_{x,z,\pi,t}\sbra{i^*-L\le t\le i^*\mid\Ecal_x\land\Ecal_z\land\Ecal_\pi}\\
&=\frac{L+1}{m'}\ge2^{-10d}.
\tag{by independence}
\end{align*}
This, combined with \Cref{eq:clm:lem:dyadic_weight_after_cond_5_4}, proves \Cref{eq:clm:lem:dyadic_weight_after_cond_5_1} and completes the proof of \Cref{clm:lem:dyadic_weight_after_cond_5}.
\end{proof}

Given \Cref{clm:lem:dyadic_weight_after_cond_5}, we can put a much tighter upper bound on $\delta$, as demanded in \Cref{clm:lem:dyadic_weight_after_cond_4}.
Observe that $8^{-d}\le\abs{\frac{|f(x,\rho)|}n-\frac a{2^d}}\le4^{-d}$ implies that $\abs{\frac{|f(x,\rho)|}n-\frac{a'}{2^d}}\ge2^{-d}-4^{-d}\ge8^{-d}$ holds for any integer $a'\ne a$.
Hence when the event in \Cref{clm:lem:dyadic_weight_after_cond_5} happens, we have
$$
\err\pbra{\frac{|f(x,\rho)|}n,d}\ge8^{-d}.
$$
\Cref{clm:lem:dyadic_weight_after_cond_2} implies this should happen with probability at most $\poly(\eps)$.
Hence by \Cref{clm:lem:dyadic_weight_after_cond_5}, this means $\delta\le \poly(\eps)$.
Therefore \Cref{eq:lem:dyadic_weight_after_cond_3} follows directly by \Cref{clm:lem:dyadic_weight_after_cond_4}. 
This completes the proof of \Cref{lem:dyadic_weight_after_cond}.
\end{proof}

We now have the requisite tools to prove \Cref{prop:independence_after_cond}, the main result of \Cref{subsec:matching_moments}.
We restate it below for convenience.

\propindependenceaftercond*

\begin{proof}
Fix an arbitrary $\rho$ and shorthand $T=T_\rho,\gamma=\gamma_\rho$.
Assume there is a maximal set of $R$ disjoint $k$-tuples of output bits with the wrong distribution.
Then it suffices to show $R=O_{d,k,A}(1)$, since we can set $T$ to be the union of these output bits, and every $k$-tuple that does not have distribution $\Ucal_{\gamma}^k$ must intersect $T$ by maximality.
Moreover, $|T|\le k\cdot R=O_{d,k,A}(1)$.

Now for the $i$-th $k$-tuple with the wrong distribution, we have a degree-$k$ multilinear polynomial $P_i\colon\bin^n\to\bin$ such that 
\begin{align}\label{eq:lem:independence_after_cond_1}
\E_{x\sim\bin^{[m]\setminus S}}\sbra{P_i(f(x,\rho))}\ge\E_{z\sim\Ucal_\gamma^n}\sbra{P_i(z)}+2^{-kd}.
\end{align}
To see this, let $W\subseteq\bin^k$ witness the total variation distance $\delta$ between $\Ucal_\gamma^k$ and the $k$-tuple output in $f(\Ucal^{[m]\setminus S},\rho)$, i.e.,
\begin{align}\label{eq:lem:independence_after_cond_2}
\Pr_{x\sim\bin^{[m]\setminus S}}\sbra{\text{the $k$-tuple output of }f(x,\rho)\in W}\ge\Pr_{z\sim\Ucal_\gamma^k}\sbra{z\in W}+\delta.
\end{align}
Then we define the polynomial $P_i\colon\bin^n\to\bin$ as the indicator function that the $k$-tuple output lies in $W$. This is a $k$-junta and is naturally of degree $k$.
Since $\gamma$ is an integer multiple of $2^{-d}$, we know that the probability density function of $\Ucal_\gamma^k$ has granularity $2^{-kd}$.
In addition, since $f$ is $d$-local, the probability density function of the $k$-tuple output of $f(\Ucal^{[m]\setminus S},\rho)$ also has granularity $2^{-kd}$.
As the two distributions are different, their total variation distance is $\delta\ge2^{-kd}$.
Putting this into \Cref{eq:lem:independence_after_cond_2} gives \Cref{eq:lem:independence_after_cond_1}.

Given the constructions of $P_i$'s, we define $P=\sum_{i\in[R]}P_i$.
Then by \Cref{eq:lem:independence_after_cond_1}, we have
\begin{align*}
\E_{x\sim\bin^{[m]\setminus S}}\sbra{P(f(x,\rho))}\ge\E_{z\sim\Ucal_\gamma^n}\sbra{P(z)}+2^{-kd}\cdot R.
\end{align*}
Since each $P_i$ is a degree $k$ polynomial of the output bits of $f$ and $f$ is $d$-local, $P$ is a degree $k$ polynomial of the output bits of $f$, and $P(f(x,\rho))$ is degree $k\cdot d$ polynomial of $x$.
By \Cref{fct:low-deg_beyond_exp}, with probability at least $2^{-O(kd)}$ over $x\sim\bin^{[m]\setminus S}$, we have
$$
P(f(x,\rho))\ge\E_x\sbra{P(f(x,\rho))}\ge\E_{z\sim\Ucal_\gamma^n}\sbra{P(z)}+2^{-kd}\cdot R.
$$
Thus by \Cref{lem:dyadic_weight_after_cond} and a union bound, with probability at least $2^{-O(kd)}-\poly(\eps)$, we have
$$
P(f(x,\rho))\ge\E_{z\sim\Ucal_\gamma^n}\sbra{P(z)}+2^{-kd}\cdot R
\quad\text{and}\quad
\abs{\frac{|f(x,\rho)|}n-\gamma}\le\frac1{n^{1/(800d)}}.
$$
Recall that $|S|\le dA$.
Hence by randomizing also the coordinates in $S$, with probability at least $2^{-dA}\pbra{2^{-O(kd)} - \poly(\eps)}$ over $y\sim\bin^m$, we have
$$
P(f(y))\ge\E_{z\sim\Ucal_\gamma^n}[P(z)]+2^{-kd}\cdot R
\quad\text{and}\quad
\abs{\frac{|f(y)|}n-\gamma}\le\frac1{n^{1/(800d)}}.
$$
Since $f(\Ucal^m)$ is $\eps$-close to a symmetric distribution $\Dcal$ over $\bin^n$, we also have that with probability at least $2^{-dA}\pbra{2^{-O(kd)} - \poly(\eps)} - \eps$ over $w\sim\Dcal$, 
\begin{align}\label{eq:lem:independence_after_cond_4}
P(w)\ge\E_{z\sim\Ucal_\gamma^n}[P(z)]+2^{-kd}\cdot R
\quad\text{and}\quad
\abs{\frac {|w|}{n}-\gamma}\le\frac1{n^{1/(800d)}}.
\end{align}

Call an input $w \in \bin^n$ \emph{large} if 
\[
    P(w)\ge\E_{z\sim\Ucal_\gamma^n}[P(z)]+2^{-kd}\cdot R.
\]
We will show that it is very unlikely for $w \sim \Dcal$ to be large while having weight close to $\gamma n$.
Recall that $T$ is the union of the output bits in each of the $R$ $k$-tuples with the wrong distribution, and define the distribution $\Dcal^*$ to be the marginal distribution of $\Dcal$ on $T$ conditioned on $w \sim \Dcal$ satisfying
\begin{equation}\label{eq:w_bias_close_to_gamma}
    \abs{\frac{|w|}{n}-\gamma}\le\frac1{n^{1/(800d)}}.
\end{equation}
Then by sequentially coupling the $t \coloneqq |T| = Rk$ bits based on their marginals, we have
\begin{align*}
\tvdist{\Ucal_\gamma^T-\Dcal^*}
&\le\sum_{i=1}^t\max_{\substack{0\le j\le i-1 \\ w \in \supp{\Dcal^*}}}\abs{\gamma-\frac{|w|-j}{n-i+1}}
\tag{by union bound}\\
&\le\sum_{i=1}^t\max_{w \in \supp{\Dcal^*}}\pbra{\abs{\gamma-\frac{|w|}{n}}+\abs{\frac{|w|}{n}-\frac{|w|-i+1}{n-i+1}}}\\
&\le\sum_{i=1}^t\pbra{\max_{w \in \supp{\Dcal^*}}\abs{\gamma-\frac{|w|}{n}}+\frac t{n-t+1}}\\
&\le\frac t{n^{1/(800d)}}+\frac{t^2}{n-t+1}.
\tag{by \Cref{eq:w_bias_close_to_gamma}}
\end{align*}
Recall that $P=\sum_{i\in[R]}P_i$ where the $P_i\colon\bin^n\to\bin$ are supported on disjoint sets of $k$-tuples.
Thus \Cref{fct:hoeffding} implies
\[
    \Pr_{x\sim\Ucal_\gamma^n}\sbra{x \text{ is large}} = \Pr_{x\sim\Ucal_\gamma^n}\sbra{\sum_{i \in [R]} P_i(x)-\E_{z\sim\Ucal_\gamma^n}[P_i(z)] \ge 2^{-kd}\cdot R} \le \exp\cbra{-\frac{R}{2^{2kd}}}.
\]
Observing that $\sum_{i} P_i$ only depends on the bits in $T$, we can upper bound the event in \Cref{eq:lem:independence_after_cond_4} by
\begin{align*}
\Pr_{w\sim\Dcal}\sbra{w \text{ is large}
\land
\abs{\frac{|w|}n-\gamma}\le\frac1{n^{1/(800d)}}} &\le
\Pr_{w\sim\Dcal}\sbra{w \text{ is large}
\mid
\abs{\frac{|w|}n-\gamma}\le\frac1{n^{1/(800d)}}}\\
&\le \Pr_{x\sim\Ucal_\gamma^n}\sbra{x \text{ is large}} + \tvdist{\Ucal_\gamma^T-\Dcal^*} \\
&\le \exp\cbra{-\frac{R}{2^{2kd}}} + \frac t{n^{1/(800d)}}+\frac{t^2}{n-t+1}.
\end{align*}
Combining with the lower bound from \Cref{eq:lem:independence_after_cond_4}, we have
$$
\exp\cbra{-\frac{R}{2^{2kd}}} + \frac{Rk}{n^{1/(800d)}}+\frac{(Rk)^2}{n-Rk+1}\ge 2^{-dA}\pbra{2^{-O(kd)} - \poly(\eps)} - \eps.
$$
By our assumptions on the size of $n, k,$ and $\eps$, we have $R\le O_{d,k,A}(1)$ as desired.
This completes the proof of \Cref{prop:independence_after_cond}.
\end{proof}

\subsection{Kolmogorov Distance}\label{subsec:kol_dist}

Our second step toward proving \Cref{thm:main} is to show that the output weight distribution of $f(\Ucal^m)$ is close in Kolmogorov distance\footnote{Recall the \emph{Kolmogorov distance} between two distributions $\Pcal$, $\Qcal$ is given by $\sup_{t \in \Rbb}|\Pr_{x\in \Pcal}[x \ge t] - \Pr_{y\in \Qcal}[y \ge t]|$.} to some binomial distribution $\Bin(n,\gamma)$, where $\gamma$ is an integer multiple of $2^{-d}$.
Moreover in the case of $\gamma = 1/2$, we will show that these distributions are close even accounting for parity.

We will require the fact that biased $k$-wise independence fools a simple type of threshold function.
In particular, we will use the following special case of \cite[Theorem I.5]{gopalan2010fooling}.

\begin{lemma}\label{lem:biased_kol_dist}
    Let $\gamma \in (0,1), k \in \Nbb,$ and $t \in \Rbb^{\ge 0}$. If $\Dcal$ is a distribution on $\bin^n$ such that every $k$-tuple of bits in $[n]$ has distribution $\Ucal_{\gamma}^k$, then
    \[
        \left|\Pr_{x\sim \Dcal}\sbra{x_1 + \cdots + x_n \ge t} - \Pr_{x\sim \Ucal_\gamma^n}\sbra{x_1 + \cdots + x_n \ge t}\right| \le O\pbra{\frac{\polylog(k)}{\sqrt{k}\cdot \poly(\min\cbra{\gamma, 1-\gamma})}}.
    \]
\end{lemma}

We now leverage \Cref{lem:biased_kol_dist} to show the output weight of $f(\Ucal^m)$ is close in Kolmogorov distance to a certain binomial distribution.
Note the following can be viewed as a generalization of Lemma 5.2 in (the arXiv version of) \cite{kane2024locally2} to arbitrary biases and mixtures.
The proof there holds with minor modifications, but we reproduce it for completeness.

\begin{proposition}\label{prop:kol_dist}
    Let $k\ge 2$ and $\ell \ge 1$ be integers, and let $f_1, \dots, f_\ell\colon\bin^m\to\bin^n$ be $d$-local functions.
    For each $f_i$, suppose no input bit affects more than $n/A$ output bits, and suppose there exists a subset $T_i\subseteq[n]$ of size $|T_i|\le O_{d,k,A}(1)$ such that every $k$-tuple of output bits in $[n]\setminus T_i$ has distribution $\Ucal_{\gamma}^k$, where $\gamma=a/2^d$ for some fixed integer $0 < a < 2^d$. 
    If $n$ is sufficiently large in terms of $d$, $k$, and $A$, then any mixture $F$ of the $f_i$'s and any $t \in \Rbb$ satisfy
    \[
        \left|\Pr\sbra{|F(\Ucal^m)| \ge t} - \Pr\sbra{\Bin(n, \gamma) \ge t}\right| \le O\pbra{\frac{\polylog(k)}{\sqrt{k}\cdot \poly(\min\cbra{\gamma, 1-\gamma})}}.
    \]
    Moreover if $a = 2^{d-1}$, then there exists an $\eta \in [0,1]$ such that for any $\delta \in (0,1/2)$ and $t\in \Rbb$, we have
    \[
        \left|\Pr\sbra{|F(\Ucal^m)| \ge t \text{ and } |F(\Ucal^m)| \text{ is even}} - \eta\Pr\sbra{\Bin(n, 1/2) \ge t}\right|
    \]
    and 
    \[
        \left|\Pr\sbra{|F(\Ucal^m)| \ge t \text{ and } |F(\Ucal^m)| \text{ is odd}} - (1-\eta)\Pr\sbra{\Bin(n, 1/2) \ge t}\right|
    \]
    are at most
    \[
        O\pbra{\frac{\log(1/\delta)^{O(d)^d}}{\sqrt{A\delta}} + \frac{\polylog(k)}{\sqrt{k}} + \delta}.
    \]
\end{proposition}
\begin{proof}
    It suffices to prove both results for arbitrary $f_i$, henceforth denoted $f$. We also use $T$ to denote $T_i$ in this simpler setting.
    
    We will first handle the general case, assuming $\gamma \le 1/2$ for notational simplicity, and afterwards address the $\gamma = 1/2$ case.
    Our main tool will be \Cref{lem:biased_kol_dist}. 
    By assumption, each $k$-tuple of output bits in $[n]\setminus T$ has distribution $\Ucal_{\gamma}^k$ under $f(\Ucal^m)$, so \Cref{lem:biased_kol_dist} implies
    \begin{align*}
        \Pr\sbra{|f(\Ucal^m)| \ge t} &\ge \Pr\sbra{|f(\Ucal^m)[[n]\setminus T]| \ge t} \\
        &\ge \Pr\sbra{\Bin(n-|T|, \gamma) \ge t} - O\pbra{\frac{\polylog(k)}{\sqrt{k}\cdot \poly(\gamma)}} \tag{by \Cref{lem:biased_kol_dist}} \\
        &\ge\Pr\sbra{\Bin(n, \gamma) \ge t + |T|} - O\pbra{\frac{\polylog(k)}{\sqrt{k}\cdot \poly(\gamma)}} \\
        &= \Pr\sbra{\Bin(n, \gamma) \ge t} - \Pr\sbra{t \le \Bin(n, \gamma) < t + |T|} - O\pbra{\frac{\polylog(k)}{\sqrt{k}\cdot \poly(\gamma)}} \\
        &\ge \Pr\sbra{\Bin(n, \gamma) \ge t} - O\pbra{\frac{|T|}{\sqrt{\gamma(1-\gamma)n}} + \frac{\polylog(k)}{\sqrt{k}\cdot \poly(\gamma)}} \tag{by \Cref{fct:bin_interval}} \\ 
        &\ge \Pr\sbra{\Bin(n, \gamma) \ge t} - O\pbra{\frac{\polylog(k)}{\sqrt{k}\cdot \poly(\gamma)}},
    \end{align*}
    where the final inequality follows from our assumption that $n$ is sufficiently large.
    A similar upper bound follows from comparing the complement distribution $1^n - f(\Ucal^m)$ to the binomial distribution with bias $1-\gamma$.
    This concludes the proof of the general case.

    We now turn to the case of $\gamma = 1/2$. 
    The proof will proceed similarly, but with some additional technical work.
    Here, we require a structural result about low degree $\Fbb_2$-polynomials.
    The following is essentially \cite[Theorem 3.1]{chattopadhyay2020xor}.
    \begin{theorem}\label{parity randomization theorem}
    Let $p$ be a degree-$d$ polynomial over $\mathbb{F}_2^n$ and $\delta \in (0,1/2)$. There exists a subset $R\subseteq [n]$ with $|R| \leq \log(1/\delta)^{O(d)^d}$ such that if we write $p(x) = p(x_{R^c}, x_R)$ where $x_R$ and $x_{R^c}$ are the coordinates in $R$ and not in $R$ respectively, then with probability at least $1-\delta$ over the choice of a random value of $x_{R^c}$ we have that
    \begin{equation}\label{parity randomization theorem eqn}
        \left|\Pr_{x_R}\sbra{p(x_{R^c},x_R) = 1} - \Pr_x\sbra{p(x)=1}\right| < \delta.
    \end{equation}
    \end{theorem}
    In words, \Cref{parity randomization theorem} says that we can (with high probability) re-randomize the output of the polynomial $p$ by simply re-randomizing the input coordinates of a small set $R$.

    Define $p\colon \bin^m \to \bin$ by $p(x) = |f(x)| \bmod 2$ (i.e., the parity of $f$'s output).
    Applying \Cref{parity randomization theorem} yields a set $R$ of at most $\log(1/\delta)^{O(d)^d}$ input bits.
    Let $S \subseteq [n]$ be the set of output bits in $T$ or affected by input bits in $R$.
    Note that
    \[
        |S| \le \frac{n|R|}A + |T| \le \frac{n\log(1/\delta)^{O(d)^d}}A+O_{d,k,A}(1)=\frac{n\log(1/\delta)^{O(d)^d}}A,
    \]
    where we used the fact that $n$ is sufficiently large in terms of $d,k$ and $A$ for the last equality.
    
    We will set $\eta = \Pr_{x\sim \Ucal^m}[p(x) = 0]=\Pr[|f(\Ucal^m)|\text{ is even}]$ and show
    \begin{align*}
        &\phantom{\le}\Pr\sbra{|f(\Ucal^m)| \ge t \text{ and } |f(\Ucal^m)| \text{ is even}} - \eta\Pr\sbra{\Bin(n, 1/2) \ge t}\\
        &\ge -O\pbra{\frac{\log(1/\delta)^{O(d)^d}}{\sqrt{A\delta}} + \frac{\polylog(k)}{\sqrt{k}} + \delta}.
    \end{align*}
    The case of odd $|f(\Ucal^m)|$ is almost identical, and as above, a similar upper bound follows from analyzing the complement distribution $1^n - f(\Ucal^m)$.

    Before proceeding further, we define several variables and events for the sake of future clarity. 
    Let 
    \begin{equation}\label{eq:prop:kol_dist_def_C}
    C\coloneqq\frac{|S|}2 - \sqrt{\frac{|S|}\delta}-|T|\ge\frac{|S|}2-\frac{\sqrt n\cdot\log(1/\delta)^{O(d)^d}}{\sqrt{A\delta}},
    \end{equation}
    where we again used the fact that $n$ is sufficiently large in terms of $d$, $k$, and $A$.
    Let \textsf{EVEN} be the event that $|f(\Ucal^m)|$ is even (and similarly for \textsf{ODD}). 
    Additionally, let \textsf{BIG} be the event that $|f(\Ucal^m)[[n]\setminus S]| \ge t-C$. Finally, let \textsf{GOOD} be the event that $x_{R^c}$ satisfies \Cref{parity randomization theorem eqn} (and \textsf{BAD} be the complement event). We have
    \begin{align}
        \Pr\sbra{|f(\Ucal^m)| \ge t \textrm{ and } \textsf{EVEN}} &\ge \Pr\sbra{\textsf{BIG} \textrm{ and } |f(\Ucal^m)[S]| \ge C \textrm{ and } \textsf{EVEN}} 
        \notag\\
        &= \Pr\sbra{\textsf{BIG} \textrm{ and } \textsf{EVEN}} - \Pr\sbra{\textsf{BIG} \textrm{ and } |f(\Ucal^m)[S]| < C \textrm{ and } \textsf{EVEN}} 
        \notag\\
        &\ge \Pr\sbra{\textsf{BIG}}\cdot \Pr\sbra{\textsf{EVEN} \bigm| \textsf{BIG}} - \Pr\sbra{|f(\Ucal^m)[S]| < C}.\label{eq:prop:kol_dist_t_and_even}
    \end{align}
    By assumption, each $k$-tuple of output bits in $[n]\setminus T$ has distribution $\Ucal_{1/2}^k$.
    Since $T\subseteq S$, \Cref{lem:biased_kol_dist} implies
    \begin{equation}\label{eq:prop:kol_dist_big_first}
    \Pr\sbra{\textsf{BIG}} \ge \Pr\sbra{\Bin(n-|S|, 1/2) \ge t-C} - O\pbra{\frac{\polylog(k)}{\sqrt{k}}}.
    \end{equation}
    Now let $X\sim\Bin(n-|S|,1/2)$ and $Y\sim\Bin(|S|,1/2)$. Then $X+Y\sim\Bin(n,1/2)$ and
    \begin{align*}
    &\phantom{\le}\Pr\sbra{\Bin(n-|S|, 1/2) \ge t-C}\\
    &=\Pr\sbra{X\ge t-C}=\Pr\sbra{X+Y\ge t-C+Y}\\
    &\ge \Pr[X+Y\ge t]-\Pr[t\le X+Y<t-C+Y]\\
    &=\Pr[\Bin(n,1/2)\ge t]-\E_Y\sbra{\Pr_X[t\le X+Y<t-C+Y]}\\
    &\ge\Pr[\Bin(n,1/2)\ge t]-O\pbra{\E_Y\sbra{\frac{(Y-C)\cdot1_{Y\ge C}}{\sqrt n}}}
    \tag{by \Cref{fct:bin_interval} and independence of $X,Y$}\\
    &=\Pr[\Bin(n,1/2)\ge t]-O\pbra{\E_Y\sbra{\frac{Y-C}{\sqrt n}}+\E_Y\sbra{\frac{(C-Y)\cdot1_{Y<C}}{\sqrt n}}}\\
    &\ge\Pr[\Bin(n,1/2)\ge t]-O\pbra{\E_Y\sbra{\frac{Y-C}{\sqrt n}}+\E_Y\sbra{\frac{\abs{Y-|S|/2}}{\sqrt n}}}
    \tag{since $C\le|S|/2$}\\
    &\ge\Pr[\Bin(n,1/2)\ge t]-O\pbra{\frac{\pbra{|S|/2-C}+\sqrt{|S|}}{\sqrt n}}
    \tag{since $Y\sim\Bin(|S|,1/2)$}\\
    &=\Pr[\Bin(n,1/2)\ge t]-O\pbra{\frac{|S|/2-C}{\sqrt n}}
    \tag{since $C\le|S|/2-\sqrt{|S|/\delta}$}\\
    &\ge\Pr[\Bin(n,1/2)\ge t]-O\pbra{\frac{\log(1/\delta)^{O(d)^d}}{\sqrt{A\delta}}}.
    \tag{by \Cref{eq:prop:kol_dist_def_C}}
    \end{align*}
    Combining with \Cref{eq:prop:kol_dist_big_first}, we have
    \begin{align}
        \Pr\sbra{\textsf{BIG}} \ge \Pr\sbra{\Bin(n, 1/2) \ge t} - O\pbra{\frac{\log(1/\delta)^{O(d)^d}}{\sqrt{A\delta}}+ \frac{\polylog(k)}{\sqrt{k}}}. \label{eq:cdf_f_Bc}
    \end{align}

    To lower bound the conditional probability $\Pr\sbra{\textsf{EVEN} \bigm| \textsf{BIG}}$, it will be slightly more convenient to upper bound $\Pr\sbra{\textsf{ODD} \bigm| \textsf{BIG}}$. For this, observe that the events $\textsf{BIG}$ and $\textsf{GOOD}$ depend only on input bits in $R^c$.
    Hence by \Cref{parity randomization theorem}, conditioned on events $\textsf{BIG}$ and $\textsf{GOOD}$, event $\textsf{ODD}$ (which rerandomizes input bits in $R$) happens with probability at most $1-\eta+\delta$.
    Moreover, event $\textsf{BAD}$ itself happens with probability at most $\delta$.
    Therefore
    \begin{align*}
        \Pr\sbra{\textsf{ODD} \bigm| \textsf{BIG}} &= \Pr\sbra{\textsf{GOOD} \bigm| \textsf{BIG}}\cdot \Pr\sbra{\textsf{ODD} \bigm| \textsf{GOOD} \text{ and } \textsf{BIG}} \\
        & \ \ \: + \Pr\sbra{\textsf{BAD} \bigm| \textsf{BIG}}\cdot \Pr\sbra{\textsf{ODD} \bigm| \textsf{BAD} \text{ and } \textsf{BIG}} \\
        &\le \Pr\sbra{\textsf{ODD} \bigm| \textsf{GOOD}\text{ and } \textsf{BIG}} + \Pr\sbra{\textsf{BAD} \bigm| \textsf{BIG}} \\
        &\le \Pr\sbra{\textsf{ODD} \bigm| \textsf{GOOD}\text{ and } \textsf{BIG}} + \frac{\Pr\sbra{\textsf{BAD}}}{\Pr\sbra{\textsf{BIG}}} \\
        &\le 1-\eta + \delta + \frac{\delta}{\Pr\sbra{\textsf{BIG}}}
        \le 1-\eta + \frac{2\delta}{\Pr\sbra{\textsf{BIG}}}.
    \end{align*}
    Combining with \Cref{eq:cdf_f_Bc}, we have that
    \begin{align*}
        \Pr\sbra{\textsf{BIG}}\cdot \Pr\sbra{\textsf{EVEN} \bigm| \textsf{BIG}}
        &\ge \eta \Pr\sbra{\Bin(n,1/2) \ge t} - O\pbra{\frac{\log(1/\delta)^{O(d)^d}}{\sqrt{A\delta}} + \frac{\polylog(k)}{\sqrt{k}} + \delta}.
    \end{align*}
    In light of \Cref{eq:prop:kol_dist_t_and_even} and to complete our proof, it remains to bound $\Pr\sbra{|f(\Ucal^m)[S]| < C}$. We have
    \begin{align*}
        \Pr\sbra{|f(\Ucal^m)[S]| < C} &\le \Pr\sbra{|f(\Ucal^m)[S\setminus T]| < \frac{|S|}{2} - \sqrt{\frac{|S|}{\delta}}-|T|} \\
        &\le \Pr\sbra{|f(\Ucal^m)[S\setminus T]| < \frac{|S\setminus T|}{2} - \sqrt{\frac{|S \setminus T|}{\delta}}} < \delta,
    \end{align*}
    where the final inequality follows from Chebyshev's inequality and the observation that the outputs in $S\setminus T$ are $2$-wise independent.
\end{proof}

\subsection{Approximate Continuity}\label{subsec:approx_continuity}

Our third step in proving \Cref{thm:main} is a type of continuity result for the output weight of $f(\Ucal^m)$. 
The argument will require three ingredients from prior works.
The first allows us to find many independent collections of output bits.
Recall we may view the input-output dependencies of a $d$-local function $f\colon\bin^m\to\bin^n$ as a hypergraph on the vertex set $[n]$ with one edge for each of the $m$ input bits containing all the output bits it affects.
By the locality assumption, no vertex has degree more than $d$.

\begin{lemma}[{\cite[Corollary 4.11]{kane2024locality}}]\label{lem:neighborhoods}
Let $G$ be a hypergraph on $n$ vertices with maximum degree at most $d$. For any increasing function $F\colon\mathbb{N}\rightarrow \mathbb{N}$, there exists a set $S$ of edges in $G$ whose removal yields at least\footnote{Recall $O_{d,F}(1)$ denotes a quantity whose value is constant once $d$ and $F$ are fixed.} $r = n/O_{d,F}(1)$ vertices in $G$ whose neighborhoods have size at most $t=O_{d,F}(1)$ and are pairwise non-adjacent, and satisfies $|S| \leq r/F(t)$.
\end{lemma}

The neighborhoods from \Cref{lem:neighborhoods} appear in our context as collections of independent groups of output bits.
When we analyze the behavior of their Hamming weight, they become a sum of independent integer random variables.
This sum has a ``continuous'' output distribution unless almost all of the integer random variables are nearly constant modulo some integer $s \ge 2$.
In fact, a similar property holds even if the integer random variables are only noticeably non-constant modulo $s$ for $s \ge 3$, except now the continuity only holds for weights that are an even distance apart.
More formally, we will use the following density comparison result, which is a special case of \cite[Theorem A.1]{kane2024locally2}.

\begin{lemma}[{\cite[Theorem A.1]{kane2024locally2}}]\label{lem:comp_llt}
Let $t\ge1$ be an integer, and let $X_1,\ldots,X_n$ be independent random variables in $\cbra{0,1,\ldots,t}$.
For each $i\in[n]$ and integer $s\ge1$, define $p_{s,i}=\max_{x\in\Zbb}\Pr\sbra{X_i\equiv x\Mod s}$.
Suppose for some $L > 0$,
\begin{equation*}
\sum_{i\in[n]}(1-p_{s,i})\ge L\cdot n
\quad\text{holds for all $s\in\cbra{2,3,\dots, t}$.}
\end{equation*}
Then for any $x\in\Zbb$ and $\Delta\in\Zbb$, we have
$$
\Pr\sbra{\sum_{i\in[n]}X_i=x}-\Pr\sbra{\sum_{i\in[n]}X_i=x+\Delta}\le O_{L,t}\pbra{\frac{|\Delta|}{n}}.
$$
If instead $\sum_{i\in[n]}(1-p_{s,i})\ge L\cdot n$ only holds for $s\in\cbra{3,4,\dots, t}$, the same conclusion still holds for any even $\Delta \in \Zbb$.
\end{lemma}

In order to satisfy the variability assumption of \Cref{lem:comp_llt}, we will need to exploit the fact that two coupled, $\gamma$-biased random vectors have different Hamming weight distributions, as long as part of their entries are independent.
The following is essentially \cite[Lemma 4.4]{kane2024locality}.

\begin{lemma}\label{lem:anticoncentration_after_coupling_all}
    Let $(X,Y,Z,W)$ be a random variable where $X,Z\in\bin$ and $Y,W\in\bin^{t-1}$.
    Let $q\ge\min\cbra{3,t+1}$ be an integer.\footnote{If $q\ge t+1$, then one may instead apply \Cref{lem:anticoncentration_after_coupling_all} with modulus $t+1$, since $X+|Y|\equiv Z+|W|\Mod q$ is equivalent to $X+|Y|=Z+|W|$ for $q\ge t+1$.}
    Assume
    \begin{itemize}
    \item $X$ is independent from $(Z,W)$ and $Z$ is independent from $(X,Y)$, and
    \item $(X,Y)$ and $(Z,W)$ have distribution $\Ucal_\gamma^t$ for some $\gamma\in(0,1)$.
    \end{itemize}
    Then we have
    $$
    \Pr\sbra{X+|Y|\equiv Z+|W|\Mod q} < 1.
    $$
    Moreover, the same conclusion holds for $q \ge 2$ when $\gamma \ne 1/2$.
\end{lemma}
We remark that the ``moreover'' part of the conclusion is not explicitly stated in \cite{kane2024locality}, but the proof can be easily modified to show this.
For completeness, we include the full details in \Cref{app:missing_sec:characterize}.

We can now state and prove the main result of this subsection.
\begin{proposition}\label{prop:continuity}
    Let $f\colon\bin^m\to\bin^n$ be a $d$-local function with $n$ sufficiently large in terms of $d$, and let $\gamma \in (0,1)$ be an integer multiple of $2^{-d}$.
    Then the distribution $f(\Ucal^m)$ can be written as a mixture of distributions $E$ and $W$, where
    \begin{enumerate}
        \item $\tvdist{E - \Dcal} \ge 1 - \exp\cbra{-\Omega_d(n)}$ for any symmetric distribution $\Dcal$ over $\bin^n$ with weights $\gamma n \pm n^{2/3}$, and

        \item For all $w \in \cbra{0,1,\dots, n}$ and $\Delta \in \Zbb$ if $\gamma \ne 1/2$ (or even $\Delta \in \Zbb$ if $\gamma = 1/2$), we have
        \[
            \bigg|\Pr\sbra{|W| = w} - \Pr\sbra{|W| = w + \Delta}\bigg| \le O_d\pbra{\frac{|\Delta|}{n}}.
        \]
    \end{enumerate}
\end{proposition}

The above proposition is very similar to \cite[Proposition 4.9]{kane2024locally2}, and our proof largely follows the one present there.
At a high level, we proceed by using \Cref{lem:neighborhoods} to obtain many independent output neighborhoods.
We then classify these neighborhoods according to whether their marginal distributions differ from the $\gamma$-biased product distribution or not.
The former situation corresponds to the first conclusion of \Cref{prop:continuity}, while the latter corresponds to the second.

\begin{proof}
    Let $S\subseteq [m]$ be the set of input coordinates promised by \Cref{lem:neighborhoods}, taking $F(t)$ to be a sufficiently large multiple of $2^{2dt}$.
    For each conditioning $\rho \in \bin^S$, consider the restricted function $f_\rho\colon \bin^{[m]\setminus S}\to\bin^n$ defined by $f_\rho(x) = f(x,\rho)$.
    We call a $\rho$ \emph{good} if at least half of the $r$ non-adjacent neighborhoods $\cbra{N_i}$ (of size at most $t \le C_d$) promised by \Cref{lem:neighborhoods} satisfy $f_\rho(\Ucal^{[m]\setminus S})[N_i] = \Ucal_\gamma^{N_i}$.
    Set
    \[
        E \coloneqq \E_{\rho \,:\, \rho \text { is not good}} f_\rho(\Ucal^{[m]\setminus S}) \qquad\text{and}\qquad W \coloneqq \E_{\rho \,:\, \rho \text { is good}} f_\rho(\Ucal^{[m]\setminus S}),
    \]
    so that $f(\Ucal^m) = \beta E + (1-\beta) W$ where $\beta \in [0,1]$ is the fraction of not good $\rho$.
    
    We first prove conclusion 1. 
    Suppose $\rho\in \bin^S$ is not good.
    Let $N$ be one of the at least $r/2$ non-adjacent neighborhoods with $f_\rho(\Ucal^{[m]\setminus S})[N] \ne \Ucal_\gamma^{N}$.
    Since $N$ depends on at most $dt$ many input bits and $\gamma$ is an integer multiple of $2^{-d}$, $f_\rho(\Ucal^{[m]\setminus S})[N]$ is at least $2^{-dt}$-far from $\Ucal_\gamma^{N}$.
    Moreover, define
    $$
    \nu \coloneqq \min_{x \in \supp{\Dcal}}\frac{\Ucal_\gamma^{n}(x)}{\Dcal(x)}.
    $$
    Let $x \in \supp{\Dcal}$ of weight $k=\gamma n\pm n^{2/3}$.
    Since $\Dcal$ is symmetric, we have $\Dcal(x) \le 1/\binom{n}{k}$.
    Thus,
    \begin{align}
        \nu 
        &\ge\frac{\gamma^{k}(1-\gamma)^{n-k}}{1/\binom{n}{k}}
        \ge \frac{2^{n\Hcal(\frac{k}{n})}}{\sqrt{8k(1-\frac{k}{n})}} \cdot \gamma^{k}(1-\gamma)^{n-k}\tag{by \Cref{fct:individual_binom}} \\
        &= \frac{1}{\sqrt{8|x|(1-\frac{k}{n})}} \cdot \pbra{\frac{\gamma}{\frac{k}{n}}}^{k}\pbra{\frac{1-\gamma}{1-\frac{k}{n}}}^{n-k} \notag \\
        &= \frac{1}{\sqrt{8|x|(1-\frac{k}{n})}} \cdot \pbra{\frac{\gamma}{\gamma \pm n^{-1/3}}}^{k}\pbra{\frac{1-\gamma}{1-\gamma \mp n^{-1/3}}}^{n-k}  \notag \\
        &\ge \Theta_d(n^{-1/2}) \cdot \pbra{\frac{1}{1 + O_d(n^{-1/3})}}^{n} \ge \exp\cbra{-O_d(n^{2/3})}, \label{eq:eta_LB}
    \end{align}
    where we used the fact that $|k-\gamma n| \le n^{2/3}$ and $n$ sufficiently large for the last line.
    Noting that the non-adjacency of the neighborhoods implies the distributions $f_\rho(\Ucal^m)[N_i]$ and $f_\rho(\Ucal^{[m]\setminus S})[N_j]$ are independent for $i \ne j$, we can apply \Cref{lem:tvdist_after_product} to conclude 
    \[
        \tvdist{f_\rho(\Ucal^{[m]\setminus S}) - \Dcal} \ge 1-2 \exp\cbra{-r2^{-2dt - 2}}/\nu.
    \]
    Then \Cref{lem:tvdist_after_conditioning} and our choice of $F(t) = \Omega(2^{2dt})$ sufficiently large imply
    \begin{align*}
        \tvdist{E - \Dcal} &\ge 1-2^{|S|}\cdot\pbra{2 \exp\cbra{-r2^{-2dt - 2}}}/\nu \\
        &\ge1 - 2\pbra{\exp\cbra{\frac{r}{F(t)} - \frac{r}{2^{2dt+2}}}}/\nu \\
        &\ge 1 - \exp\cbra{-\Omega_d\pbra{n}}\cdot \exp\cbra{O_d(n^{2/3})} \ge 1 - \exp\cbra{-\Omega_d\pbra{n}} \tag{by \Cref{eq:eta_LB}}.
    \end{align*}
    
    It remains to verify conclusion 2.
    Fix an arbitrary good $\rho$ for the remainder of the argument.
    We assume without loss of generality that for some $r' \ge r/2$, the neighborhoods $N(1), \dots, N(r') \subseteq [n]$ satisfy $f_\rho(\Ucal^{[m]\setminus S})[N_i] = \Ucal_\gamma^{N_i}$.
    
    Let $B \subseteq [m]\setminus S$ be the set of input bits that do not affect any central elements (i.e., the $r'$ output bits that generate $N(1), \dots, N(r')$), which we henceforth refer to as \emph{extraneous inputs}.
    For each conditioning $\sigma \in \bin^B$, we define the restricted functions $f_{\rho, \sigma}\colon \bin^{[m]\setminus (S\cup B)} \to \bin^n$ by $f_{\rho, \sigma}(x) = f(x,\rho, \sigma)$, so that
    \[
        f_\rho = \E_{\sigma \in \bin^B} \sbra{f_{\rho, \sigma}}.
    \]
    Observe that the value of every output bit in $[n] \setminus \pbra{N(1) \cup \cdots \cup N(r')}$ is fixed for all $f_{\rho, \sigma}$, and the weight of the output bits of each neighborhood becomes a random variable
    \[
        X_{\sigma, i} \coloneqq \sum_{j \in N(i)} f_{\rho, \sigma}(\Ucal^{[m]\setminus (S\cup B)})[\cbra{j}].
    \]
    In particular, the total weight of the output of $f_{\rho, \sigma}(\Ucal^{[m]\setminus (S\cup B)})$ is some constant plus the sum of the $X_{\sigma, i}$'s, which are independent. 
    We would like to claim that \Cref{lem:comp_llt} can be applied to this situation with high probability.
    
    From here, we proceed similarly to \cite[Claims 5.16 \& 5.23]{kane2024locality}. 
    For each integer $s \ge 2$, define
    \[
        p_{\sigma, s, i} = \max_{x\in \Zbb} \Pr\sbra{X_{\sigma, i} \equiv x \Mod{s} \mid \rho, \sigma},
    \]
    where recall we previously fixed a good $\rho$.
    
    \begin{claim}\label{clm:X_i_weight_fixed}
        For any $i \in [r']$ and $s \ge 3$, there exists some $\sigma \in \bin^B$ satisfying $p_{\sigma, s, i} < 1$ (i.e., $X_{\sigma, i}$ is not constant modulo $s$). Moreover if $\gamma \ne 1/2$, the same is true for $s = 2$.
    \end{claim}
    \begin{proof}
        Consider the neighborhood $N \coloneqq N(i)$ of size $t_i \le t$, and let $\Ical$ be the set of input bits that the output bit $i$ depends on.
        Note $|\Ical| \le d$.
        We randomly sample $\sigma \in \bin^B$ and two independent $\lambda, \lambda' \in \bin^\Ical$. 
        Since $f_{\rho}[N]$ only depends on the bits in $B \cup \Ical$, we can define 
        $(V,Z) = f_{\rho}((\sigma, \lambda))$ and $(V',Z') = f_{\rho}((\sigma, \lambda'))$, where $V, V' \in \bin$ are the values of $N$'s center and $Z, Z' \in \bin^{t_i - 1}$ are the values of the remaining bits in $N$. 
        Observe that $V$ is independent from $(V', Z')$ and likewise $V'$ is independent from $(V, Z)$, so we may apply \Cref{lem:anticoncentration_after_coupling_all} to conclude
        \[  
            \Pr\sbra{V+|Z|\equiv V'+|Z'|\Mod s} < 1.
        \]
        In particular, there must exist some $\sigma$ where $X_{\sigma,i}$ is not constant modulo $s$.
    \end{proof}
    
    For our fixed good $\rho$ and some fixed $2\leq s \leq t$ (or $3\leq s \leq t$ if $\gamma=1/2$), let $\Ecal_i$ be the event that $X_{\sigma, i}$ is at least $2^{-d}$-far from constant modulo $s$.
    Since $X_{\sigma, i}$ depends on at most $d$ input bits (namely, the bits $i$ depends on), if it is not constant modulo $s$, it must be at least $2^{-d}$-far from constant. 
    Furthermore, the bits in the neighborhood $N(i)$ depend on at most $dt$ input bits, so
    \[
        \Pr_{\sigma \in \bin^B}\sbra{\Ecal_i} \ge 2^{-dt}.
    \]
    Recall that the neighborhoods are non-adjacent (after removing the edges in $S$), so the extraneous bits used to determine $X_{\sigma, i}$ are disjoint from those used to determine $X_{\sigma, j}$ for $i\neq j$.
    Thus, whether or not $X_i$ is constant modulo $s$ is independent of whether $X_j$ is.
    Applying Chernoff's inequality (\Cref{fct:chernoff}) with $\delta = 1/2$, we have
    \begin{align*}
        \Pr_{\sigma \in \bin^B} \sbra{\sum_{i\in [r']} \Ecal_i \le 2^{-dt - 1}r'} \le \exp\cbra{-\frac{2^{-dt}\cdot r'}{8}} \le \exp\cbra{-\frac{2^{-dt}\cdot r}{16}} = \exp\cbra{-\Omega_d(n)}.
    \end{align*}

    Now call a $\sigma \in \bin^B$ \emph{fluid} if $\sum_i \Ecal_i \ge 2^{-dt - 1}r'$ for \emph{every} $2\leq s \leq t$ (or $3\leq s \leq t$ if $\gamma=1/2$).
    By a union bound and the fact that $t = O_d(1)$, we find
    \begin{equation}\label{eq:good_sigma}
        \Pr_{\sigma \in \bin^B}\sbra{\sigma \text{ is fluid}} = 1-\exp\cbra{-\Omega_d(n)}.
    \end{equation}
    Note any fluid $\sigma$ satisfies
    \[
        \sum_{i\in[r']}(1-p_{\sigma, s,i}) \ge 2^{-dt-1}\cdot r'\cdot 2^{-d} \ge 2^{-d(t+1)-1}\cdot r \ge \Omega_d(n),
    \]
    so we can apply \Cref{lem:comp_llt} to find for any $w\in\Zbb$ and $\Delta\in\Zbb$ (or even $\Delta \in\Zbb$ if $\gamma = 1/2$) that
    \begin{equation}\label{eq:conditioned_continuity}
            \left|\Pr_{x\sim \Ucal^{[m]\setminus (S\cup B)}}\sbra{|f_{\rho, \sigma}(x)| = w} - \Pr_{x\sim \Ucal^{[m]\setminus (S\cup B)}}\sbra{|f_{\rho, \sigma}(x)| = w + \Delta}\right| \le O_d\pbra{\frac{|\Delta|}{n}}.
    \end{equation}
    Taking the mixture over all such subfunctions, we have
    \begin{align*}
        \phantom{=}&\:\bigg|\Pr\sbra{|W| = w} - \Pr\sbra{|W| = w + \Delta}\bigg| \\
        =&\: \bigg|\E_{\text{good } \rho} \sbra{\Pr_{x\sim \Ucal^{[m]\setminus S}}\sbra{|f_\rho(x)| = w}} - \E_{\text{good }\rho} \sbra{\Pr_{x\sim \Ucal^{[m]\setminus S}}\sbra{|f_\rho(x)| = w+\Delta}}\bigg| \\
        \le&\: \E_{\substack{\text{good } \rho \\ \sigma}}\bigg|\Pr_{x\sim \Ucal^{[m]\setminus (S\cup B)}}\sbra{|f_{\rho, \sigma}(x)| = w} -  \Pr_{x\sim \Ucal^{[m]\setminus (S\cup B)}}\sbra{|f_{\rho, \sigma}(x)| = w+\Delta}\bigg| \tag{by triangle inequality} \\
        \le&\: e^{-\Omega_d(n)} + \E_{\substack{\text{good } \rho \\ \text{fluid } \sigma}}\bigg|\Pr_{x\sim \Ucal^{[m]\setminus (S\cup B)}}\sbra{|f_{\rho, \sigma}(x)| = w} -  \Pr_{x\sim \Ucal^{[m]\setminus (S\cup B)}}\sbra{|f_{\rho, \sigma}(x)| = w+\Delta}\bigg| \tag{by \Cref{eq:good_sigma}} \\
        \le&\: e^{-\Omega_d(n)} + O_d\pbra{\frac{|\Delta|}{n}}. \tag{by \Cref{eq:conditioned_continuity}}
    \end{align*}
    If $\Delta = 0$, then $\bigg|\Pr\sbra{|W| = w} - \Pr\sbra{|W| = w + \Delta}\bigg|$ is trivially at most $O_d\pbra{\frac{|\Delta|}{n}}$.
    Hence we may assume $|\Delta| \ge 1$, in which case $e^{-\Omega_d(n)} + O_d\pbra{\frac{|\Delta|}{n}} \le O_d\pbra{\frac{|\Delta|}{n}}$.
    This concludes the proof.
\end{proof}

\subsection{Putting It Together}\label{subsec:put_together}

In this final subsection, we combine our earlier results to prove \Cref{thm:main}.
Recall \Cref{lem:distance_to_sym} gives that the distance between $f(\Ucal^m)$ and any symmetric distribution $\Pcal$ is
\begin{equation}\label{eq:lem:lem:distance_to_sym:restated}
    \tvdist{f(\Ucal^m) - \Pcal}  = \Theta\pbra{\tvdist{|f(\Ucal^m)|-|\Pcal|} + \tvdist{f(\Ucal^m)-f(\Ucal^m)_\sym}},
\end{equation}
where the symmetrization $f(\Ucal^m)_\sym$ is the distribution resulting from randomly permuting the coordinates of a string $x\sim f(\Ucal^m)$.
Under the assumption $f(\Ucal^m)$ is $\eps$-close to a symmetric distribution $\Dcal$, \Cref{eq:lem:lem:distance_to_sym:restated} implies $\tvdist{f(\Ucal^m) - f(\Ucal^m)_\sym} \le O(\eps)$.
Thus it remains to show that the weight distribution of $f(\Ucal^m)$ is close to the weight distribution of a mixture of the form specified in \Cref{thm:main}.
We first prove the simpler case where no input bit of $f$ affects many output bits, the output distribution of $f$ restricted to (almost) any small set of output bits looks like a $\gamma$-biased product distribution, and the symmetric distribution $\Dcal$ is supported on weights around $\gamma n$.

\begin{lemma}\label{lem:close_to_bin_weights_mixture}
    Let $k\ge 2$ and $\ell \ge 1$ be integers, and let $f_1, \dots, f_\ell\colon\bin^m\to\bin^n$ be $d$-local functions.
    For each $f_i$, suppose no input bit affects more than $n/A$ output bits, and suppose there exists a subset $T_i\subseteq[n]$ of size $|T_i|\le O_{d,k,A}(1)$ such that every $k$-tuple of output bits in $[n]\setminus T_i$ has distribution $\Ucal_{\gamma}^k$, where $\gamma=a/2^d$ for some fixed integer $0 \le a \le 2^d$.
    
    Furthermore, assume there is some mixture $F$ of the $f_i$'s such that $F(\Ucal^m)$ is $\eps$-close to a symmetric distribution $\Dcal$ over $\bin^n$ which is supported on strings of weight $\gamma n \pm n^{2/3}$.
    Then if $\gamma\ne1/2$ and $n$ is sufficiently large in terms of $d$, $k$, $A$, $\ell$, and $\eps$, we have
    \[
        \tvdist{|F(\Ucal^m)| - \Bin(n,\gamma)} \le O_d(\eps + k^{-1/5}).
    \]
    Moreover if $\gamma = 1/2$, then there exists a mixture $\Qcal = \eta|\Deven| + (1-\eta)|\Dodd|$ satisfying
    \[
        \tvdist{|F(\Ucal^m)| - \Qcal} \le O_d\pbra{\eps + k^{-1/5} + \log(k)\sqrt{\frac{\log(A)^{O(d)^d}}{A^{1/3}} + \frac{\polylog(k)}{\sqrt{k}}}}.
    \]
\end{lemma}

Much of the following proof overlaps with \cite[Section 5.3]{kane2024locally2}.
The overall idea is to use \Cref{prop:continuity} to argue that the output weight of (most of) $f(\Ucal^m)$ satisfies a continuity property: weights at distance $\Delta$ apart are assigned mass that differs by at most $O_d(\Delta/n)$.
Since \Cref{prop:kol_dist} shows any contiguous weight interval is given roughly the same probability as a binomial distribution, our continuity property implies that the output weight of $f(\Ucal^m)$ behaves similarly to a binomial distribution pointwise.
In the case of $\gamma = 1/2$, we can combine the precise control \Cref{prop:kol_dist} provides on the probability of being in a fixed interval and even/odd with the parity continuity of \Cref{prop:continuity} to carry out a similar argument.

\begin{proof}
    We first address the case of $\gamma \in \bin$, where several of the tools we have developed (e.g., \Cref{prop:kol_dist} and \Cref{prop:continuity}) do not apply.
    We will show the argument for $\gamma = 0$; the case of $\gamma = 1$ is essentially identical.
    Let $T = \bigcup_i T_i$, and observe that by assumption, all output bits in $[n] \setminus T$ are identically zero.
    Thus,
    \begin{equation}\label{eq:gamma_is_0_binom}
        \tvdist{|F(\Ucal^m)| - \Bin(n,0)} = \Pr\sbra{F(\Ucal^m)[T] \ne 0^{T}}.
    \end{equation}
    Let us briefly consider the symmetrized distribution $F(\Ucal^m)_\sym$.
    For clarity, let $t = |T|$.
    Since any string $x \sim F(\Ucal^m)$ has Hamming weight $|x| \le t \le \ell\cdot O_{d,k,A}(1)$, we have
    \begin{align}
        \Pr\sbra{F(\Ucal^m)_\sym[T] \ne 0^{T}} &\le 1 - \frac{\binom{n-t}{t}}{\binom{n}{t}} = 1 - \prod_{i=0}^{t-1}\frac{n-t-i}{n-i} \notag \\
        &\le 1 - \pbra{1 - \frac{t}{n-t}}^t \le 1 - \exp\cbra{-\frac{2t^2}{n-t}} \le \eps, \label{eq:f_sym_T_not_zero}
    \end{align}
    since $n$ is sufficiently large in terms of $d$, $k$, $\ell$, $A$, and $\eps$.
    By \Cref{lem:distance_to_sym} and our initial assumption, we know that
    \[
        \Pr\sbra{F(\Ucal^m)[T] \ne 0^{T}} - \Pr\sbra{F(\Ucal^m)_\sym[T] \ne 0^{T}} \le \tvdist{F(\Ucal^m) - F(\Ucal^m)_\sym} = O(\eps).
    \]
    Combining with \Cref{eq:gamma_is_0_binom} and \Cref{eq:f_sym_T_not_zero} yields the desired upper bound of 
    \[
        \tvdist{|F(\Ucal^m)| - \Bin(n,0)} \le \Pr\sbra{F(\Ucal^m)_\sym[T] \ne 0^{T}} + O(\eps) = O(\eps).
    \]

    We now turn to the case of $\gamma \not\in \cbra{0,1/2,1}$, assuming for notational convenience that $\gamma \le 1/2$.
    Afterwards, we will describe the necessary modifications to handle the remaining $\gamma = 1/2$ case.
    Define $\kappa \coloneqq \sqrt{\polylog(k) / (\sqrt{k}\cdot \poly(\gamma))}$ with the exponents on the $\polylog(k)$ and $\poly(\gamma)$ corresponding to those in \Cref{prop:kol_dist}, and partition $\cbra{0,1,\dots, n}$ into (consecutive) intervals of length $\Theta(\kappa\sqrt{n})$.
    Observe that for any such interval $\Ical$, \Cref{prop:kol_dist} implies
    \begin{equation}\label{eq:close_in_intervals}
        \left|\Pr\sbra{|F(\Ucal^m)|\in \Ical} - \Pr\sbra{\Bin(n, \gamma) \in \Ical}\right| = O(\kappa^2).
    \end{equation}
    Next, we apply \Cref{prop:continuity} to each $f_i$ to write $f_i(\Ucal^m) = \lambda_i E_i + (1-\lambda_i)W_i$, where
    \begin{enumerate}
        \item\label{itm:1_large_dist} $\tvdist{E_i - \Dcal} \ge 1 - \exp\cbra{-\Omega_d(n)}$, and

        \item For all $w \in \cbra{0,1,\dots, n}$ and $\Delta \in \Zbb$, we have
        \begin{equation}\label{eq:lem:close_to_bin_weights_mixture:continuity}
            \bigg|\Pr\sbra{|W_i| = w} - \Pr\sbra{|W_i| = w + \Delta}\bigg| \le O_d\pbra{\frac{|\Delta|}{n}}.
        \end{equation}
    \end{enumerate}
    
    Since $F$ is a mixture of the $f_i$'s, there exist $c_1, \dots, c_\ell$ such that
    \[
        F(\Ucal^m) = \sum_{i} c_i\lambda_i E + \pbra{1 - \sum_{i} c_i\lambda_i} W,
    \]
    where
    \[
        E = \frac{1}{\sum_{i} c_i\lambda_i}\cdot \sum_{i} c_i\lambda_i E_i \quad\text{and}\quad W = \frac{1}{1 - \sum_{i} c_i\lambda_i}\cdot \sum_{i} c_i(1-\lambda_i) W_i.
    \]
    For each $i \in [\ell]$, our distance bound between $E_i$ and $\Dcal$ (\Cref{itm:1_large_dist}) guarantees an event $\Ecal_i$ with mass at least $1 - \exp\cbra{-\Omega_d(n)}$ in $E_i$, but mass at most $\exp\cbra{-\Omega_d(n)}$ in $\Dcal$.
    Thus if we define $\Ecal = \cup \Ecal_i$, then
    \[
        \eps \ge \tvdist{F(\Ucal^m) - \Dcal} \ge \sum_{i} c_i\lambda_i(1 - \exp\cbra{-\Omega_d(n)}) - \ell\cdot \exp\cbra{-\Omega_d(n)}.
    \]
    In particular, $\sum_{i} c_i\lambda_i \le O(\eps)$, and
    \begin{equation}\label{eq:TVD_on_W_Dcal}
        \tvdist{W - F(\Ucal^m)} \le O(\eps).
    \end{equation}

    By expanding the definition of total variation distance, we find
    \begin{align}
        \tvdist{|W| - \Bin(n,\gamma)} &= \frac{1}{2} \sum_{w=0}^n \left|\Pr\sbra{|W| = w} - \Pr\sbra{\Bin(n,\gamma) = w}\right| \notag \\
        &= \frac{1}{2} \sum_{\Ical}\sum_{w\in \Ical} \left|\Pr\sbra{|W| = w} - \Pr\sbra{\Bin(n,\gamma) = w}\right| \notag \\ 
        &\le \sum_{\substack{\Ical\\w\in \Ical}}\left|\Pr\sbra{|W| = w} - \frac{\Pr\sbra{|W| \in \Ical}}{|\Ical|}\right| + \left|\frac{\Pr\sbra{|W| \in \Ical} - \Pr\sbra{\Bin(n,\gamma) \in \Ical}}{|\Ical|}\right| \notag \\
        &\qquad + \left|\Pr\sbra{\Bin(n,\gamma) = w} - \frac{\Pr\sbra{\Bin(n,\gamma) \in \Ical}}{|\Ical|}\right|. \notag
    \end{align}
For any $\delta \in (0,1)$, \Cref{fct:hoeffding} implies that all but $O(\delta)$ of the mass of $\Bin(n,\gamma)$ is supported on $O(\log(1/\delta)/\kappa)$ intervals, which we call \emph{big} (and the remaining intervals \emph{small}).
Moreover, $|W|$ also only assigns $C\cdot(\eps + \kappa\log(1/\delta) + \delta)$ mass to small intervals for some sufficiently large constant $C > 0$, as otherwise we obtain the contradiction
\begin{align*}
    C\cdot(\eps + \kappa\log(1/\delta)) &\le C\cdot(\eps + \kappa\log(1/\delta) + \delta) - O(\delta) \\
    &\le \sum_{\text{small } \Ical}\Pr\sbra{|W| \in \Ical} - \sum_{\text{small } \Ical} \Pr\sbra{\Bin(n,\gamma) \in \Ical} \\
    &= \sum_{\text{big } \Ical}\Pr\sbra{\Bin(n,\gamma) \in \Ical} - \sum_{\text{big } \Ical} \Pr\sbra{|W| \in \Ical} \\
    &\le \sum_{\text{big } \Ical}\abs{\Pr\sbra{\Bin(n,\gamma) \in \Ical} - \Pr\sbra{|W| \in \Ical}} \\
    &\le \sum_{\text{big } \Ical}\Big(\abs{ \Pr\sbra{\Bin(n,\gamma) \in \Ical} - \Pr\sbra{|F(\Ucal^m)| \in \Ical}} \\
    &\qquad + \abs{\Pr\sbra{|F(\Ucal^m)| \in \Ical} - \Pr\sbra{|W| \in \Ical}}\Big) \\
    &\le O(\kappa\log(1/\delta)) + 2\tvdist{|F(\Ucal^m)| - |W|} \le O(\eps + \kappa\log(1/\delta)),
\end{align*}
where the final two inequalities follow from \Cref{eq:close_in_intervals} and \Cref{eq:TVD_on_W_Dcal}, respectively.
Hence,
\begin{align}
    &\phantom{\le}\tvdist{|W| - \Bin(n,\gamma)}
    \notag\\
    &\le O(\eps + \kappa\log(1/\delta) + \delta) + \sum_{\text{big } \Ical}\sum_{w\in \Ical} \left|\Pr\sbra{|W| = w} - \frac{\Pr\sbra{|W| \in \Ical}}{|\Ical|}\right|\notag \\
    &\qquad\: + \left|\frac{\Pr\sbra{|W| \in \Ical} - \Pr\sbra{\Bin(n,\gamma) \in \Ical}}{|\Ical|}\right| + \left|\Pr\sbra{\Bin(n,\gamma) = w} - \frac{\Pr\sbra{\Bin(n,\gamma) \in \Ical}}{|\Ical|}\right| \label{eq:decomposed_sum}.
\end{align}
Clearly, 
\begin{align}
    &\phantom{\le}\sum_{w\in \Ical}\left|\frac{\Pr\sbra{|W| \in \Ical} - \Pr\sbra{\Bin(n,\gamma) \in \Ical}}{|\Ical|}\right|
    \notag\\
    &= \left|\Pr\sbra{|W| \in \Ical} - \Pr\sbra{\Bin(n,\gamma) \in \Ical}\right|\notag \\
    &\le \left|\Pr\sbra{|W| \in \Ical} - \Pr\sbra{|F(\Ucal^m)| \in \Ical}\right| + \left|\Pr\sbra{|F(\Ucal^m)| \in \Ical} - \Pr\sbra{\Bin(n,\gamma) \in \Ical}\right|\notag \\
    &\le \left|\Pr\sbra{|W| \in \Ical} - \Pr\sbra{|F(\Ucal^m)| \in \Ical}\right| + O(\kappa^2). \tag{by \Cref{eq:close_in_intervals}}
\end{align}
Summing over all big intervals, the first term is at most $2\tvdist{W - F(\Ucal^m)} \le O(\eps)$ by \Cref{eq:TVD_on_W_Dcal}, and the second term is at most $O\pbra{\kappa \log(1/\delta)}$.

Additionally, note that \Cref{eq:lem:close_to_bin_weights_mixture:continuity} and the triangle inequality imply
\[
    \left|\Pr\sbra{|W| = w} - \frac{\Pr\sbra{|W| \in \Ical}}{|\Ical|}\right| \le \max_{y\in \Ical} \left|\Pr\sbra{|W| = w} - \Pr\sbra{|W| = y}\right| \le O_d\pbra{\frac{\kappa}{\sqrt{n}}}.
\]
Summing over all $w \in \Ical$ gives an upper bound of $O_d(\kappa^2)$, and further summing over big intervals gives $O_d(\kappa\log(1/\delta))$.
The sum of the 
\[
    \left|\Pr\sbra{\Bin(n,\gamma) = w} - \frac{\Pr\sbra{\Bin(n,\gamma) \in \Ical}}{|\Ical|}\right|
\]
terms can be bounded similarly, since
\[
    \max_{y\in\Ical}\left|\Pr\sbra{\Bin(n,\gamma) = w} - \Pr\sbra{\Bin(n,\gamma)=y}\right| \le O\pbra{\frac{\kappa \sqrt{n}}{\gamma(1-\gamma)n}} = O_d\pbra{\frac{\kappa}{\sqrt{n}}}. \tag{see \Cref{clm:bin_difference}}
\]

Combining these inequalities together with \Cref{eq:decomposed_sum}, we find 
\[
    \tvdist{|W| - \Bin(n,\gamma)} \le O_d(\eps + \delta + \kappa\log(1/\delta)).
\]
Set $\delta = k^{-1/5}$ and recall $\kappa = \sqrt{\polylog(k) / (\sqrt{k}\cdot \poly(\gamma))}$.
Again applying \Cref{eq:TVD_on_W_Dcal}, we conclude
\[
    \tvdist{|F(\Ucal^m)| - \Bin(n,\gamma)} \le \tvdist{|F(\Ucal^m)| - |W|} + \tvdist{|W| - \Bin(n,\gamma)} = O_d(\eps + k^{-1/5}).
\]

We now consider the case of $\gamma = 1/2$.
The proof is almost identical to the previous case, so we only highlight the relevant differences.
In this setting, \Cref{prop:continuity} only provides a continuity guarantee on weights differing by an even integer.
Hence, we refine the previously considered intervals into their even and odd parts.
Note this only changes the number of intervals in our analysis by a factor of two.
Moreover, \Cref{prop:kol_dist} (applied with $\delta = A^{-1/3}$) now provides a bound on the Kolmogorov distance between $|f(\Ucal^m)|$ and a mixture $\Mcal = \eta|\Deven| + (1-\eta)|\Dodd|$ of 
\[
    O\pbra{\frac{\log(A)^{O(d)^d}}{A^{1/3}} + \frac{\polylog(k)}{\sqrt{k}}} \text{\reflectbox{$\coloneqq$}~} \kappa^2.
\]
Carrying out the remaining steps with $\delta$ set to $k^{-1/5}$ as before, we conclude 
\begin{align*}
    \tvdist{|F(\Ucal^m)| - \Mcal} &\le \tvdist{|F(\Ucal^m)| - |W|} + \tvdist{|W| - \Mcal} \\
    &= O_d(\eps + \delta + \kappa\log(1/\delta)) \\
    &= O_d\pbra{\eps + k^{-1/5} + \log(k)\sqrt{\frac{\log(A)^{O(d)^d}}{A^{1/3}} + \frac{\polylog(k)}{\sqrt{k}}}}. \qedhere
\end{align*}
\end{proof}

We now use \Cref{lem:close_to_bin_weights_mixture} to prove the general case.
To obtain that lemma's required assumptions, we first condition on all input bits of large degree.
This will certainly result in the setting where no input bit affects many output bits, but using \Cref{prop:independence_after_cond}, it also lets us conclude the restricted functions generate distributions which resemble $\gamma$-biased product distributions.
For the final condition of \Cref{lem:close_to_bin_weights_mixture}, we need the symmetric distribution $\Dcal$ to be supported on weights around $\gamma n$.
Through a somewhat laborious but straightforward calculation, we show in \Cref{clm:gamma_restricted_TVD} that the restrictions of $f(\Ucal^m)$ with bias roughly $\gamma$ are relatively close to $\Dcal$ conditioned on its output weight being $\gamma n \pm n^{2/3}$, as desired.

\begin{lemma}\label{lem:combining_weights_primitive}
    Let $f\colon\bin^m\to\bin^n$ be a $d$-local function, and let $\eps \in (0,1)$.
    Assume $f(\Ucal^m)$ is $\eps$-close to a symmetric distribution $\Dcal$ over $\bin^n$ where $n$ is sufficiently large in terms of $d$ and $\eps$, and $\eps$ is sufficiently small in terms of $d$.
    Then the distribution over the Hamming weight of $f(\Ucal^m)$ is $O_d\pbra{\frac{1}{\log(1/\eps)}}^{1/5}$-close to a convex combination of the form
    \[
        \sum_{\substack{a\in [0,2^d] \cap \Zbb \\ 
        a \ne 2^{d-1}}} c_a\cdot\Bin(n, a/2^d) + c_e \cdot |\Deven| + c_o\cdot |\Dodd|.
    \]
\end{lemma}

\begin{proof}
    Let $k$ be the largest even integer less than $\log(1/\eps)/C_d$, where $C_d\ge1$ is a sufficiently large constant depending only on $d$.
    By choosing $\eps$ small enough, we may assume $k \ge 4$ (in order to later apply \Cref{fct:k-moments}).
    Define $S \subseteq [m]$ to be the set of input bits with degree at least $n/k$.
    Note that by the locality assumption, $|S| \le dk$.
    For each conditioning $\rho \in \bin^S$ on the bits in $S$, \Cref{prop:independence_after_cond} guarantees a subset $T_\rho \subseteq [n]$ of size $|T_\rho|\le O_{d,k}(1)$ and a bias $\gamma_\rho=a_\rho/2^d$, where $0 \le a_\rho \le 2^d$ is an integer, such that every $k$-tuple of output bits in $[n]\setminus T_\rho$ has distribution $\Ucal_{\gamma_\rho}^k$.

    We proceed by grouping the restricted functions according to their biases.
    More formally, we write $f(\Ucal^m)$ as the mixture
    \[
        f(\Ucal^m) = \sum_{\gamma} \Pr_\rho\sbra{\gamma_\rho = \gamma}\cdot f_\gamma(\Ucal^{[m]\setminus S}) \quad\text{where}\quad f_\gamma(\Ucal^{[m]\setminus S}) \coloneqq \E_{\rho} \sbra{f(\Ucal^{[m]\setminus S}, \rho) \mid \gamma_\rho = \gamma}.
    \]
    Similarly, we let $\Dcal_\gamma$ denote the distribution $\Dcal$ conditioned on the Hamming weight being $\gamma n \pm n^{2/3}$.
    Since $f(\Ucal^m)$ is close to $\Dcal$ by assumption, $f_\gamma(\Ucal^{[m]\setminus S})$ should be close to $\Dcal_\gamma$.
    We formalize this intuition in the following claim.
    \begin{claim}\label{clm:gamma_restricted_TVD}
        If $\Pr_\rho\sbra{\gamma_\rho = \gamma} > 0$, then $\tvdist{f_\gamma(\Ucal^{[m]\setminus S}) - \Dcal_\gamma} \le O_d(\sqrt{\eps})$.
    \end{claim}
    \noindent The proof of \Cref{clm:gamma_restricted_TVD} is routine but rather tedious, so we defer the details to \Cref{app:missing_sec:characterize}.

    For each $\gamma$, we combine \Cref{clm:gamma_restricted_TVD} and \Cref{lem:close_to_bin_weights_mixture} 
    to deduce
    \begin{align}
        \tvdist{|f_\gamma(\Ucal^{[m]\setminus S})| - \Pcal_\gamma} &\le O_d\pbra{\sqrt{\eps} + k^{-1/5} + \log(k)\sqrt{\frac{(\log(k))^{O(d)^d}}{k^{1/3}} + \frac{\polylog(k)}{\sqrt{k}}}} \notag \\
        &\le O_d\pbra{\sqrt{\eps} + \pbra{\frac{1}{\log(1/\eps)}}^{1/5}} \tag{since $k = \Theta_d(\log(1/\eps))$} \\
        &\le O_d\pbra{\frac{1}{\log(1/\eps)}}^{1/5},\label{eq:conditioned_weight_distance}
    \end{align}
    where
    \[
        \Pcal_\gamma = 
        \begin{cases}
            \Bin(n,\gamma) & \text{ if } \gamma \ne 1/2 \\
            \eta |\Deven| + (1-\eta) |\Dodd| & \text{ if } \gamma = 1/2
        \end{cases}
    \]
    for some $\eta \in [0,1]$.
    Define the mixture $\Pcal \coloneqq \sum_{\gamma} \Pr_\rho\sbra{\gamma_\rho = \gamma}\cdot \Pcal_\gamma$.
    Writing the output weight of $f$ as a convex combination of the conditionings, we find
    \begin{align*}
        \tvdist{|f(\Ucal^m)| - \Pcal} &= \tvdist{\sum_{\gamma} \Pr_\rho\sbra{\gamma_\rho = \gamma}\cdot |f_\gamma(\Ucal^{[m]\setminus S})| - \sum_{\gamma} \Pr_\rho\sbra{\gamma_\rho = \gamma} \cdot \Pcal_\gamma} \\
        &\le \sum_{\gamma} \Pr_\rho\sbra{\gamma_\rho = \gamma} \tvdist{|f_\gamma(\Ucal^{[m]\setminus S})| - \Pcal_\gamma} \tag{by triangle inequality} \\ 
        &\le O_d\pbra{\frac{1}{\log(1/\eps)}}^{1/5}. \tag{by \Cref{eq:conditioned_weight_distance}}
    \end{align*}
    This completes the proof of \Cref{lem:combining_weights_primitive}.
\end{proof}

\Cref{lem:combining_weights_primitive} states that the weight distribution of $f(\Ucal^m)$ must be close to a mixture of specific distributions.
The following lemma provides additional information about this mixture by describing the structure of the mixing weights.

\begin{lemma}\label{lem:combining_weights}
    Let $f\colon\bin^m\to\bin^n$ be a $d$-local function, and let $\eps \in (0,1)$.
    Assume $f(\Ucal^m)$ is $\eps$-close to a symmetric distribution $\Dcal$ over $\bin^n$ where $n$ is sufficiently large in terms of $d$ and $\eps$, and $\eps$ is sufficiently small in terms of $d$.
    Then the distribution over the Hamming weight of $f(\Ucal^m)$ is $O_d\pbra{\frac{1}{\log(1/\eps)}}^{1/5}$-close to a convex combination of the form
    \[
        \sum_{\substack{a\in [0,2^d] \cap \Zbb \\ 
        a \ne 2^{d-1}}} c_a\cdot\Bin(n, a/2^d) + c_e \cdot |\Deven| + c_o\cdot |\Dodd|,
    \]
    where each $c_a = c_a' / 2^C$ for some integer $0 \le c_a' \le 2^C$ and a fixed integer $C = O_d(1)$.
    Moreover, there exist at most $2^C$ many degree-$d$ $\Fbb_2$-polynomials $\{p_i \colon \Fbb_2^m \to \Fbb_2\}$, each with $O_d(n)$ monomials, such that 
    \[
        c_e = \frac{1}{2^C}\cdot \sum_i \Pr_{x\sim \Ucal^m}\sbra{p_i(x) = 0} \quad\text{and}\quad c_o = \frac{1}{2^C}\cdot \sum_i \Pr_{x\sim \Ucal^m}\sbra{p_i(x) = 1}.
    \]
\end{lemma}

\begin{proof}
    We know from \Cref{lem:combining_weights_primitive} that $f(\Ucal^m)$'s weight distribution is $O_d\pbra{\frac{1}{\log(1/\eps)}}^{1/5}$-close to a convex combination of the form
    \[
        \Pcal = \sum_{\substack{a\in [0,2^d] \cap \Zbb \\ 
        a \ne 2^{d-1}}} c_a\cdot\Bin(n, a/2^d) + c_e \cdot |\Deven| + c_o\cdot |\Dodd|,
    \]
    so it remains to reason about the mixing weights.
    To this end, let $S \subseteq [m]$ be the set of input bits with degree at least $n/2^{100d}$.
    Observe that this is a smaller set than the $S$ used in the proof of \Cref{lem:combining_weights_primitive}; this will ultimately provide stronger control over the mixing weights.
    
    For each conditioning $\sigma \in \bin^S$ on the bits in $S$, \Cref{lem:dyadic_weight_after_cond} guarantees an integer $0\le a_\sigma\le2^d$ such that
    \begin{equation}\label{eq:concentration_around_a_sigma}
        \Pr_{x\sim\Ucal^{[m]\setminus S}}\sbra{\abs{\frac{\abs{f(x,\sigma)}}n-\frac{a_\sigma}{2^d}}\ge\frac1{n^{1/(800d)}}}\le \poly(\eps).
    \end{equation}
    This also implies for any integer $b\ne a_\sigma$, we have
    \begin{align}
        \Pr_{x\sim\Ucal^{[m]\setminus S}}\sbra{\abs{\frac{\abs{f(x,\sigma)}}n-\frac{b}{2^d}}\le\frac1{n^{1/(800d)}}} &\le \Pr_{x\sim\Ucal^{[m]\setminus S}}\sbra{\abs{\frac{\abs{f(x,\sigma)}}n-\frac{a_\sigma}{2^d}}\ge \frac1{2^d} - \frac1{n^{1/(800d)}}} \notag \\
        &\le \Pr_{x\sim\Ucal^{[m]\setminus S}}\sbra{\abs{\frac{\abs{f(x,\sigma)}}n-\frac{a_\sigma}{2^d}} \ge \frac1{n^{1/(800d)}}} \le \poly(\eps). \label{eq:anticoncentration_around_b}
    \end{align}
    Define $p_\sigma \colon \bin^{[m]\setminus S} \to \bin$ to be $|f(x, \sigma)| \bmod 2$.
    Note that $p_\sigma$ is a sum (modulo 2) of $n$ output bits, each of which depends on at most $d$ input bits.
    Hence, $p_\sigma$ can be expressed as a degree-$d$ $\Fbb_2$-polynomial with $O_d(n)$ monomials.
    We will show that the conclusion of \Cref{lem:combining_weights} is satisfied by
    \begin{equation}\label{eq:qcal_defn}
        \Qcal = \sum_{\substack{a\in [0,2^d] \cap \Zbb \\ 
        a \ne 2^{d-1}}} \frac{c_a'}{2^{|S|}}\cdot\Bin(n, a/2^d) + \frac{c_e'}{2^{|S|}} \cdot |\Deven| + \frac{c_o'}{2^{|S|}}\cdot |\Dodd|,
    \end{equation}
    where
    \[
        c_a' = \#\cbra{\sigma : a_\sigma = a}, \quad c_e' = \sum_{\sigma : a_\sigma = 2^{d-1}} \Pr_{x\sim \Ucal^{[m]\setminus S}}\sbra{p_\sigma(x) = 0},\ \text{and}\ c_o' = \sum_{\sigma : a_\sigma = 2^{d-1}} \Pr_{x\sim \Ucal^{[m]\setminus S}}\sbra{p_\sigma(x) = 1}.
    \]

    The proof will proceed in two steps.
    First, we will show that $c_a$ is essentially the probability that $|f(\Ucal^m)|/n$ is close to $a/2^d$.
    Second, we will show that this probability is closely approximated by the fraction of conditionings which concentrate around $a$.
    (As usual, there are some additional considerations in the $a = 2^{d-1}$ case.)

    For notational convenience, define $c_{2^{d-1}} = c_e + c_o$ and
    \begin{itemize}
        \item $\delta_a = \Pr_{x \sim \Ucal^m}\sbra{\abs{\frac{|f(x)|}{n} - \frac{a}{2^d}} \le \frac{1}{n^{1/(800d)}}}$,
        \item $\delta_e = \Pr_{x \sim \Ucal^m}\sbra{\abs{\frac{|f(x)|}{n} - \frac{1}{2}} \le \frac{1}{n^{1/(800d)}} \land |f(x)| \text{ is even}}$,
        \item $\delta_o = \Pr_{x \sim \Ucal^m}\sbra{\abs{\frac{|f(x)|}{n} - \frac{1}{2}} \le \frac{1}{n^{1/(800d)}} \land |f(x)| \text{ is odd}}$.
    \end{itemize}
    
    \begin{claim}\label{clm:ca_close_to_deltaa}
        For any $0 \le a \le 2^d$, we have
        $\abs{c_a - \delta_a} \le O_d\pbra{\frac{1}{\log(1/\eps)}}^{1/5}$.
        Moreover, the same upper bound holds on $\abs{c_e - \delta_e}$ and $ \abs{c_o - \delta_o}$.
    \end{claim}
    \begin{proof}[Proof of \Cref{clm:ca_close_to_deltaa}]
        We first record a number of consequences of Hoeffding's inequality (\Cref{fct:hoeffding}) for sufficiently large $n$.
        We have
        \begin{equation*}
             \Pr_{w \sim \Bin(n, a/2^d)}\sbra{\abs{\frac{w}{n} - \frac{a}{2^d}} \le \frac{1}{n^{1/(800d)}}} \ge 1 - \exp\cbra{-2n^{1 - (1/(400d))}} \ge 1 - e^{-\sqrt{n}}
        \end{equation*}
        and
        \begin{align*}
            \Pr_{w\sim \Bin(n, 1/2)}\sbra{\abs{\frac{w}{n} - \frac{1}{2}} \le \frac{1}{n^{1/(800d)}} \mid w \text{ is even}} &= 1 - \Pr_{w\sim \Bin(n, 1/2)}\sbra{\abs{\frac{w}{n} - \frac{1}{2}} > \frac{1}{n^{1/(800d)}} \mid w \text{ is even}} \\
            &\ge 1 - \frac{\Pr_{w\sim \Bin(n, 1/2)}\sbra{\abs{\frac{w}{n} - \frac{1}{2}} > \frac{1}{n^{1/(800d)}}}}{\Pr_{w\sim \Bin(n, 1/2)}\sbra{w \text{ is even}}} \\
            &\ge 1 - 2\exp\cbra{-2n^{1 - (1/(400d))}} \ge 1 - e^{-\sqrt{n}}.
        \end{align*}
        Furthermore for $b \ne a$, 
        \begin{align*}
            \Pr_{w \sim \Bin(n, b/2^d)}\sbra{\abs{\frac{w}{n} - \frac{a}{2^d}} \le \frac{1}{n^{1/(800d)}}} &\le \Pr_{w \sim \Bin(n, b/2^d)}\sbra{\abs{\frac{w}{n} - \frac{b}{2^d}} \le \frac{1}{2^d} - \frac{1}{n^{1/(800d)}}} \\
            &\le \exp\cbra{-2n\pbra{\frac{1}{2^d} - \frac{1}{n^{1/(800d)}}}^2} \le e^{-\sqrt{n}}.
        \end{align*}
        Thus for any $0 \le a \le 2^d$, we have
        \begin{align*}
            \delta_a' \coloneqq \Pr_{w\sim \Pcal}\sbra{\abs{\frac{w}{n} - \frac{a}{2^d}} \le \frac{1}{n^{1/(800d)}}} &= c_a\cdot \Pr_{w\sim \Bin(n, a/2^d)}\sbra{\abs{\frac{w}{n} - \frac{a}{2^d}} \le \frac{1}{n^{1/(800d)}}} \\
            &\qquad + \sum_{b\ne a} c_b\cdot \Pr_{w\sim \Bin(n, b/2^d)}\sbra{\abs{\frac{w}{n} - \frac{a}{2^d}} \le \frac{1}{n^{1/(800d)}}} \\
            &\le c_a + e^{-\sqrt{n}}.
        \end{align*}
        Similarly, $\delta_a' \ge c_a - e^{-\sqrt{n}}$.
        Combining, we find that
        \begin{align*}
            \abs{c_a - \delta_a} \le \abs{c_a - \delta_a'} + \abs{\delta_a' - \delta_a} &\le e^{-\sqrt{n}} + \tvdist{\Pcal - |f(\Ucal^m)|} \le O_d\pbra{\frac{1}{\log(1/\eps)}}^{1/5}.
        \end{align*}

        For the case of $\delta_e$, we have
        \begin{align*}
            \delta_e' &\coloneqq \Pr_{w\sim \Pcal}\sbra{\abs{\frac{w}{n} - \frac{1}{2}} \le \frac{1}{n^{1/(800d)}} \land w \text{ is even}} \\
            &= c_e\cdot \Pr_{w\sim \Bin(n, 1/2)}\sbra{\abs{\frac{w}{n} - \frac{1}{2}} \le \frac{1}{n^{1/(800d)}} \mid w \text{ is even}} \\
            &\qquad + \sum_{b\ne 2^{d-1}} c_b\cdot \Pr_{w\sim \Bin(n, b/2^d)}\sbra{\abs{\frac{w}{n} - \frac{1}{2}} \le \frac{1}{n^{1/(800d)}} \land w \text{ is even}} \\
            &\le c_e + \sum_{b\ne 2^{d-1}} c_b\cdot \Pr_{w\sim \Bin(n, b/2^d)}\sbra{\abs{\frac{w}{n} - \frac{1}{2}} \le \frac{1}{n^{1/(800d)}}} \le c_e + e^{-\sqrt{n}}
        \end{align*}
        and
        \begin{align*}
            \delta_e' &\ge c_e\cdot \Pr_{w\sim \Bin(n, 1/2)}\sbra{\abs{\frac{w}{n} - \frac{1}{2}} \le \frac{1}{n^{1/(800d)}} \mid w \text{ is even}} \ge c_e - e^{-\sqrt{n}}.
        \end{align*}
        Thus,
        \begin{align*}
            \abs{c_e - \delta_e} \le \abs{c_e - \delta_e'} + \abs{\delta_e' - \delta_e} &\le e^{-\sqrt{n}} + \tvdist{\Pcal - |f(\Ucal^m)|} \le O_d\pbra{\frac{1}{\log(1/\eps)}}^{1/5}.
        \end{align*}
        The analysis of $\delta_o$ is essentially identical.
    \end{proof}

    Recall $p_\sigma(x) = |f(x, \sigma)| \bmod 2$.
    \begin{claim}\label{clm:deltaa_close_to_bias_count}
        For any $0 \le a \le 2^d$, we have $\abs{\delta_a - \frac{\#\cbra{\sigma : a_\sigma = a}}{2^{|S|}}} \le \poly(\eps).$
        Moreover, the same upper bound holds on 
        \[
            \abs{\delta_e - \frac{1}{2^{|S|}} \sum_{\sigma : a_\sigma = 2^{d-1}} \Pr_{x\sim \Ucal^{[m]\setminus S}}\sbra{p_\sigma(x) = 0}} \quad\text{and}\quad \abs{\delta_o - \frac{1}{2^{|S|}} \sum_{\sigma : a_\sigma = 2^{d-1}} \Pr_{x\sim \Ucal^{[m]\setminus S}}\sbra{p_\sigma(x) = 1}}.
        \]
    \end{claim}
    \begin{proof}[Proof of \Cref{clm:deltaa_close_to_bias_count}]
        We can express $\delta_a$ as
        \[
            \frac{1}{2^{|S|}}\Bigg(\sum_{\sigma : a_\sigma = a} \Pr_{x \sim \Ucal^{[m]\setminus S}}\sbra{\abs{\frac{|f(x, \sigma)|}{n} - \frac{a}{2^d}} \le n^{-\frac1{800d}}} + \sum_{\sigma : a_\sigma \ne a} \Pr_{x \sim \Ucal^{[m]\setminus S}}\sbra{\abs{\frac{|f(x, \sigma)|}{n} - \frac{a}{2^d}} \le n^{-\frac1{800d}}}\Bigg).
        \]
        If $a_\sigma = a$, then 
        \[
            1 - \poly(\eps) \le \Pr_{x \sim \Ucal^{[m]\setminus S}}\sbra{\abs{\frac{|f(x, \sigma)|}{n} - \frac{a}{2^d}} \le \frac{1}{n^{1/(800d)}}} \le 1
        \]
        by \Cref{eq:concentration_around_a_sigma}. 
        Additionally, if $a_\sigma \ne a$, then
        \[
            0 \le \Pr_{x \sim \Ucal^{[m]\setminus S}}\sbra{\abs{\frac{|f(x, \sigma)|}{n} - \frac{a}{2^d}} \le \frac{1}{n^{1/(800d)}}} \le \poly(\eps)
        \]
        by \Cref{eq:anticoncentration_around_b}.
        Thus,
        \begin{align*}
            \delta_a \le \frac{1}{2^{|S|}}\Bigg(\#\cbra{\sigma : a_\sigma = a} + \#\cbra{\sigma : a_\sigma \ne a}\cdot \poly(\eps) \Bigg) \le \frac{\#\cbra{\sigma : a_\sigma = a}}{2^{|S|}} + \poly(\eps).
        \end{align*}
        The lower bound on $\delta_a$ follows similarly.

        For the case of $\delta_e$, we have
        \begin{align*}
            \delta_e &= \frac{1}{2^{|S|}}\Bigg(\sum_{\sigma : a_\sigma = 2^{d-1}} \Pr_{x \sim \Ucal^{[m]\setminus S}}\sbra{\abs{\frac{|f(x, \sigma)|}{n} - \frac{1}{2}} \le \frac{1}{n^{1/(800d)}} \land |f(x, \sigma)| \text{ is even}} \\
            & \qquad \qquad + \sum_{\sigma : a_\sigma \ne 2^{d-1}} \Pr_{x \sim \Ucal^{[m]\setminus S}}\sbra{\abs{\frac{|f(x, \sigma)|}{n} - \frac{1}{2}} \le \frac{1}{n^{1/(800d)}}\land |f(x, \sigma)| \text{ is even}}\Bigg) \\
            &\le \frac{1}{2^{|S|}}\Bigg(\sum_{\sigma : a_\sigma = 2^{d-1}} \Pr_{x \sim \Ucal^{[m]\setminus S}}\sbra{|f(x, \sigma)| \text{ is even}} + \sum_{\sigma : a_\sigma \ne 2^{d-1}} \Pr_{x \sim \Ucal^{[m]\setminus S}}\sbra{\abs{\frac{|f(x, \sigma)|}{n} - \frac{1}{2}} \le \frac{1}{n^{1/(800d)}}}\Bigg) \\
            &\le \frac{1}{2^{|S|}}\sum_{\sigma : a_\sigma = 2^{d-1}} \Pr_{x \sim \Ucal^{[m]\setminus S}}\sbra{p_\sigma(x) = 0} + \poly(\eps)
        \end{align*}
        and
        \begin{align*}
            \delta_e &\ge \frac{1}{2^{|S|}}\sum_{\sigma : a_\sigma = 2^{d-1}} \Pr_{x \sim \Ucal^{[m]\setminus S}}\sbra{\abs{\frac{|f(x, \sigma)|}{n} - \frac{1}{2}} \le \frac{1}{n^{1/(800d)}} \land |f(x, \sigma)| \text{ is even}} \\
            &\ge \frac{1}{2^{|S|}}\sum_{\sigma : a_\sigma = 2^{d-1}} \pbra{\Pr_{x \sim \Ucal^{[m]\setminus S}}\sbra{|f(x, \sigma)| \text{ is even}} - \Pr_{x \sim \Ucal^{[m]\setminus S}}\sbra{\abs{\frac{|f(x, \sigma)|}{n} - \frac{1}{2}} > \frac{1}{n^{1/(800d)}}}} \\
            &\ge \frac{1}{2^{|S|}}\sum_{\sigma : a_\sigma = 2^{d-1}} \Pr_{x \sim \Ucal^{[m]\setminus S}}\sbra{p_\sigma(x) = 0} - \poly(\eps).
        \end{align*}
        The case of $\delta_o$ is essentially identical.
    \end{proof}

    Combining \Cref{clm:ca_close_to_deltaa} and \Cref{clm:deltaa_close_to_bias_count} gives
    \begin{equation*}
        \abs{c_a - \frac{\#\cbra{\sigma : a_\sigma = a}}{2^{|S|}}} \le O_d\pbra{\frac{1}{\log(1/\eps)}}^{1/5}
    \end{equation*}
    for all $0 \le a \le 2^d$, and
    \[
        \abs{c_e - \frac{1}{2^{|S|}} \sum_{\sigma : a_\sigma = 2^{d-1}} \Pr_{x\sim \Ucal^{[m]\setminus S}}\sbra{p_\sigma(x) = 0}} \le O_d\pbra{\frac{1}{\log(1/\eps)}}^{1/5}
    \]
    (and similarly for $c_o$).
    Hence $|f(\Ucal^m)|$ is $O_d\pbra{\frac{1}{\log(1/\eps)}}^{1/5}$-close to $\Qcal$ (defined in \Cref{eq:qcal_defn}), as desired.
\end{proof}

Now that we have the appropriate result for weight distributions, our main result \Cref{thm:main} quickly follows from \Cref{lem:distance_to_sym}.
We restate \Cref{thm:main} below for convenience.

\thmclassificationfull*

\begin{proof}
    We will prove 
    \begin{equation}\label{eq:final_to_prove}
        \tvdist{f(\Ucal^m) - \Mcal} \le C_d\cdot \pbra{\frac{1}{\log(1/\eps)}}^{1/5},
    \end{equation}
    where $\Mcal$ is a distribution of the form in the theorem statement, and $C_d \ge 1$ is a sufficiently large constant depending only on $d$.
    (Importantly, we will want $1/C_d$ to be at most the required upper bound on $\eps$ in the premise of \Cref{lem:combining_weights}.)
    We assume $\eps \le 1/C_d$, as otherwise the bound in \Cref{eq:final_to_prove} exceeds 1 and trivially holds.
    We also assume $d \ge 1$, as otherwise $f$ is a constant function, so we can set $\Mcal$ to be either the 0 or 1-biased product distribution.
    
    Combining our original distance assumption with \Cref{lem:distance_to_sym}, we have
    \[
        \eps \ge \tvdist{f(\Ucal^m) - \Dcal} = \Theta(\tvdist{|f(\Ucal^m)|-|\Dcal|} + \tvdist{f(\Ucal^m)-f(\Ucal^m)_\sym}).
    \]
    In particular, $\tvdist{f(\Ucal^m)-f(\Ucal^m)_\sym} = O(\eps)$.
    By \Cref{lem:combining_weights}, we have that the distribution over the Hamming weight of $f(\Ucal^m)$ is $O_d\pbra{\frac{1}{\log(1/\eps)}}^{1/5}$-close to a distribution $\Qcal$ of the form
    \[
        \sum_{\substack{a\in [0,2^d] \cap \Zbb \\ 
        a \ne 2^{d-1}}} c_a\cdot\Bin(n, a/2^d) + c_e \cdot |\Deven| + c_o\cdot |\Dodd|,
    \]
    where the mixing weights have the desired properties.
    Let $\Mcal$ be the symmetric distribution over $\bin^n$ with weight distribution $|\Mcal| = \Qcal$.
    Again applying \Cref{lem:distance_to_sym}, we have
    \begin{align*}
        \tvdist{f(\Ucal^m) - \Mcal} &= O\pbra{\tvdist{|f(\Ucal^m)| - \Qcal} + \tvdist{f(\Ucal^m)-f(\Ucal^m)_\sym}} = O_d\pbra{\frac{1}{\log(1/\eps)}}^{1/5}. \qedhere
    \end{align*}
\end{proof}

\section*{Acknowledgments}
We thank Farzan Byramji and anonymous reviewers for helpful comments on an earlier draft of this paper.

\bibliographystyle{alpha} 
\bibliography{biblio}

@article{viola2012complexity,
  title={The complexity of distributions},
  author={Viola, Emanuele},
  journal={SIAM Journal on Computing},
  volume={41},
  number={1},
  pages={191--218},
  year={2012},
  publisher={SIAM}
}

@inproceedings{filmus2023sampling,
  title={Sampling and Certifying Symmetric Functions},
  author={Filmus, Yuval and Leigh, Itai and Riazanov, Artur and Sokolov, Dmitry},
  booktitle={Approximation, Randomization, and Combinatorial Optimization. (APPROX/RANDOM)},
  year={2023}
}

@article{rossman2014monotone,
  title={The monotone complexity of k-clique on random graphs},
  author={Rossman, Benjamin},
  journal={SIAM Journal on Computing},
  volume={43},
  number={1},
  pages={256--279},
  year={2014},
  publisher={SIAM}
}

@misc{wiki:Vandermonde_matrix,
   author = "Wikipedia",
   title = "{Vandermonde matrix} --- {W}ikipedia{,} The Free Encyclopedia",
   year = "2025",
   howpublished = {\url{https://en.wikipedia.org/wiki/Vandermonde_matrix}},
   note = "[Online; accessed 3-November-2025]"
}

@article{alweiss2021improved,
  title={Improved bounds for the sunflower lemma},
  author={Alweiss, Ryan and Lovett, Shachar and Wu, Kewen and Zhang, Jiapeng},
  journal={Annals of Mathematics},
  volume={194},
  number={3},
  pages={795--815},
  year={2021},
  publisher={Princeton University and the Institute for Advanced Study}
}

@article{blumer1989learnability,
  title={Learnability and the {V}apnik-{C}hervonenkis dimension},
  author={Blumer, Anselm and Ehrenfeucht, Andrzej and Haussler, David and Warmuth, Manfred K},
  journal={Journal of the ACM (JACM)},
  volume={36},
  number={4},
  pages={929--965},
  year={1989},
  publisher={ACM New York, NY, USA}
}

@article{valiant1984theory,
  title={A theory of the learnable},
  author={Valiant, Leslie G},
  journal={Communications of the ACM},
  volume={27},
  number={11},
  pages={1134--1142},
  year={1984},
  publisher={ACM New York, NY, USA}
}

@inproceedings{kearns1994learnability,
  title={On the learnability of discrete distributions},
  author={Kearns, Michael and Mansour, Yishay and Ron, Dana and Rubinfeld, Ronitt and Schapire, Robert E and Sellie, Linda},
  booktitle={Proceedings of the twenty-sixth annual ACM symposium on Theory of computing},
  pages={273--282},
  year={1994}
}

@misc{cover2006elements,
  title={Elements of Information Theory},
  author={Cover, Thomas M and Thomas, Joy A},
  year={2006},
  publisher={Hoboken, NJ: Wiley-Interscience}
}

@InProceedings{watts2023unconditional,
  author =	{Bene Watts, Adam and Parham, Natalie},
  title =	{{Unconditional Quantum Advantage for Sampling with Shallow Circuits}},
  booktitle =	{17th Innovations in Theoretical Computer Science Conference (ITCS 2026)},
  pages =	{17:1--17:12},
  series =	{Leibniz International Proceedings in Informatics (LIPIcs)},
  ISBN =	{978-3-95977-410-9},
  ISSN =	{1868-8969},
  year =	{2026},
  volume =	{362},
  editor =	{Saraf, Shubhangi},
  publisher =	{Schloss Dagstuhl -- Leibniz-Zentrum f{\"u}r Informatik},
  address =	{Dagstuhl, Germany},
  URL =		{https://drops.dagstuhl.de/entities/document/10.4230/LIPIcs.ITCS.2026.17},
  URN =		{urn:nbn:de:0030-drops-253048},
  doi =		{10.4230/LIPIcs.ITCS.2026.17}
}

@phdthesis{haastad1986computational,
  title={Computational limitations for small depth circuits},
  author={H{\aa}stad, Johan},
  year={1986},
  school={Massachusetts Institute of Technology}
}

@article{jerrum1986random,
  title={Random generation of combinatorial structures from a uniform distribution},
  author={Jerrum, Mark R and Valiant, Leslie G and Vazirani, Vijay V},
  journal={Theoretical computer science},
  volume={43},
  pages={169--188},
  year={1986},
  publisher={Elsevier}
}

@article{babai1987random,
  title={Random oracles separate {PSPACE} from the polynomial-time hierarchy},
  author={Babai, L{\'a}szi{\'o}},
  journal={Information Processing Letters},
  volume={26},
  number={1},
  pages={51--53},
  year={1987},
  publisher={Elsevier}
}

@article{boppana1987one,
  title={One-way functions and circuit complexity},
  author={Boppana, Ravi B and Lagarias, Jeffrey C},
  journal={Information and Computation},
  volume={74},
  number={3},
  pages={226--240},
  year={1987},
  publisher={Elsevier}
}

@article{impagliazzo1996efficient,
  title={Efficient cryptographic schemes provably as secure as subset sum},
  author={Impagliazzo, Russell and Naor, Moni},
  journal={Journal of cryptology},
  volume={9},
  number={4},
  pages={199--216},
  year={1996},
  publisher={Springer}
}

@inproceedings{lovett2011bounded,
  title={Bounded-depth circuits cannot sample good codes},
  author={Lovett, Shachar and Viola, Emanuele},
  booktitle={2011 IEEE 26th Annual Conference on Computational Complexity},
  pages={243--251},
  year={2011},
  organization={IEEE}
}

@article{viola2014extractors,
  title={Extractors for circuit sources},
  author={Viola, Emanuele},
  journal={SIAM Journal on Computing},
  volume={43},
  number={2},
  pages={655--672},
  year={2014},
  publisher={SIAM}
}

@inproceedings{beck2012large,
  title={Large deviation bounds for decision trees and sampling lower bounds for {AC0}-circuits},
  author={Beck, Chris and Impagliazzo, Russell and Lovett, Shachar},
  booktitle={2012 IEEE 53rd Annual Symposium on Foundations of Computer Science},
  pages={101--110},
  year={2012},
  organization={IEEE}
}

@article{viola2020sampling,
  title={Sampling lower bounds: boolean average-case and permutations},
  author={Viola, Emanuele},
  journal={SIAM Journal on Computing},
  volume={49},
  number={1},
  pages={119--137},
  year={2020},
  publisher={SIAM}
}

@inproceedings{viola2023new,
  title={New sampling lower bounds via the separator},
  author={Viola, Emanuele},
  booktitle={38th Computational Complexity Conference (CCC 2023)},
  year={2023}
}

@article{de2012extractors,
  title={Extractors and lower bounds for locally samplable sources},
  author={De, Anindya and Watson, Thomas},
  journal={ACM Transactions on Computation Theory (TOCT)},
  volume={4},
  number={1},
  pages={1--21},
  year={2012},
  publisher={ACM New York, NY, USA}
}

@inproceedings{chattopadhyay2022space,
  title={The space complexity of sampling},
  author={Chattopadhyay, Eshan and Goodman, Jesse and Zuckerman, David},
  booktitle={13th Innovations in Theoretical Computer Science Conference,(ITCS 2022)},
  year={2022}
}

@inproceedings{viola2012extractors,
  title={Extractors for {T}uring-machine sources},
  author={Viola, Emanuele},
  booktitle={International Workshop on Approximation Algorithms for Combinatorial Optimization},
  pages={663--671},
  year={2012},
  organization={Springer}
}

@inproceedings{cohen2016extractors,
  title={Extractors for near logarithmic min-entropy},
  author={Cohen, Gil and Schulman, Leonard J},
  booktitle={2016 IEEE 57th Annual Symposium on Foundations of Computer Science (FOCS)},
  pages={178--187},
  year={2016},
  organization={IEEE}
}

@inproceedings{chattopadhyay2016explicit,
  title={Explicit two-source extractors and resilient functions},
  author={Chattopadhyay, Eshan and Zuckerman, David},
  booktitle={Proceedings of the forty-eighth annual ACM symposium on Theory of Computing},
  pages={670--683},
  year={2016}
}

@article{viola2012bit,
  title={Bit-Probe Lower Bounds for Succinct Data Structures},
  author={Viola, Emanuele},
  journal={SIAM Journal on Computing},
  volume={41},
  number={6},
  pages={1593},
  year={2012},
  publisher={Society for Industrial and Applied Mathematics}
}

@inproceedings{daskalakis2013learning,
  title={Learning sums of independent integer random variables},
  author={Daskalakis, Constantinos and Diakonikolas, Ilias and O'Donnell, Ryan and Servedio, Rocco A and Tan, Li-Yang},
  booktitle={2013 IEEE 54th Annual Symposium on Foundations of Computer Science},
  pages={217--226},
  year={2013},
  organization={IEEE}
}

@inproceedings{daskalakis2014faster,
  title={Faster and sample near-optimal algorithms for proper learning mixtures of {G}aussians},
  author={Daskalakis, Constantinos and Kamath, Gautam},
  booktitle={Conference on Learning Theory},
  pages={1183--1213},
  year={2014},
  organization={PMLR}
}

@inproceedings{dinur2006fourier,
  title={On the {F}ourier tails of bounded functions over the discrete cube},
  author={Dinur, Irit and Friedgut, Ehud and Kindler, Guy and O'Donnell, Ryan},
  booktitle={Proceedings of the thirty-eighth annual ACM symposium on Theory of computing},
  pages={437--446},
  year={2006}
}

@inproceedings{kane2024locality,
  title={Locality Bounds for Sampling {H}amming Slices},
  author={Kane, Daniel M and Ostuni, Anthony and Wu, Kewen},
  booktitle={Proceedings of the 56th Annual ACM Symposium on Theory of Computing},
  pages={1279--1286},
  year={2024}
}

@inproceedings{kane2024locally2,
  title={Locally Sampleable Uniform Symmetric Distributions},
  author={Kane, Daniel M and Ostuni, Anthony and Wu, Kewen},
  booktitle={Proceedings of the 57th Annual ACM Symposium on Theory of Computing},
  pages={1807--1816},
  year={2025}
}

@inproceedings{bellare1994randomness,
  title={Randomness-efficient oblivious sampling},
  author={Bellare, Mihir and Rompel, John},
  booktitle={Proceedings 35th Annual Symposium on Foundations of Computer Science},
  pages={276--287},
  year={1994},
  organization={IEEE}
}

@inproceedings{shaltiel2024explicit,
  title={Explicit codes for poly-size circuits and functions that are hard to sample on low entropy distributions},
  author={Shaltiel, Ronen and Silbak, Jad},
  booktitle={Proceedings of the 56th Annual ACM Symposium on Theory of Computing},
  pages={2028--2038},
  year={2024}
}

@inproceedings{grinberg2018indistinguishability,
  title={Indistinguishability by adaptive procedures with advice, and lower bounds on hardness amplification proofs},
  author={Grinberg, Aryeh and Shaltiel, Ronen and Viola, Emanuele},
  booktitle={2018 IEEE 59th Annual Symposium on Foundations of Computer Science (FOCS)},
  pages={956--966},
  year={2018},
  organization={IEEE}
}

@inproceedings{chattopadhyay2020xor,
  title={{XOR} lemmas for resilient functions against polynomials},
  author={Chattopadhyay, Eshan and Hatami, Pooya and Hosseini, Kaave and Lovett, Shachar and Zuckerman, David},
  booktitle={Proceedings of the 52nd Annual ACM SIGACT Symposium on Theory of Computing},
  pages={234--246},
  year={2020}
}

@article{diakonikolas2010bounded,
  title={Bounded independence fools halfspaces},
  author={Diakonikolas, Ilias and Gopalan, Parikshit and Jaiswal, Ragesh and Servedio, Rocco A and Viola, Emanuele},
  journal={SIAM Journal on Computing},
  volume={39},
  number={8},
  pages={3441--3462},
  year={2010},
  publisher={SIAM}
}

@article{diakonikolas2016learning,
  title={Learning Structured Distributions.},
  author={Diakonikolas, Ilias},
  journal={Handbook of Big Data},
  volume={267},
  pages={10--1201},
  year={2016}
}

@article{yatracos1985rates,
  title={Rates of convergence of minimum distance estimators and Kolmogorov's entropy},
  author={Yatracos, Yannis G},
  journal={The Annals of Statistics},
  volume={13},
  number={2},
  pages={768--774},
  year={1985},
  publisher={Institute of Mathematical Statistics}
}

@inproceedings{aliakbarpour2016learning,
  title={Learning and testing junta distributions},
  author={Aliakbarpour, Maryam and Blais, Eric and Rubinfeld, Ronitt},
  booktitle={Conference on Learning Theory},
  pages={19--46},
  year={2016},
  organization={PMLR}
}

@inproceedings{lindsay1995mixture,
  title={Mixture Models: Theory, Geometry and Applications},
  author={Lindsay, Bruce G},
  booktitle={NSF-CBMS Regional Conference Series in Probability and Statistics},
  pages={i--163},
  year={1995},
  organization={JSTOR}
}

@article{birge1987estimating,
  title={Estimating a density under order restrictions: Nonasymptotic minimax risk},
  author={Birg{\'e}, Lucien},
  journal={The Annals of Statistics},
  pages={995--1012},
  year={1987},
  publisher={JSTOR}
}

@article{dumbgen2009maximum,
  title={Maximum likelihood estimation of a log-concave density and its distribution function: Basic properties and uniform consistency},
  author={D{\"u}mbgen, Lutz and Rufibach, Kaspar},
  journal={Bernoulli},
  pages={40--68},
  year={2009},
  publisher={JSTOR}
}

@inproceedings{gopalan2010fooling,
  title={Fooling functions of halfspaces under product distributions},
  author={Gopalan, Parikshit and O'Donnell, Ryan and Wu, Yi and Zuckerman, David},
  booktitle={2010 IEEE 25th Annual Conference on Computational Complexity},
  pages={223--234},
  year={2010},
  organization={IEEE}
}

@inproceedings{yu2024sampling,
  title={Sampling, flowers and communication},
  author={Yu, Huacheng and Zhan, Wei},
  booktitle={15th Innovations in Theoretical Computer Science Conference (ITCS 2024)},
  pages={100--1},
  year={2024},
  organization={Schloss Dagstuhl--Leibniz-Zentrum f{\"u}r Informatik}
}

@InProceedings{grier2025quantum,
  author =	{Grier, Daniel and Kane, Daniel M. and Morris, Jackson and Ostuni, Anthony and Wu, Kewen},
  title =	{{Quantum Advantage from Sampling Shallow Circuits: Beyond Hardness of Marginals}},
  booktitle =	{17th Innovations in Theoretical Computer Science Conference (ITCS 2026)},
  pages =	{73:1--73:14},
  series =	{Leibniz International Proceedings in Informatics (LIPIcs)},
  ISBN =	{978-3-95977-410-9},
  ISSN =	{1868-8969},
  year =	{2026},
  volume =	{362},
  editor =	{Saraf, Shubhangi},
  publisher =	{Schloss Dagstuhl -- Leibniz-Zentrum f{\"u}r Informatik},
  address =	{Dagstuhl, Germany},
  URL =		{https://drops.dagstuhl.de/entities/document/10.4230/LIPIcs.ITCS.2026.73},
  URN =		{urn:nbn:de:0030-drops-253607},
  doi =		{10.4230/LIPIcs.ITCS.2026.73}
}

@inproceedings{alekseev2025sampling,
  title={Sampling Permutations with Cell Probes is Hard},
  author={Yaroslav Alekseev and Mika G\"o\"os and Konstantin Myasnikov and Artur Riazanov and Dmitry Sokolov},
  booktitle={Proceedings of the 58th Annual ACM Symposium on Theory of Computing},
  year={2026}
}

@article {FSS84,
    AUTHOR = {Furst, Merrick and Saxe, James B. and Sipser, Michael},
     TITLE = {Parity, circuits, and the polynomial-time hierarchy},
   JOURNAL = {Math. Systems Theory},
  FJOURNAL = {Mathematical Systems Theory. An International Journal on
              Mathematical Computing Theory},
    VOLUME = {17},
      YEAR = {1984},
    NUMBER = {1},
     PAGES = {13--27},
      ISSN = {0025-5661},
   MRCLASS = {68Q25 (94C15)},
  MRNUMBER = {738749},
MRREVIEWER = {Claus-Peter\ Schnorr},
       DOI = {10.1007/BF01744431},
       URL = {https://doi.org/10.1007/BF01744431},
}

@article {Ajt83,
    AUTHOR = {Ajtai, M.},
     TITLE = {{$\Sigma \sp{1}\sb{1}$}-formulae on finite structures},
   JOURNAL = {Ann. Pure Appl. Logic},
  FJOURNAL = {Annals of Pure and Applied Logic},
    VOLUME = {24},
      YEAR = {1983},
    NUMBER = {1},
     PAGES = {1--48},
      ISSN = {0168-0072,1873-2461},
   MRCLASS = {03C13 (03B10 03B15 03C35)},
  MRNUMBER = {706289},
MRREVIEWER = {A.\ H.\ Lachlan},
       DOI = {10.1016/0168-0072(83)90038-6},
       URL = {https://doi.org/10.1016/0168-0072(83)90038-6},
}

@inproceedings{yao1985separating,
  title={Separating the polynomial-time hierarchy by oracles},
  author={Yao, Andrew Chi-Chih},
  booktitle={26th Annual Symposium on Foundations of Computer Science (sfcs 1985)},
  pages={1--10},
  year={1985},
  organization={IEEE}
}

@inproceedings{hastad1986almost,
  title={Almost optimal lower bounds for small depth circuits},
  author={H{\aa}stad, Johan},
  booktitle={Proceedings of the eighteenth annual ACM symposium on Theory of computing},
  pages={6--20},
  year={1986}
}

\appendix
\section{An Example Towards the Exact Characterization}\label{app:open_problems}

In this appendix, we provide the full details behind the example described in \Cref{sec:open_prob}.
Consider the distribution $\Pcal = \Ucal_{1/4}^n + 2^{-n-1}\Deven - 2^{-n-1}\Dodd$.
(Recall $\Pcal$ is not a mixture, but rather the distribution that assigns $x$ probability $\Ucal_{1/4}^n(x) + 2^{-n-1}\Deven(x) - 2^{-n-1}\Dodd(x)$.)

\begin{claim}\label{clm:conj_exp}
    $\Pcal$ is not a distribution of the form given by \Cref{thm:main}.
\end{claim}

\begin{proof}
    Suppose by contradiction there exists a mixture of the form specified by \Cref{thm:main}:
    \begin{equation}\label{eq:clm:conj_exp_1}
        \Qcal = \sum_{\substack{a\in [0,2^d] \cap \Zbb \\ 
        a \ne 2^{d-1}}} c_a\cdot\Ucal_{a/2^d}^n + c_e \cdot \Deven + c_o\cdot \Dodd
    \end{equation}
    where $\Pcal = \Qcal$ and $\{c_a\},c_e,c_o$ are nonnegative values summing to $1$.
    For a subset $S \subseteq [n]$, consider the parity function $\chi_S(x) = (-1)^{\sum_{i\in S} x_i}$.
    Observe that for any set $S$ and bias $\gamma$, we have
    \[
        \E_{x\sim \Ucal_\gamma^n}[\chi_S(x)] = \prod_{i\in S}\E_{x\sim \Ucal_\gamma^n}[(-1)^{x_i}] = \pbra{1-2\gamma}^{|S|}.
    \]
    Additionally, 
    \begin{itemize}
        \item If $|S| = 0$, we have $\E_{x\sim \Deven}[\chi_S(x)] = \E_{x\sim \Dodd}[\chi_S(x)] = 1$,
        \item If $0 < |S| < n$, we have $\E_{x\sim \Deven}[\chi_S(x)] = \E_{x\sim \Dodd}[\chi_S(x)] = 0$,
        \item If $|S| = n$, we have $\E_{x\sim \Deven}[\chi_S(x)] = 1$ and $\E_{x\sim \Dodd}[\chi_S(x)] = -1$.
    \end{itemize}
    Thus,
    \[
        \E_{x\sim \Pcal}[\chi_S(x)] = \E_{x\sim \Ucal_{1/4}^n}[\chi_S(x)] = \frac{1}{2^{|S|}}
    \]
    for every $S$ of size $0 < |S| < n$.
    In order for $\Pcal = \Qcal$, it must be that for every such $S$, we have
    \begin{equation}\label{eq:VDMD}
        \sum_{\substack{a\in [0,2^d] \cap \Zbb \\ 
        a \ne 2^{d-1}}} c_a\cdot \pbra{1 - \frac{a}{2^{d-1}}}^{|S|} = \E_{x\sim \Qcal}[\chi_S(x)] = \E_{x\sim \Pcal}[\chi_S(x)] = \frac{1}{2^{|S|}}.
    \end{equation}
    
    Let $b_a = 1 - \frac{a}{2^{d-1}}$.
    We claim that $c_e + c_o = 0$, so \Cref{eq:VDMD} also holds for $S = \emptyset$.
    Indeed, applying the Cauchy-Schwarz inequality to the established $|S|=1,2$ cases gives
    \[
        \frac{1}{4} \le \pbra{\sum_a c_a b_a}^2 \le \pbra{\sum_a c_a} \cdot \pbra{\sum_a c_a b_a^2} \le \frac{1}{4} \pbra{\sum_a c_a}.
    \]
    That is, $\sum_a c_a = 1$.
    Since we assume \Cref{eq:clm:conj_exp_1} is a mixture, we know that $\{c_a\}$'s and $c_e,c_o$ are nonnegative values that sum to $1$; this enforces $c_e=c_o=0$, as claimed.
    
    We now consider the left-hand side of \Cref{eq:VDMD} for $|S| = 0,1, \dots, 2^d - 1$, this corresponds to the $2^d \times 2^d$ Vandermonde matrix
    \begin{equation*}
        V = \begin{pmatrix}
            1 & b_0 & b_0^2 & \cdots & b_0^{2^d - 1} \\
            1 & b_1 & b_1^2 & \cdots & b_1^{2^d - 1} \\
            1 & b_2 & b_2^2 & \cdots & b_2^{2^d - 1} \\
            \vdots & \vdots & \vdots & \ddots & \vdots \\
            1 & b_{2^d} & b_{2^d}^2 & \cdots & b_{2^d}^{2^d - 1}.
        \end{pmatrix}.
    \end{equation*}
    Note that $V$ does not have a row corresponding to $b_{2^{d-1}}$ as $a=2^{d-1}$ will be handled by $\Deven$ and $\Dodd$.
    Since the $b_a$'s are all distinct, $V$ is invertible (see, e.g., \cite{wiki:Vandermonde_matrix}), so the $c_a$'s are uniquely determined to be $c_a = 1$ for $a = 2^{d-2}$ and 0 otherwise.
    Recall we already established $c_e = c_o = 0$, so this implies $\Qcal = \Ucal_{1/4}^n$, creating the contradiction
    \[
        \frac{1}{2^{n}} = \E_{x\sim \Ucal_{1/4}^n}[\chi_{[n]}(x)] = \E_{x\sim \Qcal}[\chi_{[n]}(x)] = \E_{x\sim \Pcal}[\chi_{[n]}(x)] = \frac{1}{2^{n}} + \frac{1}{2^{n+1}} - \frac{-1}{2^{n+1}} =  \frac{1}{2^{n-1}}. \qedhere
    \]
\end{proof}

\begin{claim}
    $\Pcal$ can be sampled exactly by a 3-local function.
\end{claim}
\begin{proof}
    Let $\Dall$ be the uniform distribution over $\bin^n$. 
    Consider the distribution $\Qcal$ defined by sampling $x \sim \Deven$ and $y \sim \Dall$, and returning the bitwise AND $z = x \land y$. 
    Clearly, $\Qcal$ can be sampled with a 3-local function (where $x\sim\Deven$ is sampled using the telescoping construction discussed in the introduction).
    We will show $\Pcal = \Qcal$.

    For a nonnegative integer $k$, observe that
    \begin{align*}
        \Pr_{z\sim \Qcal} \sbra{z = 1^k 0^{n-k}} &= \sum_{S \subseteq [n-k]} \Pr_{x\sim \Deven} \sbra{x = 1^k S} \Pr_{y\sim \Dall} \sbra{y = 1^k \cdots, y|_S = 0^{|S|}} \\
        &= \sum_{S \subseteq [n-k]} \frac{1+(-1)^{k+|S|}}{2^n} \cdot \frac{1}{2^{k + |S|}} \\
        &= \sum_{S \subseteq [n-k]} \pbra{\frac{1}{2^{n+k}} \cdot \frac{1}{2^{|S|}}} + (-1)^k \sum_{S \subseteq [n-k]}\pbra{\frac{1}{2^{n+k}} \pbra{-\frac{1}{2}}^{|S|}} \\
        &= \frac{1}{2^{n+k}}\cdot \pbra{\frac{3}{2}}^{n-k} + \frac{(-1)^k}{2^{n+k}}\cdot\pbra{\frac{1}{2}}^{n-k} \\
        &= \frac{3^{n-k} + (-1)^k}{4^n} = \Pr_{z\sim \Pcal} \sbra{z = 1^k 0^{n-k}}. 
    \end{align*}
    Since both $\Pcal$ and $\Qcal$ are symmetric distributions, the above calculation implies they must be equal.
\end{proof}

\section{Missing Proofs in Section \ref{sec:prelim}}\label{sec:app_prelim}

\subsection{Proof of Fact \ref{fct:bin_interval}}

\begin{proof}[Proof of \Cref{fct:bin_interval}]
    It suffices to assume $\gamma\in(0,1/2]$ and show for any integer $k$, we have
    \begin{equation}\label{eq:clm:bin_interval_1}
    \Pr\sbra{\Bin(n, \gamma)=k} \le \frac{O(1)}{\sqrt{\gamma n}}.
    \end{equation}
    To this end, we draw samples from $\Bin(n,\gamma)$ in the following way.
    \begin{itemize}
        \item For each $i\in[n]$, sample an unbiased random coin $B_i\in\bin$ and sample an independent $W_i\in\bin$ with probability $\Pr[W_i=1]=2\gamma$ and $\Pr[W_i=0]=1-2\gamma$.
        \item Define $X_i=W_i\cdot B_i$ for each $i\in[n]$. Then output $\sum_{i\in[n]}X_i$.
    \end{itemize}
    
    Now define $\Ecal$ to be the event that $\sum_{i\in[n]}W_i\le\gamma n$.
    Then by \Cref{fct:chernoff} with $\delta=1/2$ and $\mu=2\gamma n$, we have
    \begin{equation}\label{eq:clm:bin_interval_2}
        \Pr[\Ecal]\le e^{-\gamma n/4}.
    \end{equation}
    For fixed $W=(W_1,\ldots,W_n)$ under which $\Ecal$ does not happen, let $S=\cbra{i\in[n] : W_i=1}$.
    Then,
    \begin{align}
        \Pr\sbra{\sum_{i\in[n]}X_i=k\mid W}
        &=\Pr\sbra{\sum_{i\in S}B_i=k}
        =\Pr\sbra{\Bin(|S|,1/2)=k}
        \le\frac{O(1)}{\sqrt{|S|}}\le\frac{O(1)}{\sqrt{\gamma n}},
        \label{eq:clm:bin_interval_3}
    \end{align}
    where we use $|S|\ge\gamma n$ for the last inequality.
    
    Now we prove \Cref{eq:clm:bin_interval_1}:
    \begin{align*}
        \text{LHS of \Cref{eq:clm:bin_interval_1}}
        &=\Pr\sbra{\sum_{i\in[n]}X_i=k}
        \le\Pr[\Ecal]+\Pr\sbra{\sum_{i\in[n]}X_i=k\mid\neg\Ecal}\\
        &\le e^{-\gamma n/4}+\frac{O(1)}{\sqrt{\gamma n}}
        \tag{by \Cref{eq:clm:bin_interval_2} and \Cref{eq:clm:bin_interval_3}}\\
        &\le\frac{O(1)}{\sqrt{\gamma n}}=\text{RHS of \Cref{eq:clm:bin_interval_1}}.
        \tag*{\qedhere}
    \end{align*}
\end{proof}

\subsection{Proof of Fact \ref{clm:bin_difference}}
\begin{proof}[Proof of \Cref{clm:bin_difference}]
    By potentially swapping $a$ and $b$ or by replacing $(\gamma, a, b)$ with $(1-\gamma, n-a, n-b)$, we may assume $b \ge a$ and $\Pr\sbra{\Bin(n,\gamma) = a} \ge \Pr\sbra{\Bin(n,\gamma) = b}$.
    From here, we proceed similarly to the proof of \Cref{fct:bin_interval}:
    define $X_i = W^{(X)}_i \cdot B^{(X)}_i$ and $Y_i = W^{(Y)}_i \cdot B^{(Y)}_i$, where each $B^{(\cdot)}_i\in \bin$ is an independent unbiased random coin, and each $W^{(\cdot)}_i \in \bin$ is an independent random variable satisfying $\Pr[W^{(\cdot)}_i = 1] = 2\gamma$ and $\Pr[W^{(\cdot)}_i = 0] = 1-2\gamma$.
    For each $w \in \bin^n$, we additionally define $S_w = \cbra{i\in[n] : w_i=1}$.
    Then,
    \begin{align*}
        \phantom{=}& \Pr\sbra{\Bin(n,\gamma) = a} - \Pr\sbra{\Bin(n,\gamma) = b} \\
        =& \sum_w \Pr\sbra{W^{(X)} = w}\Pr\sbra{\sum_{i\in[n]}X_i=a\mid W^{(X)}=w}\\
        &\qquad\qquad- \sum_w\Pr\sbra{W^{(Y)} = w}\Pr\sbra{\sum_{i\in[n]}Y_i=b\mid W^{(Y)}=w} \\
        =& \sum_w\Pr\sbra{W^{(X)} = w}\pbra{\Pr\sbra{\sum_{i\in S_w}B_i^{(X)}=a} - \Pr\sbra{\sum_{i\in S_w}B_i^{(Y)}=b}} \\
        =& \sum_w\Pr\sbra{W^{(X)} = w}\pbra{\Pr\sbra{\Bin(|S_w|,1/2)=a} - \Pr\sbra{\Bin(|S_w|,1/2)=b}} \\
        \le& \Pr\sbra{W^{(X)} = 0^n} + \sum_{w\ne 0^n}\Pr\sbra{W^{(X)} = w}\cdot O\pbra{\frac{b-a}{|S_w|}},
    \end{align*}
    where the final inequality is somewhat standard (see, e.g., \cite[Fact A.3]{kane2024locally2}).

    For the remainder of the argument we will assume $\gamma \in (0,1/2]$; the remaining case is similar.
    Define $\Ecal$ to be the event that $\sum_{i\in[n]}W^{(X)}_i\le\gamma n$.
    Then by \Cref{fct:chernoff} with $\delta=1/2$ and $\mu=2\gamma n$, we have
    \begin{equation*}
        \Pr[\Ecal]\le e^{-\gamma n/4},
    \end{equation*}
    so we may continue our calculation as follows:
    \begin{align*}
        &\phantom{\le}\Pr\sbra{W^{(X)} = 0^n} + \sum_{w\ne 0^n}\Pr\sbra{W^{(X)} = w}\cdot O\pbra{\frac{b-a}{|S_w|}}\\
        &\le \Pr[\Ecal] + \sum_{w : |S_w| \ge \gamma n}\Pr\sbra{W^{(X)} = w}\cdot O\pbra{\frac{b-a}{|S_w|}} \\
        &\le e^{-\gamma n/4} + O\pbra{\frac{b-a}{\gamma n}} \\
        &\le O\pbra{\frac{b-a}{\gamma(1-\gamma)n}}. \qedhere
    \end{align*}
\end{proof}

\section{Missing Proofs in Section \ref{sec:characterize}}\label{app:missing_sec:characterize}

\subsection{Missing Proofs in Subsection \ref{subsec:matching_moments}}

\subsubsection{Proof of Claim \ref{clm:lem:dyadic_weight_after_cond_1}}
\begin{proof}[Proof of \Cref{clm:lem:dyadic_weight_after_cond_1}]
Assume $\err(p,d)>2\cdot2^{-30d}$.
Then for any $x\in\bin^{m-|S|}$ with $\abs{\frac{|f(x,\rho)|}n-p}\le2^{-30d}$, we have $\err\pbra{\frac{|f(x,\rho)|}n,d}\ge2^{-30d}$.
Combining with \Cref{eq:lem:dyadic_weight_after_cond_2}, we then have
\begin{align*}
\Pr_{y\sim\bin^m}\sbra{\err\pbra{\frac{|f(y)|}n,d}\ge2^{-30d}}
&\ge\Pr\sbra{y_S=\rho}\cdot\Pr_{x\sim\bin^{[m]\setminus S}}\sbra{\abs{\frac{|f(x,\rho)|}n-p}\le2^{-30d}}\ge2^{-|S|-1}.
\end{align*}
Recall $|S| \le dA$ and $\eps < 2^{-cdA}$ for some large constant $c > 0$.
Thus for sufficiently large $n$, this contradicts \Cref{lem:dyadic_weight_eps}.
\end{proof}

\subsubsection{Proof of Claim \ref{clm:lem:dyadic_weight_after_cond_2}}
\begin{proof}[Proof of \Cref{clm:lem:dyadic_weight_after_cond_2}]
Observe that
\begin{align*}
\Pr_{x\sim\bin^{[m]\setminus S}}\sbra{\err\pbra{\frac{|f(x,\rho)|}n,d}>\frac1{n^{1/(800d)}}}
&=\Pr_{y\sim\bin^m}\sbra{\err\pbra{\frac{|f(y)|}n,d}>\frac1{n^{1/(800d)}}\mid y_S=\rho}\\
&\le2^{|S|}\cdot\Pr_{y\sim\bin^m}\sbra{\err\pbra{\frac{|f(y)|}n,d}>\frac1{n^{1/(800d)}}}\\
&\le 2^{dA}\cdot O\pbra{\eps+e^{-n^{0.9}}}
\tag{by \Cref{lem:dyadic_weight_eps}} \\
&\le \poly(\eps) \tag{since $\eps < 2^{-cdA}$}
\end{align*}
for sufficiently large $n$.
\end{proof}

\subsubsection{Proof of Claim \ref{clm:lem:dyadic_weight_after_cond_4}}
\begin{proof}[Proof of \Cref{clm:lem:dyadic_weight_after_cond_4}]
Observe that the LHS event has the following two cases:
\begin{itemize}
\item $\err\pbra{\frac{|f(x,\rho)|}n,d}>n^{-1/(800d)}$. By \Cref{clm:lem:dyadic_weight_after_cond_2}, this happens with probability $\poly(\eps)$.
\item $\err\pbra{\frac{|f(x,\rho)|}n,d}\le n^{-1/(800d)}$ but $\abs{\frac{|f(x,\rho)|}n-\frac a{2^d}}>\frac1{n^{1/(800d)}}$. Then it means $\abs{\frac{|f(x,\rho)|}n-\frac{a'}{2^d}}\le n^{-1/(800d)}$ for some $a'\ne a$.
Then
$$
\abs{\frac{|f(x,\rho)|}n-\frac a{2^d}}\ge 2^{-d}-n^{-1/(800d)}>4^{-d}
$$
as $n$ is sufficiently large in terms of $d$. This happens with probability at most $\delta$.
\qedhere
\end{itemize}
\end{proof}

\subsection{Missing Proofs in Subsection \ref{subsec:approx_continuity}}

For the convenience of the reader, we include a full proof of \Cref{lem:anticoncentration_after_coupling_all}.
We emphasize that aside from minor modifications to handle the case of $\gamma = 1/2$ and relax some quantitative dependencies, the proof is essentially copied verbatim from \cite{kane2024locality}.

\begin{proof}[Proof of \Cref{lem:anticoncentration_after_coupling_all}]
    We assume without loss of generality $\gamma \in (0,1/2]$ by flipping the zeros and ones of $(X,Y,Z,W)$ if necessary.
    Observe that this preserves the congruence.
    If $t=1$ then we have that $\Pr\sbra{X+|Y|=Z+|W|}=\Pr\sbra{X=Z}$ as $q\ge2$.
    Since $X$ and $Z$ are independent copies of $\Ucal_\gamma^1$, we find
    \begin{equation}\label{eq:lem:anticoncentration_after_coupling_2}
    \Pr\sbra{X=Z}=\Pr\sbra{X=1}^2+(1-\Pr\sbra{X=1})^2 = \gamma^2+(1-\gamma)^2 < 1,
    \end{equation}
    where we use the fact that $\gamma\in(0,1/2]$.
    
    Now we assume $t, q\ge2$.
    Expand $\Pr\sbra{X+|Y|\equiv Z+|W|\Mod q}$ as
    \begin{align}\label{eq:lem:anticoncentration_after_coupling_1}
    \sum_{x,z\in\bin}\Pr\sbra{X=x,Z=z}\Pr\sbra{x+|Y|\equiv z+|W|\Mod q\mid X=x,Z=z}.
    \end{align}
    For fixed $x$ and $z$, consider the distribution of $x+|Y|\bmod q$ conditioned on $X=x,Z=z$.
    Since $Z$ is independent from $(X,Y)$, it is the same as the distribution, denoted by $\Pcal_x$, of $x+|Y|\bmod q$ conditioned on $X=x$.
    Similarly define $\Qcal_z$ as the distribution of $z+|W|\bmod q$ conditioned on $Z=z$ (or equivalently, conditioned on $Z=z,X=x$).
    
    Since $(X,Y)$ has distribution $\Ucal_\gamma^t$, $\Pcal_0$ has distribution $\Dcal_0$, the distribution of $|V|\bmod q$ for $V\sim\Ucal_\gamma^{t-1}$.
    Similarly, $\Qcal_1$ has distribution $\Dcal_1$, the distribution of $1+|V|\bmod q$ for $V\sim\Ucal_\gamma^{t-1}$.
    Hence by \Cref{fct:tvdist},
    \[
        \Pr\sbra{|Y|\equiv1+|W|\Mod q\mid X=0,Z=1} \le1-\tvdist{\Pcal_0-\Qcal_1} = 1-\tvdist{\Dcal_0-\Dcal_1}.
    \]
    The same bound holds for $\Pr\sbra{1+|Y|\equiv|W|\Mod q\mid X=1,Z=0}$.
    Plugging back into \Cref{eq:lem:anticoncentration_after_coupling_1}, we have
    \begin{align*}
        \Pr\sbra{X+|Y|\equiv Z+|W|\Mod q} &\le \Pr[X=Z]+\Pr[X\neq Z]\cdot\pbra{1-\tvdist{\Dcal_0-\Dcal_1}} \\
        &= \Pr[X=Z]+\pbra{1-\Pr[X=Z]}\pbra{1-\tvdist{\Dcal_0-\Dcal_1}} \\
        &= 1 - \tvdist{\Dcal_0-\Dcal_1}\pbra{1 - \Pr[X=Z]}.
    \end{align*}
    We know $\Pr[X=Z] < 1$ by \Cref{eq:lem:anticoncentration_after_coupling_2}, so the desired result follows from showing $\tvdist{\Dcal_0-\Dcal_1} > 0$ for any choice of $q\ge 3$, as well as for $q=2$ if $\gamma \ne 1/2$.

    For this we use Fourier analysis. 
    Let $\omega_q=e^{2\pi i/q}$ be the primitive $q$-th root of unity.
    We consider the following quantity
    $$
    Q=\abs{\E_{X\sim\Dcal_0}\sbra{\omega_q^X}-\E_{X\sim\Dcal_1}\sbra{\omega_q^X}}.
    $$
    On the one hand, we have
    \begin{equation}\label{eq:clm:tvdist_gamma_biased_shift_1}
    Q\le\sum_{c\in\Zbb/q\Zbb}\abs{\omega_q^c\cdot\pbra{\Dcal_0(c)-\Dcal_1(c)}}=\sum_{c\in\Zbb/q\Zbb}\abs{\Dcal_0(c)-\Dcal_1(c)}=2\cdot\tvdist{\Dcal_0-\Dcal_1}.
    \end{equation}
    On the other hand, we have 
    \begin{align}
    Q
    &=\abs{\pbra{1-\gamma+\gamma\cdot\omega_q}^{t-1}-\omega_q\cdot\pbra{1-\gamma+\gamma\cdot\omega_q}^{t-1}}
    \tag{by the definition of $\Dcal_0$ and $\Dcal_1$}\\
    &=\abs{1-\omega_q}\cdot\abs{1-\gamma+\gamma\cdot\omega_q}^{t-1}.
    \label{eq:clm:tvdist_gamma_biased_shift_2}
    \end{align}
    Let $r=\sin^2\pbra{\frac\pi q}$.
    Then
    $$
    \abs{1-\omega_q}=\sqrt{\pbra{1-\cos\pbra{\frac{2\pi}q}}^2+\sin^2\pbra{\frac{2\pi}q}}
    =2\cdot\abs{\sin\pbra{\frac\pi q}}
    =2\sqrt r
    $$
    and
    \begin{align*}
    \abs{1-\gamma+\gamma\cdot\omega_q}
    &=\sqrt{\pbra{1-\gamma+\gamma\cdot\cos\pbra{\frac{2\pi}q}}^2+\gamma^2\cdot\sin^2\pbra{\frac{2\pi}q}}
    \notag\\
    &=\sqrt{1-4\gamma(1-\gamma)\cdot\sin^2\pbra{\frac\pi q}}
    =\sqrt{1-4\gamma(1-\gamma)r}.
    \end{align*}
    Combining these with \Cref{eq:clm:tvdist_gamma_biased_shift_1} and \Cref{eq:clm:tvdist_gamma_biased_shift_2}, we have
    \begin{equation*}
    \tvdist{\Dcal_0-\Dcal_1}
    \ge\sqrt{r\cdot\pbra{1-4\gamma(1-\gamma)r}^{t-1}},
    \end{equation*}
    which is strictly larger than 0 unless $\gamma = 1/2$ and $q = 2$.
    This completes the proof.
\end{proof}

\subsection{Missing Proofs in Subsection \ref{subsec:put_together}}

\begin{proof}[Proof of \Cref{clm:gamma_restricted_TVD}]
    For clarity, we define/recall the following notation:
    \begin{itemize}
        \item $\Ecal^*(x)$ is the event that $|x| = \gamma n \pm n^{2/3}$,

        \item $p_\gamma = \Pr_\rho\sbra{\gamma_\rho = \gamma}$,

        \item $\Fcal_\gamma = \E_\rho \sbra{f(\Ucal^{[m]\setminus S}, \rho) \mid \gamma_\rho = \gamma}$,

        \item $\Gcal_\gamma$ is $f(\Ucal^m)$ conditioned on $\Ecal^*$,

        \item $\Dcal_\gamma$ is $\Dcal$ conditioned on $\Ecal^*$.
    \end{itemize}
    We will individually bound $\tvdist{\Fcal_\gamma - \Gcal_\gamma}$ and $\tvdist{\Gcal_\gamma - \Dcal_\gamma}$, and obtain our claim via the triangle inequality.
    Throughout the following, let $\Ecal$ be an arbitrary event.

    We first compare $\Fcal_\gamma$ and $\Gcal_\gamma$.
    By definition,
    \begin{align}
        \Pr_{x\sim \Gcal_\gamma}\sbra{\Ecal(x)} &= \frac{\Pr_{x\sim f(\Ucal^m)}\sbra{\Ecal(x) \land \Ecal^*(x)}}{\Pr_{x\sim f(\Ucal^m)}\sbra{\Ecal^*(x)}} \notag = \frac{\sum_{\alpha} p_\alpha \Pr_{x\sim \Fcal_\alpha}\sbra{\Ecal(x) \land \Ecal^*(x)}}{\sum_{\alpha} p_\alpha \Pr_{x\sim \Fcal_\alpha}\sbra{\Ecal^*(x)}} \notag \\
        &= \frac{p_\gamma\Pr_{x\sim \Fcal_\gamma}\sbra{\Ecal(x) \land \Ecal^*(x)} + \sum_{\alpha \ne \gamma} p_\alpha \Pr_{x\sim \Fcal_\alpha}\sbra{\Ecal(x) \land \Ecal^*(x)}}{p_\gamma\Pr_{x\sim \Fcal_\gamma}\sbra{\Ecal^*(x)} + \sum_{\alpha \ne \gamma} p_\alpha \Pr_{x\sim \Fcal_\alpha}\sbra{\Ecal^*(x)}}. \label{eq:prob_Estar_in_G}
    \end{align}
    We will separately bound each of the four terms appearing in \Cref{eq:prob_Estar_in_G}.
    Note that if $\alpha \ne \gamma$, then $|\alpha - \gamma| \ge 2^{-d}$.
    Thus, if $|x|$ is close to $\gamma n$, it must be reasonably far from $\alpha n$.
    That is,
    \begin{align*}
        \Pr_{x\sim \Fcal_\alpha}\sbra{\Ecal^*(x)} &=  \Pr_{x\sim \Fcal_\alpha}\sbra{\big||x|-\gamma n\big| \le n^{2/3}} \notag \\
        &\le \Pr_{x\sim \Fcal_\alpha}\sbra{\big||x|-\alpha n\big| \ge \big|\alpha n - \gamma n\big| - n^{2/3}} \notag \\
        &\le \Pr_{x\sim \Fcal_\alpha}\sbra{\big||x|-\alpha n\big| \ge 2^{-d}n - n^{2/3}} \tag{since $|\alpha - \gamma| \ge 2^{-d}$} \\
        &\le \Pr_{x\sim \Fcal_\alpha}\sbra{\big||x|-\alpha n\big| \ge 2^{-d-1}n} \tag{since $n \gg d$}.
    \end{align*}
    Recall for each restricted function $f_\rho(\Ucal^{[m]\setminus S})$ there exists some subset $T_\rho \subseteq [n]$ of size $|T_\rho| \le O_{d,k}(1)$ such that every $k$-tuple of bits in $[n] \setminus T_\rho$ sampled from $f_\rho(\Ucal^{[m]\setminus S})$ has distribution $\Ucal_{\gamma_\rho}^k$ (i.e., is $k$-wise independent for some even $k$).
    Define $T$ to be the union over all $T_\rho$ in the mixture $\Fcal_\alpha$, and note that $|T| \le 2^{|S|}\cdot O_{d,k}(1) = O_{d,k}(1)$.
    Let $\bar{x}$ denote the restriction of $x$ to the bits in $[n] \setminus T$.
    Then we may continue the previous chain of inequalities by
    \begin{align}
        \Pr_{x\sim \Fcal_\alpha}\sbra{\Ecal^*(x)} &\le \Pr_{x\sim \Fcal_\alpha}\sbra{\big||\bar{x}|-\alpha n\big| \ge 2^{-d-1}n - |T|} \notag \\
        &\le \Pr_{x\sim \Fcal_\alpha}\sbra{\big||\bar{x}|-\alpha (n-|T|)\big| \ge 2^{-d-1}n - (1+\alpha)|T|} \notag \\
        &\le \Pr_{x\sim \Fcal_\alpha}\sbra{\big||\bar{x}|-\E[\bar{x}]\big| \ge 2^{-d-2}n} \tag{since $n \gg d,k,\ell$} \\
        &\le 2\pbra{\frac{nk}{\pbra{2^{-d-2}n}^2}}^{k/2} \tag{by \Cref{fct:k-moments}} \\
        &\le\pbra{\frac{2^{2d+5}\cdot k}{n}}^{k/2}. \label{eq:Estar_upper_F}
    \end{align}

    Now consider $p_\gamma\Pr_{x\sim \Fcal_\gamma}\sbra{\Ecal^*(x)}$.
    We know by assumption that $p_\gamma > 0$, so there must exist some setting $\rho$ of the bits in $S$ such that $\gamma_\rho = \gamma$.
    Hence, we in fact have the stronger lower bound
    \begin{equation}\label{eq:p_gamma_LB}
        p_\gamma \ge 2^{-|S|}.
    \end{equation}
    For the remaining factor, we find
    \begin{align}
        \Pr_{x\sim \Fcal_\gamma}\sbra{\Ecal^*(x)} &= \Pr_{x\sim \Fcal_\gamma}\sbra{\big||x|-\gamma n\big| \le n^{2/3}} \notag \\
        &\ge 1 - \Pr_{x\sim \Fcal_\gamma}\sbra{\big||\bar{x}|-\gamma(n-|T|)\big| > n^{2/3} - (1+\gamma)|T|} \notag \\
        &\ge 1 - \Pr_{x\sim \Fcal_\gamma}\sbra{\big||\bar{x}|-\E[\bar{x}]\big| > \frac{n^{2/3}}{2}} \tag{since $n \gg d,k$} \\
        &\ge 1 - 2\pbra{\frac{nk}{(n^{2/3}/2)^2}}^{k/2} \tag{by \Cref{fct:k-moments}} \\
        &\ge 1 - \pbra{\frac{8k}{n^{1/3}}}^{k/2}. \label{eq:LB_on_Estar_F}
    \end{align}

    Finally, we consider $p_\gamma\Pr_{x\sim \Fcal_\gamma}\sbra{\Ecal(x) \land \Ecal^*(x)}$.
    We can again use \Cref{eq:p_gamma_LB} to lower bound $p_\gamma$.
    Additionally, 
    \begin{align}
        \Pr_{x\sim \Fcal_\gamma}\sbra{\Ecal(x) \land \Ecal^*(x)} &= \Pr_{x\sim \Fcal_\gamma}\sbra{\Ecal(x)} - \Pr_{x\sim \Fcal_\gamma}\sbra{\Ecal(x)\land \neg\Ecal^*(x)} \notag \\
        &\ge \Pr_{x\sim \Fcal_\gamma}\sbra{\Ecal(x)} - \pbra{1 - \Pr_{x\sim \Fcal_\gamma}\sbra{\Ecal^*(x)}} \notag \\
        &\ge \Pr_{x\sim \Fcal_\gamma}\sbra{\Ecal(x)} - \pbra{\frac{8k}{n^{1/3}}}^{k/2}, \label{eq:E_and_Estar_LB_F}
    \end{align}
    where the final inequality uses \Cref{eq:LB_on_Estar_F}.
    Substituting \Cref{eq:Estar_upper_F}, \Cref{eq:p_gamma_LB}, \Cref{eq:LB_on_Estar_F}, and \Cref{eq:E_and_Estar_LB_F} into \Cref{eq:prob_Estar_in_G}, we find that
    \begin{align*}
        \Pr_{x\sim \Gcal_\gamma}\sbra{\Ecal(x)} &= \frac{p_\gamma\Pr_{x\sim \Fcal_\gamma}\sbra{\Ecal(x) \land \Ecal^*(x)} + \sum_{\alpha \ne \gamma} p_\alpha \Pr_{x\sim \Fcal_\alpha}\sbra{\Ecal(x) \land \Ecal^*(x)}}{p_\gamma\Pr_{x\sim \Fcal_\gamma}\sbra{\Ecal^*(x)} + \sum_{\alpha \ne \gamma} p_\alpha \Pr_{x\sim \Fcal_\alpha}\sbra{\Ecal^*(x)}} \\
        &\le \frac{\Pr_{x\sim \Fcal_\gamma}\sbra{\Ecal(x)} + \frac{1}{p_\gamma}\max_{\alpha \ne \gamma} \Pr_{x\sim \Fcal_\alpha}\sbra{\Ecal^*(x)}}{\Pr_{x\sim \Fcal_\gamma}\sbra{\Ecal^*(x)}} \\
        &\le \frac{\Pr_{x\sim \Fcal_\gamma}\sbra{\Ecal(x)} + 2^{|S|}\pbra{\frac{2^{2d+5}\cdot k}{n}}^{k/2}}{1 - \pbra{\frac{8k}{n^{1/3}}}^{k/2}}.
    \end{align*}
    Rearranging gives
    \begin{equation}\label{eq:G_minus_F}
        \Pr_{x\sim \Gcal_\gamma}\sbra{\Ecal(x)} - \Pr_{x\sim \Fcal_\gamma}\sbra{\Ecal(x)} \le \pbra{\frac{8k}{n^{1/3}}}^{k/2} + 2^{|S|}\pbra{\frac{2^{2d+5}\cdot k}{n}}^{k/2} \le n^{-k/10},
    \end{equation}
    since $n$ is sufficiently large in terms of $d,k,$ and $\eps$.
    Similarly, we find that
    \begin{align*}
        \Pr_{x\sim \Gcal_\gamma}\sbra{\Ecal(x)} &= \frac{p_\gamma\Pr_{x\sim \Fcal_\gamma}\sbra{\Ecal(x) \land \Ecal^*(x)} + \sum_{\alpha \ne \gamma} p_\alpha \Pr_{x\sim \Fcal_\alpha}\sbra{\Ecal(x) \land \Ecal^*(x)}}{p_\gamma\Pr_{x\sim \Fcal_\gamma}\sbra{\Ecal^*(x)} + \sum_{\alpha \ne \gamma} p_\alpha \Pr_{x\sim \Fcal_\alpha}\sbra{\Ecal^*(x)}} \\ 
        &\ge \frac{\Pr_{x\sim \Fcal_\gamma}\sbra{\Ecal(x) \land \Ecal^*(x)}}{1 + \frac{1}{p_\gamma}\max_{\alpha \ne \gamma} \Pr_{x\sim \Fcal_\alpha}\sbra{\Ecal^*(x)}} \\
        &\ge \frac{\Pr_{x\sim \Fcal_\gamma}\sbra{\Ecal(x)} - \pbra{\frac{8k}{n^{1/3}}}^{k/2}}{1 + 2^{|S|}\pbra{\frac{2^{2d+5}\cdot k}{n}}^{k/2}},
    \end{align*}
    or equivalently
    \begin{equation}\label{eq:F_minus_G}
        \Pr_{x\sim \Fcal_\gamma}\sbra{\Ecal(x)} -  \Pr_{x\sim \Gcal_\gamma}\sbra{\Ecal(x)} \le \pbra{\frac{8k}{n^{1/3}}}^{k/2} + 2^{|S|}\pbra{\frac{2^{2d+5}\cdot k}{n}}^{k/2} \le n^{-k/10}.
    \end{equation}
    Combining \Cref{eq:G_minus_F} and \Cref{eq:F_minus_G} yields
    \begin{equation}\label{eq:TVD_F_G}
        \tvdist{\Fcal_\gamma - \Gcal_\gamma} \le n^{-k/10}.
    \end{equation}

    We now compare $\Gcal_\gamma$ and $\Dcal_\gamma$.
    For clarity, define
    \[
         \frac{A}{B} \coloneqq \frac{\Pr_{x \sim f(\Ucal^m)}\sbra{\Ecal(x) \land \Ecal^*(x)}}{\Pr_{x \sim f(\Ucal^m)}\sbra{\Ecal^*(x)}} = \Pr_{x\sim \Gcal_\gamma}\sbra{\Ecal(x)}
    \]
    and 
    \[
        \frac{A'}{B'} \coloneqq \frac{\Pr_{y \sim \Dcal}\sbra{\Ecal(y) \land \Ecal^*(y)}}{\Pr_{y \sim \Dcal}\sbra{\Ecal^*(y)}} = \Pr_{x\sim \Dcal_\gamma}\sbra{\Ecal(x)}.
    \]
    By our initial assumption, we know
    \begin{equation}\label{eq:A_vs_Aprime_and_B}
        \max\cbra{|A - A'|, |B - B'|} \le \tvdist{f(\Ucal^m) - \Dcal} \le \eps.
    \end{equation}
    Thus,
    \begin{align}
        \left|\Pr_{x\sim \Gcal_\gamma}\sbra{\Ecal(x)} - \Pr_{x\sim \Dcal_\gamma}\sbra{\Ecal(x)}\right| &= \frac{\left|AB' - A'B\right|}{BB'} \le \frac{A+B}{B(B-\eps)}\cdot \eps \tag{by \Cref{eq:A_vs_Aprime_and_B}} \\
        &\le \frac{2B}{B(B-\eps)}\cdot \eps \notag \\
        &\le 2\eps \cdot \pbra{\frac{1}{p_\gamma \Pr_{x \sim \Fcal_\gamma}\sbra{\Ecal^*(x)} - \eps}} \notag \\
        &\le 2\eps \cdot \pbra{\frac{1}{2^{-|S|} \pbra{1 - \pbra{\frac{8k}{n^{1/3}}}^{k/2}} - \eps}} \tag{by \Cref{eq:p_gamma_LB} \& \Cref{eq:LB_on_Estar_F}} \\
        &\le 2\eps \cdot \pbra{\frac{1}{2^{-2|S|} - \eps}}. \tag{since $n$ large in terms of $k,\eps$}
    \end{align}
    Recall that $|S| \le dk$, where $k \le \log(1/\eps)/C_d$ for some sufficiently large constant $C_d > 0$ depending only on $d$.
    Hence,
    \begin{equation}
        \left|\Pr_{x\sim \Gcal_\gamma}\sbra{\Ecal(x)} - \Pr_{x\sim \Dcal_\gamma}\sbra{\Ecal(x)}\right| \le O_d(\sqrt{\eps}). \label{eq:TVD_G_D}
    \end{equation}
    Combining \Cref{eq:TVD_F_G} and \Cref{eq:TVD_G_D}, we conclude
    \[
        \tvdist{\Fcal_\gamma - \Dcal_\gamma} \le \tvdist{\Fcal_\gamma - \Gcal_\gamma} + \tvdist{\Gcal_\gamma - \Dcal_\gamma} \le n^{-k/10} + O_d(\sqrt{\eps}) \le O_d(\sqrt{\eps})
    \]
    for large enough $n$.
\end{proof}

\end{document}